\pdfoutput=1
\documentclass[manuscript,screen]{acmart}

\AtBeginDocument{%
  }

\setcopyright{cc}
\setcctype{by}
\acmJournal{TQC}
\acmYear{2026}
\acmVolume{0}
\acmNumber{0}
\acmArticle{0}
\acmMonth{0}
\acmDOI{}

\usepackage{booktabs}
\usepackage{graphicx}
\usepackage{array}
\newcolumntype{L}[1]{>{\raggedright\arraybackslash}p{#1}}
\pdftrailerid{}
\usepackage{amsmath}

\usepackage{amssymb}
\usepackage{tikz}
\usetikzlibrary{arrows.meta,positioning,fit,backgrounds,calc}

\definecolor{qBlue}{HTML}{005A8D}     
\definecolor{qTeal}{HTML}{009E73}     
\definecolor{qOrange}{HTML}{E69F00}   
\definecolor{qYellow}{HTML}{F0E442}   

\AtEndPreamble{%
  \theoremstyle{acmplain}%
  \newtheorem{observation}[theorem]{Observation}%
  \theoremstyle{acmdefinition}%
  \newtheorem{remark}[theorem]{Remark}%
  \theoremstyle{acmplain}}

\newcommand{\graph}{G}                      
\newcommand{\syn}{\sigma}                   
\newcommand{\corr}{c}                       
\newcommand{\obs}{o}                        
\newcommand{\wt}[1]{W(#1)}                  
\newcommand{\bd}{\partial}                  
\newcommand{\relab}{\varphi}                
\newcommand{\yes}{\textsc{yes}}
\newcommand{\no}{\textsc{no}}
\newcommand{\silent}{\textsc{silent}}
\newcommand{\unassessable}{\textsc{unassessable}}

\begin{document}

\title{Capability-Gated Conformance Testing of Quantum Error-Correction Decoder Libraries}

\author{Jiachen Shen}
\orcid{0000-0002-7233-497X}
\affiliation{%
  \institution{University of Houston}
  \city{Houston}
  \state{Texas}
  \country{USA}}
\email{jshen28@cougarnet.uh.edu}

\author{Hui Zhong}
\authornote{Corresponding author.}
\orcid{0009-0005-7952-007X}
\affiliation{%
  \institution{Miami University}
  \city{Oxford}
  \state{Ohio}
  \country{USA}}
\email{zhongh7@miamioh.edu}

\begin{abstract}
A quantum error correction decoder is a library other people's results depend on, judged in one
dominant way. Sample errors, decode, and count wrong logical observables. We ask what else
can be checked there. Our conformance contract needs no oracle. One check asks that a returned
correction explain the syndrome in the caller's index space. The other hands a decoder one instance
under two presentations differing only in bookkeeping, where two feasible corrections of different
weight prove the heavier is not minimum-weight. Verdicts are gated on what each library declares, so
a firing contradicts a published guarantee. Nine configurations from five public libraries give
three results. Documentation answers $4$ of $54$ capability questions. Bounded-distance correctness,
the property callers most depend on, has a direct declaration yield of $0.0\%$, though its hypotheses
hold in $62.1\%$ of cases. Presentation sensitivity is real but shallow. One solver moved to a $26\%$ heavier
correction under a different numbering, which reaches the logical class at most once in twenty
thousand shots. Established evaluation misses corruptions that preserve logical parity, while one
summation over the caller's weights catches every one we injected. All $639$ certificates ship as
bundles a standalone verifier re-derives from first principles.
\end{abstract}

\ccsdesc[500]{Theory of computation~Quantum computation theory}
\ccsdesc[500]{Software and its engineering~Software testing and debugging}
\ccsdesc[300]{Software and its engineering~Software verification and validation}

\keywords{quantum error correction, decoder, conformance testing, metamorphic testing,
minimum-weight perfect matching, capability declaration, software measurement}

\maketitle

\section{Introduction}
\label{sec:intro}

A fault-tolerant quantum computer runs a decoder in its inner loop. Every round of syndrome extraction
hands a classical program a parity pattern, and that program returns a correction the machine then
applies. The decoder runs inside the experiment, not after it. It ships as a library, and other
people's results depend on it being right.

The field has one established way of asking whether a decoder is right. Sample errors, decode them,
and count how often the logical observable ends up wrong. That measurement is the basis of threshold
estimates \cite{wang2011surface,stephens2014fault}, of resource projections
\cite{gidney2021factor,beverland2022assessing}, and of the experiments that define progress
\cite{google2023suppressing,google2025belowthreshold}. It is a good measurement of what it measures.
What it measures is a decoder's predicted logical outcome, averaged over a noise distribution at
whatever fixed numbering the model happened to be built with.

Two things fall outside it. The first is the correction itself. A caller applies it, and an error rate
reads it only through the logical parity it induces, which leaves its feasibility, its edge content
within a logical class and its weight unread. The second is the choice of presentation:
which order the error mechanisms are listed in, and which integers the detectors are numbered with.
Neither is physics. Both are decisions made by whoever assembled the model, and a benchmark run at one
fixed presentation holds them constant instead of varying them.

This paper asks what can be checked at that interface without knowing the right answer, measures what
a checker built that way can reach across an audited panel of public decoder libraries, and
reports what it found.

\subsection{A decoder can be caught contradicting itself}
\label{sec:intro:certificate}

Give a decoder one instance under two presentations that differ only in bookkeeping, transport both
answers into the caller's index space, and check that each explains the syndrome it was given. If both
do and their total weights differ, the heavier answer is not a minimum-weight explanation: the lighter
one is a feasible competitor, and that is a proof. No oracle solves the instance, no ground truth is
consulted, and no assumption of a unique optimum is needed.

The principle is not new. Paxian and Biere refute a MaxSAT solver's claimed optimum by exhibiting a
better feasible solution found by a different solver \cite{paxian2023maxsat}. What is different here
is that the competing solution comes from the \emph{same} decoder, under a relabelling of its input,
so one implementation suffices; and that what a firing means is decided by what the library declares,
not by the arithmetic. Most decoder libraries are approximate on purpose, and catching an approximate
algorithm being approximate is not a finding.

That last point is why this paper's checker carries a conformance layer at all. A capability the
documentation claims can be tested and failed. A capability it declines is not tested. A capability it
never mentions can be measured on a concrete instance but cannot be failed. The vocabulary is
conformance testing's \cite{iso9646part1,iso9646part7}, and the same mechanism runs today in
cryptographic algorithm validation \cite{acvp}.

\subsection{What the audited panel turned out to look like}
\label{sec:intro:finding}

We applied this to nine decoder configurations from five public libraries, over a population of
instances small enough to enumerate exhaustively, so the counts taken over it carry no sampling error.

The headline is a fact about documentation, not about code. We audited what each configuration's
shipped documentation says about each of six capability dimensions, $54$ cells in all, and four of
them carry a claim. Bounded-distance correctness --- the property a caller most directly depends on
--- is declared by nobody. It is also declined by nobody. Every subject is silent on it, so its direct declaration
yield is $0.0\%$ over 7119 subject-by-case cells, while its hypotheses hold on $62.1\%$ of the cases.
The arithmetic is available and the claim to check it against is not. Nor is that zero an artifact of
reading declarations strictly: allowing a proved implication to carry a declared capability to it
lifts the yield to $13.8\%$ and stops there, because seven of the nine subjects declare nothing that
any implication can start from.

The rest of the measurements sit under that one. Every returned correction explained its syndrome in
the caller's index space, across $717{,}244$ judged results, which is a null we report as a null.
Presentation dependence appeared only among the approximate solvers on these instances. One of them,
mwpf's union-find configuration, returned the unique optimum under one detector numbering and a
$26\%$ heavier correction under another, on $26$ of $2867$ unique-optimum syndromes. The effect
survives replication over twenty independent relabellings, and a post hoc paired sign test gives a
directional tendency among the three configurations of that library. Measured instead on syndromes a
circuit-level experiment would really produce, the same relabelling changes a belief-propagation
decoder's correction on two to seven shots per thousand and its logical class at most once in twenty
thousand, so the effect lives at the interface and a logical error rate is nearly blind to it. The two
libraries that do declare exactness were never contradicted, on the enumerated population or in a
pre-registered search of $160{,}000$ paired comparisons on circuit-derived graphs an order of
magnitude larger. The certificate establishes non-minimality whenever two presentation-equivalent
feasible returns carry unequal weight, which is what that null covers. Exactness itself is assessed
separately, against the optimum: we attempted an exact solve of all $1600$ canonical syndromes of
that search. The oracle settled $1599$, and on every settled one both libraries returned exactly the
minimum weight (Section~\ref{sec:hunt:optimum}).

\subsection{Contributions}
\label{sec:intro:contributions}

\begin{enumerate}
\item \textbf{A capability vocabulary for decoder libraries, with its central pair proved
incomparable.} Three
capability levels and three flags, three declaration states, and a dispatch policy that keeps
\textsc{unassessable} apart from a legal skip. We prove that bounded-distance correctness and exact
minimum-weight optimisation are incomparable, in both directions, with explicit constructions
(Theorem~\ref{thm:incomparable}), so on the general weight domain a library declaring either has not
implied the other. We also state the restricted domain on which one direction does hold, and separate
what it licenses --- binding a subject to a capability on a case the evaluator has checked --- from
what it never licenses, which is inferring a declaration the library did not make
(Section~\ref{sec:model}).

\item \textbf{An oracle-free suboptimality certificate, proved, implemented and independently
verifiable.} Two feasible corrections of different total weight prove the heavier one non-minimal
(Lemma~\ref{lem:certificate}). Over $76{,}680$ paired comparisons the implementation produced $639$
certificates and contradicted an exact oracle on none of them. All $639$ are exported as
self-contained bundles and re-checked by a program that imports nothing from this project:
$639$ verified, $0$ failed (Sections~\ref{sec:certificates} and~\ref{sec:artifact}).

\item \textbf{A measurement of what a capability-gated checker can reach across the audited panel.}
Bounded-distance correctness is \textsc{silent} in $7$ of $7$ declarations and \textsc{no} in $0$ of
$7$, giving a direct declaration yield of $0.0\%$ and a derived yield of $13.8\%$ against an
applicability of $62.1\%$. No returned correction failed to explain its syndrome across $717{,}244$ judged results.
Presentation dependence moves an approximate decoder between optimal and suboptimal. Our own
$\mathit{int\text{-}weights}$ flag predicts the wrong subjects. This is a defect of our capability
vocabulary, not of any library (Section~\ref{sec:measurement}).

\item \textbf{A map of which part of a decoder's answer each evaluation method inspects.} Over $381$
distinct behaviours, error-rate benchmarking read on the observable vector a library returns, and
exhaustive observable checking, detect $0$ of the $219$ that corrupt the correction. The same
benchmark read on an observable the caller recomputes from the returned correction, which is what
seven of the nine subjects require, detects $176$ of them at the pre-registered $\alpha$, $173$
family-wise, and is blind to the $43$ that change the correction without moving its logical class. Exhaustive correction checking detects all $219$ and our
contract $215$. In the other direction our contract detects $0$ of $23$ consistent observable
truncations that the observable-based methods all catch, and neither they nor our contract detect any
of the $72$ that perturb the returned weight, while one summation over the caller's own
weights detects all $72$. The separations follow from the field each method reads, not from detection power. Neither dominates.
The correction check differs from our contract because it needs an exact oracle. On the populations measured here that oracle is affordable and the difference is a factor of
four to six, so the reason to prefer the contract is that it carries no solver dependency and emits a
checkable witness, not that the alternative is out of reach. A pre-registered falsifier fired in that
study and the numbers are from the repaired arm (Section~\ref{sec:accessmap}).
\end{enumerate}

Sections~\ref{sec:hunt} and~\ref{sec:cost} support the third and fourth: a pre-registered search for
a conformance violation against the two exactness declarations, which returned none, and a cost
measurement showing where an oracle-free check is worth having.

\subsection{How to read this paper}
\label{sec:intro:roadmap}

Section~\ref{sec:background} fixes the setting and introduces a running example that
Section~\ref{sec:certificates} certifies. Section~\ref{sec:overview} gives the whole construction
informally in two pages. Sections~\ref{sec:model} to~\ref{sec:relations} are the model, the properties
and the population. Sections~\ref{sec:measurement} to~\ref{sec:cost} are the measurements. They appear in decreasing
order of how much each says about the audited panel, not our instrument.
Section~\ref{sec:artifact} is the artifact and Section~\ref{sec:related} the related work. A reader
with fifteen minutes should read Sections~\ref{sec:overview}, \ref{sec:measurement}
and~\ref{sec:accessmap}.

\section{Decoding, and the interface a caller actually sees}
\label{sec:background}

This section fixes the setting and introduces the example the rest of the paper returns to. A reader
who knows quantum error correction can read it for the interface vocabulary alone.

\subsection{Decoding is a weighted combinatorial problem}
\label{sec:bg:problem}

A quantum error-correcting code is measured repeatedly. Each round yields parity information, not the
error itself. That information is organised as a detector error model, a list of
independent error \emph{mechanisms}, each with a probability, a set of \emph{detectors} it flips and
a set of logical \emph{observables} it flips \cite{gidney2021stim}. A run produces a \emph{syndrome},
the set of detectors that fired.

The decoder's job is to propose a correction, a subset of the mechanisms whose combined detector
pattern is exactly the syndrome. Since each mechanism $i$ has probability $p_i$ and the mechanisms are
independent, the probability of a given subset is a fixed constant times
$\exp\!\left(-\sum_i \log\frac{1 - p_i}{p_i}\right)$ over its members. That is decreasing in the
sum, so the most likely subset is the one that minimises it.
Writing $w_i$ for the log-likelihood-ratio weight of mechanism $i$ turns decoding into a
minimum-weight problem over a hypergraph. When every mechanism flips at most two detectors the
hypergraph is a graph, the virtual boundary $b$ absorbs the mechanisms that flip only one, and the
problem is a minimum-weight $T$-join, with $T = \syn$ when $|\syn|$ is even and
$T = \syn \cup \{b\}$ when it is odd \cite{dennis2002topological,fowler2012surface}. Solvers reach it as minimum-weight perfect matching
after pairing the fired detectors and the boundary in a derived complete graph, which is Edmonds'
problem \cite{edmonds1965paths} and is what the modern implementations solve
\cite{kolmogorov2009blossomv,higgott2025sparseblossom,fusionblossom}.

That reduction is why a decoder library has an interface worth testing. A caller compiles a circuit,
extracts a model, hands a graph and a syndrome to a library, and receives an answer in the indices it
supplied --- an exchange whose every step is bookkeeping the physics does not determine.

\subsection{Three things a decoder can return}
\label{sec:bg:returns}

Libraries differ in what they hand back, and the difference is not cosmetic.

Some return a \emph{correction}, the subset of mechanisms, as indices into the list the caller
supplied. That subset is an explanation in the model's own terms, not a physical operation. The caller computes
which logical observables the subset flips. The caller then maps the subset to a recovery or a frame
update on its own data.

Some return an \emph{observable vector} directly, one bit per declared logical observable, computed
inside the library. A caller that only needs to know whether its logical qubit flipped can stop there.

Some return a \emph{weight} alongside either of the above, the total $\sum_i w_i$ of the correction
the library chose. It is the quantity a caller would use to compare two candidate answers, or to decide
how much to trust one.

Figure~\ref{fig:interface} draws the interface and marks where each evaluation method looks. Which of
the three components each method inspects is what Section~\ref{sec:accessmap} measures.

\begin{figure}[t]
\centering
\begin{tikzpicture}[
  font=\small,
  core/.style={draw=black!70, fill=black!3, rounded corners=2pt, align=center,
               text width=25mm, inner xsep=1.5mm, inner ysep=1.5mm, minimum height=9mm},
  returned/.style={draw, rounded corners=2pt, align=center, text width=26mm,
                   inner xsep=1.5mm, inner ysep=1.3mm, minimum height=9mm, line width=.7pt},
  method/.style={draw, dashed, rounded corners=2pt, align=center, text width=37mm,
                 inner xsep=1.5mm, inner ysep=1.2mm, minimum height=11mm, font=\footnotesize},
  flow/.style={-{Latex[length=2mm]}, draw=black!80, line width=.65pt}]

\node[core] (caller) {caller\\[-.2ex]{\footnotesize model $M$, syndrome $\syn$}};
\node[core, text width=22mm, right=5mm of caller] (lib) {decoder library};

\node[returned, draw=qBlue, fill=qBlue!8, right=8mm of lib] (obsv) {observable vector $\obs$};
\node[returned, draw=qTeal, fill=qTeal!9, above=4mm of obsv] (corr) {correction $\corr$};
\node[returned, draw=black, fill=qYellow, line width=.9pt, below=4mm of obsv] (wgt)
  {weight $\wt{\corr}$};

\node[method, draw=qTeal, right=4mm of corr] (corrread)
  {read by: correction check,\\this paper's contract};
\node[method, draw=qBlue, right=4mm of obsv] (obsread)
  {read by: logical error rate,\\observable check, width check};
\node[method, draw=black, fill=qYellow, densely dashed, line width=.9pt, right=4mm of wgt] (wgtread)
  {\textbf{read by: no method}\\we compared, nor by our contract};

\draw[flow] (caller.east) -- node[above, font=\footnotesize] {$M, \syn$} (lib.west);
\draw[flow, draw=qTeal] (lib.north east) -- ++(3mm,0) |- (corr.west);
\draw[flow, draw=qBlue] (lib.east) -- (obsv.west);
\draw[flow] (lib.south east) -- ++(3mm,0) |- (wgt.west);

\draw[flow, draw=qTeal] (corr.east) -- (corrread.west);
\draw[flow, draw=qBlue] (obsv.east) -- (obsread.west);
\draw[flow] (wgt.east) -- (wgtread.west);
\end{tikzpicture}
\caption{What crosses the interface, and which evaluation method inspects each returned component.
The bottom row is the result of Section~\ref{sec:accessmap}. Of the three components a decoder can
return, the weight is read by none of the established evaluation methods compared there, and not by
our own contract either. It is checkable in one summation, and that summation is the arm this paper
adds to the comparison.}
\label{fig:interface}
\end{figure}
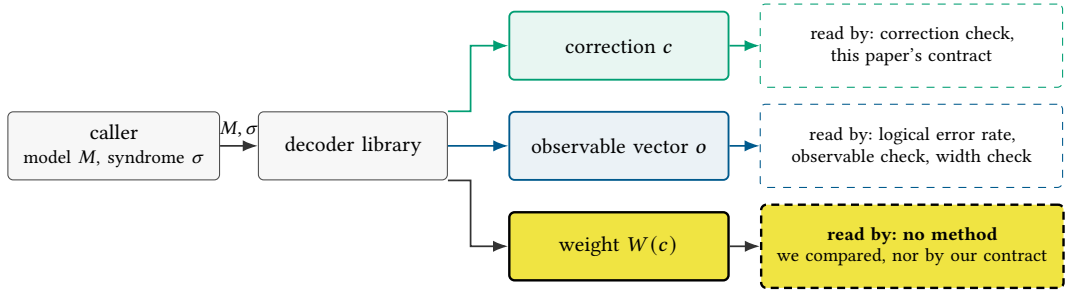

\subsection{What a logical error rate scores}
\label{sec:bg:ler}

The established measure of a decoder is its logical error rate. Sample many errors from the noise
model, decode each one, compare the decoder's predicted observable flips against the truth the
simulator knows, and report the fraction that disagree. It is the measure behind threshold estimates
\cite{wang2011surface,stephens2014fault}, resource projections \cite{gidney2021factor}, and the
hardware demonstrations that mark progress in the field
\cite{google2023suppressing,google2025belowthreshold}.

It is the right measure for the question it asks, which is how often a decoder's answer lands in the
wrong logical class over a distribution of errors. Two properties of that question matter here. It
runs at one fixed presentation, holding the mechanism ordering and detector numbering constant, so
sampling errors alone cannot report what a second presentation would have returned. And it scores the
observable prediction, so a library returning a correction is scored through arithmetic the caller
performs afterwards, and a change of correction inside one logical class does not move it at all.

\subsection{The running example}
\label{sec:bg:example}

Take the $X$-error sector of a distance-$3$ unrotated surface code memory experiment in the $Z$ basis,
under code-capacity noise at $p = 0.01$. Its detector error model has $6$ detectors and $13$
mechanisms, and because the noise is uniform every mechanism carries the same weight
$w = \log\frac{1-p}{p} = 4.595$. Weight-minimisation is therefore counting:
the best correction is the one with the fewest mechanisms.

Consider the syndrome $\syn = \{0, 1, 2, 5\}$. One correction that explains it uses the three
mechanisms $\{0,1\}$, $\{2,4\}$ and $\{4,5\}$, written by their detector endpoints, for a total weight
of $13.785$. Another uses five: $\{0,2\}$, $\{1,b\}$, $\{2,4\}$, $\{2,b\}$ and $\{4,5\}$, where $b$ is
the virtual boundary, for a total of $22.976$. Both are feasible. Detector $2$ appears in three of the
five, so it fires an odd number of times, and the reader can check the parity of every detector by
hand.

The three-mechanism answer and the five-mechanism answer came from the same decoder, on the same
instance, in the same session. The only difference between the two calls was the order in which the
$13$ mechanisms were listed. Section~\ref{sec:certificates} shows what that pair of answers proves,
without anyone solving the instance.

\section{Overview}
\label{sec:overview}

Before the definitions, here is the whole construction in plain terms, and the shape of the study
built on it.

\subsection{The check, informally}
\label{sec:ov:check}

Hand a decoder the same problem twice, described two different ways, and compare what comes back.

The two descriptions must differ only in bookkeeping, with mechanisms listed in a different order or
detectors numbered differently. Neither choice belongs to the physics. Both belong to whoever built
the model. So the two calls pose one question, and the answers can be carried into a common index
space and compared.

Two things are then worth checking, and neither needs anyone to know the right answer. Each returned
correction should explain the syndrome it was given, in the indices the caller supplied, which is a
parity computation the caller can do. And if the two answers explain the same syndrome but carry
different total weight, the heavier cannot be a minimum-weight explanation, because the lighter is a
feasible competitor. That is a proof, and the lighter answer is the evidence a reader can check.

The second check is the useful one, and its firing condition is exact. It fires when two
presentation-equivalent feasible answers carry unequal total weight, and what it then establishes is
that the heavier one is not minimum-weight. What a firing \emph{means} for a library is decided by
what that library declares.

\subsection{Why a declaration decides what a firing means}
\label{sec:ov:declaration}

Most decoder libraries are approximate on purpose. Union-find trades exactness for speed
\cite{delfosse2021almostlinear}, and belief propagation with post-processing makes no exactness
promise at all \cite{panteleev2021degenerate}. A heavier answer from one of those is the algorithm
doing what it says on the box. The same measurement against a library publishing an exactness
guarantee contradicts that guarantee.

So the verdict layer is gated on what the library declares, in the vocabulary conformance testing has
used for thirty years \cite{iso9646part1,iso9646part7} and cryptographic validation uses today
\cite{acvp}. A claimed capability can be tested and failed. A declined one is not tested. One the documentation
never mentions can be measured on an instance but cannot be failed. That third state is
where the audited panel turns out to live.

\subsection{Two disciplines that run through the study}
\label{sec:ov:study}

Section~\ref{sec:intro:roadmap} gives the order of what follows. Two disciplines run through all of
it and are worth stating before any of it.

Every confirmatory experiment was pre-registered and committed alone, before the code that acts on it
existed, and every correction made to an instrument during the work is recorded in the artifact
alongside the result it affects.

One population in this paper is enumerated completely, so the counts taken over it carry no sampling
error. Three are sampled and say so where they appear: the relabellings drawn for replication, the
shots behind an error-rate arm, and the syndromes of the search in Section~\ref{sec:hunt}.

\section{The interface and what a library declares}
\label{sec:model}

This section fixes the object we measure. We first say what crosses the boundary between a caller and
a decoder library. We then say what a library has to publish before a checker may hold it to a
guarantee, and we prove that two of those guarantees do not imply each other. The dispatch rule at the
end of the section follows from that proof.

\subsection{What crosses the interface}
\label{sec:model:interface}

A caller builds a model of its own experiment and hands that model to a decoder. In Stim the model is
a detector error model, and it is a list of error mechanisms. Each mechanism carries a probability, a
set of detectors it flips, and a set of logical observables it flips. The caller then sends syndromes
and receives answers.

Two answers are possible and they are not interchangeable. A \emph{correction} names the mechanisms
the decoder believes fired. A \emph{predicted observable vector} says only which logical observables
it believes flipped. A correction is expressed in indices, and those indices belong to the caller. The decoder was given
the caller's list, so index $i$ in the answer must mean mechanism $i$ of that list. We call this the caller's index space, and preserving it is an obligation the decoder inherits
the moment it accepts the model.

\begin{definition}[Subject]
\label{def:subject}
A \emph{subject} is a decoder exposing at least one of two operations over a detector error model
$M$ supplied by the caller. The operation $\mathrm{dec\_corr}$ maps a syndrome
$\syn \in \mathbb{F}_2^{D}$ to a set of mechanism indices in the index space of the flattened $M$.
The operation $\mathrm{dec\_obs}$ maps $\syn$ to a predicted observable vector of length $|L(M)|$.
A subject may also return the total weight of the correction it chose.
\end{definition}

Definition~\ref{def:subject} names three components an answer can carry, none of them universal: a
correction, an observable vector, and a weight. Section~\ref{sec:accessmap} measures which of the
three each established evaluation method inspects, and the answer is not all of them.

Every graphlike detector error model becomes a weighted graph in the standard way. Detectors become
vertices, a mechanism that flips two detectors becomes an edge between them, and a mechanism that
flips one detector becomes an edge to a virtual boundary vertex. Each edge carries the
log-likelihood-ratio weight $w_e = \log\!\big((1-p_e)/p_e\big)$ of its mechanism, which is the
quantity whose sum over a subset fixes that subset's probability up to a constant. Writing $\bd_D S$
for the set of detectors that an edge set $S$ flips an odd number of times, a correction $\corr$
explains a syndrome exactly when $\bd_D \corr = \syn$, and its cost is
$\wt{\corr} = \sum_{e \in \corr} w_e$.

\begin{definition}[Scope]
\label{def:scope}
The framework is defined over $M$, and the properties in this paper are evaluated on the graphlike
fragment of $M$: models in which every mechanism flips at most two detectors, carries a probability in
$(0, 1/2)$, and flips at least one detector. Any circuit whose detector error model lies in that
fragment lies inside the interface. Hyperedge mechanisms, decomposed correlated components and
logical-only mechanisms are outside it, and a model containing them is reduced before an adapter sees
it. The harness reads no topology, no gate set, and no device description.
\end{definition}

The framework performs no decoding. Its coverage quantity is the fraction of the population on which
a check can run at all, which Definition~\ref{def:yield} names.

\subsection{What a library declares}
\label{sec:model:declare}

A checker that asks a decoder for a guarantee has to know which guarantees it offers. An exact
minimum-weight solver and a fast heuristic will disagree on hard instances and only one has promised
not to, so holding both to the same standard measures the wrong thing. Conformance testing solved
this for communication protocols decades ago with an implementation conformance statement, in which
the implementer answers capability by capability before any test is selected
\cite{iso9646part1,iso9646part7}. The same idea drives cryptographic module validation today
\cite{acvp}.

We use four capabilities, and they divide the way conformance testing divides them
\cite{iso9646part1}. One is \emph{mandatory}: it is incurred by exposing the interface at all, it is
never declared, and it is always tested. The other three are \emph{optional}: they are declared or
they are not, and Definition~\ref{def:declstate} governs what a checker may do with the answer.

Capability $L_0$ is the mandatory one, and it is the correction interface's own contract. A library
that accepts the caller's model and calls its return a \emph{correction} has thereby undertaken two
things: that the return is syntactically well formed, a list of indices in range and not repeated,
and that it explains the syndrome it was given \emph{in the caller's index space}. Nothing about
optimality, distance or determinism is part of it. Every subject that returns a correction incurs $L_0$ without declaring it, because it accepts the
model. The residual property of Definition~\ref{def:pb} is therefore dispatched to every subject on
that path, and a subject exposing only an observable vector has no correction for it to constrain. A firing there is a conformance
failure, not an observation.

Testing $L_0$ is not testing something true by definition. $L_0$ is a normative obligation on the
implementation, not a description of what implementations do, and conformance testing asks exactly
whether an implementation meets such an obligation. A decoder that solves a different model perfectly
and reports in that model's indices fails it.

Capability $L_1$ is bounded-distance
correctness, and it is a conjunction whose two halves have to be checked in order. For an error within
the code's correction radius, the returned correction must explain the syndrome, \emph{and} the
residual must lie in the identity logical class. Where a subject also returns an observable vector,
the claim covers that vector: it has the width the model declares and it agrees with the returned
correction. The order is forced, not stylistic. For a stabiliser code, the residual has a logical class only when
its syndrome vanishes. A correction that does not explain the syndrome therefore has no class to be
wrong about. Appendix~\ref{app:proofs} proves the ordering. Capability
$L_2$ is exact minimum-weight optimisation. Capability $L_3$ is exact maximum-likelihood decoding.
Three flags accompany them: $\mathit{det}$ for a determinism guarantee, $\mathit{rand}$ for a declared
randomised schedule, and $\mathit{int\text{-}weights}$ for a solver whose internal weights are
integers.

\begin{definition}[Declaration state]
\label{def:declstate}
For an \emph{optional} capability $\ell$ and a subject $S$, the declaration state is $\yes$ when the
documentation of $S$ claims $\ell$, $\no$ when it declines $\ell$, and $\silent$ otherwise. A
declaration state is a fact about documentation at a stated version. It is never inferred from a
name, from behaviour, or from the absence of some other claim. A mandatory capability has no
declaration state, because it is not offered and cannot be declined.
\end{definition}

The last sentence of Definition~\ref{def:declstate} is load bearing. A capability read out of a
solver's class name, or out of the absence of an adjacent claim, is a declaration the library never
made.

\subsection{Two of the capabilities are not a chain}
\label{sec:model:incomparable}

It is tempting to order the capabilities and let a strong declaration imply a weaker one. A decoder
that always returns a minimum-weight correction sounds like it must also correct every small error.
It need not, because minimum weight and small Hamming weight are different objectives. Weights come
from error probabilities, so one unlikely mechanism can outweigh several likely ones, and the
cheapest explanation of a syndrome can flip a logical observable that a Hamming-weight-minimal
explanation would leave alone.

\begin{theorem}[$L_1$ and $L_2$ are incomparable]
\label{thm:incomparable}
$L_1 \not\Rightarrow L_2$ and $L_2 \not\Rightarrow L_1$. Both directions have explicit constructions
over legal log-likelihood-ratio weights.
\end{theorem}

For $L_2 \not\Rightarrow L_1$, take a code whose logical operator is a simple cycle of length $d$ in
the graph. Give one cycle edge weight $d$, give the remaining cycle edges weight $1$, and give every
off-cycle edge a weight larger than $d-1$. A single error on the heavy cycle edge produces a syndrome
whose exact minimum-weight explanation is the rest of the cycle, at total weight $d-1$. That
explanation completes the logical cycle, so an exact solver fails on a weight-one error. For
$L_1 \not\Rightarrow L_2$, a decoder that returns a nonidentity stabiliser on the empty syndrome
corrects every error inside the radius and returns a correction of positive weight where the optimum
has weight zero. Appendix~\ref{app:proofs} gives both constructions in full.

Theorem~\ref{thm:incomparable} decides the shape of the capability model. $L_1$ and $L_2$ do not form
a chain, so a gate on $\ell$ is opened by a declaration of $\ell$ and by nothing weaker, stronger or
adjacent. No implication is claimed among the remaining capabilities, and $L_3$ is used here only as
a declaration slot.

One implication does hold, and stating it is what keeps the dispatch rule honest. Restrict attention
to a code-capacity instance whose weights order candidate corrections by physical Hamming weight.
There the objectives of $L_1$ and $L_2$ coincide on the relevant candidates, and an exact
minimum-weight decoder corrects every error within the radius, so $L_2 \Rightarrow L_1$ on that
restricted domain. The construction above escapes this because its weights are deliberately not
Hamming-ordering.

The implication does not change what counts as a conformance verdict, and this is where a framework
of this kind usually goes wrong. An adverse verdict against $\ell$ requires the subject to have
declared $\ell$ itself. Definition~\ref{def:declstate} forbids inference from names, from behaviour
and from the absence of an adjacent claim, and an adjacent declaration is exactly such an inference.

The implication licenses a narrower conclusion about one case, not a library. A subject is bound to $L_1$ on \emph{this case} when two conditions hold. The subject declares
$L_2$. The evaluator has verified the case as code capacity with Hamming-ordering weights. The bound
comes from the theorem, not the documentation. That is a standard assume-guarantee argument, and refusing it would be as unprincipled
as inferring a declaration. Definition~\ref{def:yield} therefore reports both yields, and
Section~\ref{sec:measurement:yield} puts them side by side: the gap between them says whether a zero
came from the ecosystem or from the rule.

\subsection{What a gate does with each answer}
\label{sec:model:dispatch}

A checker meets three kinds of library. Some declare the capability, some decline it, and most say
nothing. The first two are easy. The third is where a testing framework quietly
loses its meaning, because silence looks like a passed test if the checker simply skips.

\begin{definition}[Dispatch policy]
\label{def:dispatch}
A property gated at $\ell$ returns a verdict when the declaration state of the subject at $\ell$ is
$\yes$. It returns a legal skip when the state is $\no$, following ISO/IEC 9646-7
Appendix~I.4.1.3, under which an optional capability answered NO is not tested and is not a case of
non-conformance. It returns $\unassessable$ when the state is $\silent$, following ISO/IEC 9646-1
clause 5.6.1, under which a capability statement precedes conformance assessment.
\end{definition}

\begin{figure}[t]
\centering
\begin{tikzpicture}[
  font=\fontsize{7.15}{8.6}\selectfont, node distance=3mm,
  every node/.style={inner xsep=1.4mm, inner ysep=1.15mm},
  machine/.style={draw=qBlue, fill=qBlue!7, rounded corners=2pt, align=center, line width=.7pt},
  branch/.style={rounded corners=2pt, align=center, line width=.75pt, text width=22mm},
  property/.style={draw=qOrange, fill=qOrange!11, rounded corners=2pt, align=center, line width=.8pt},
  outcome/.style={rounded corners=2pt, align=center, text width=18mm, line width=.8pt},
  gate/.style={draw=black!75, fill=black!4, align=center, line width=.8pt,
               rounded corners=2pt, text width=19mm},
  gatelabel/.style={midway, above, inner sep=.4mm},
  stem/.style={draw=black!75, line width=.6pt},
  flow/.style={-{Latex[length=1.5mm]}, draw=black!80, line width=.6pt}]

\node[machine] (circ) {circuit};
\node[machine, right=3mm of circ] (dem) {model $M$};
\node[machine, right=3mm of dem] (gen) {G: generator};
\draw[flow] (circ) -- (dem);
\draw[flow] (dem) -- (gen);

\node[branch, draw=qBlue, fill=qBlue!7, above right=1.2mm and 8mm of gen.east] (ad)
  {A: adapter\\[-.4ex]subject and declaration};
\node[branch, draw=qTeal, fill=qTeal!9, below right=1.2mm and 8mm of gen.east] (orc)
  {K: true model\\[-.4ex]independent oracle};

\coordinate (gfork) at ($(gen.east)+(4mm,0)$);
\draw[stem] (gen.east) -- (gfork);
\draw[flow] (gfork) |- (ad.west);
\draw[flow] (gfork) |- (orc.west);

\coordinate (amid) at ($(ad.east)!0.5!(orc.east)$);
\node[gate, right=4mm of amid] (gate) {gate at $\ell$\\[-.4ex]declaration state};
\coordinate (pjoin) at ($(gate.west)-(2.5mm,0)$);
\draw[stem] (ad.east) -| (pjoin);
\draw[stem] (orc.east) -| (pjoin);
\draw[flow] (pjoin) -- (gate.west);

\node[property, right=6mm of gate] (prop) {C: property};
\node[outcome, draw=qTeal, fill=qTeal!9, right=4mm of prop] (verdict) {verdict};
\draw[flow] (gate) -- node[gatelabel] {\bfseries \yes} (prop);
\draw[flow] (prop) -- (verdict);

\node[outcome, draw=qOrange, fill=qOrange!11, dashed, above=3.2mm of verdict] (skip)
  {legal skip};
\node[outcome, draw=black, fill=qYellow, densely dotted, below=3.2mm of verdict] (unk)
  {\unassessable};
\draw[flow] (gate.north) |- node[inner sep=.4mm, pos=.74, anchor=south, yshift=.5mm] {\bfseries \no} (skip.west);
\draw[flow] (gate.south) |- node[inner sep=.4mm, pos=.74, anchor=north, yshift=-.5mm] {\bfseries \silent} (unk.west);
\end{tikzpicture}
\caption{The evaluation pipeline. Three things it is meant to make visible are claims in this paper:
the whole path, not the property alone, is what is frozen; the oracle sits beside the subject rather
than inside it, so the residual and the optimum are computed against the caller's model; and the
declaration state, not the subject's behaviour, decides whether a gated property runs at all.}
\label{fig:pipeline}
\end{figure}
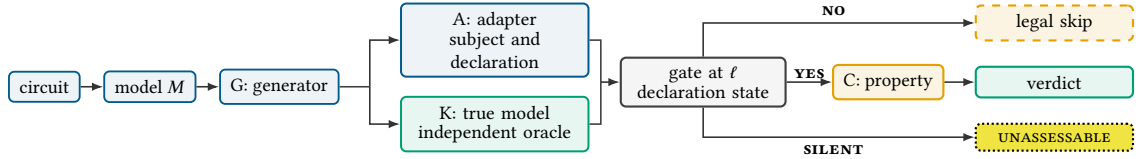

Figure~\ref{fig:pipeline} shows where the decision sits. The three outcomes are reported in three
columns and never added together, because pooling a skip with an $\unassessable$ turns a library that
publishes nothing into a library that passes everything.

\begin{definition}[Evaluation layer]
\label{def:layer}
A \emph{conformance} result is a verdict against either a declared optional capability or a
mandatory interface capability incurred by exposing the relevant operation. An
\emph{instance-measured} result is a statement about one concrete instance, established without such
a conformance obligation, whether by an independent oracle or by a sound oracle-free witness. It
holds for that instance alone. The two layers are reported separately.

Both halves have to be stated this way, because the narrower reading excludes results this paper
reports. A residual failure is a verdict and there is no declaration behind it, since $L_0$ is
mandatory. And a certificate against a subject that declares nothing is an instance-measured result
reached with no oracle at all, which is the whole point of Lemma~\ref{lem:certificate}. Where the
older term \emph{claim-based} appears in the artifact it means this conformance layer; the
\emph{direct declaration yield} keeps its name, because that quantity really does count
documentation.
\end{definition}

One number summarises what a gate lets through, and it is the number
Section~\ref{sec:measurement:yield} reports for the panel.

\begin{definition}[Dispatch yield]
\label{def:yield}
Fix a property $P$ and a population of \emph{cells}, where a cell is one subject paired with one
case. $P$ dispatches on a cell exactly when its hypotheses hold on that case \emph{and} its gate
opens for that subject, and those two conditions are independent: the hypotheses are a property of
the case and the gate is a property of the documentation. Three yields follow, and they must be
reported separately because they answer different questions.

The \emph{direct declaration yield} counts the cells where the subject's documentation declares the
gated capability itself, and declares any flag the property requires. The \emph{derived assurance
yield} also counts a cell where the gate opens on a different declared capability together with
instance assumptions the evaluator has verified and a proved implication. The \emph{applicability
yield} counts the cases where the hypotheses hold and ignores declarations entirely.

A property gated on a \emph{mandatory} capability is marked as one wherever it is reported. No
subject can decline such a gate, so its direct column is $100\%$ by construction and is not evidence
that anybody wrote anything down. Reading it as a declaration count is a category error, which is why
the residual row of the yield table in Section~\ref{sec:measurement:yield} carries a marker.
\end{definition}

Separating the three is what keeps a zero readable. A zero applicability yield means the population is
unsuitable and a different population fixes it. A high applicability yield beside a zero direct yield
means no population fixes it, and the question becomes whether a proved implication can close the
gap, which is what the derived yield measures.

The derived yield is where compositional reasoning enters, and it holds a subject to $\ell$ on a
specific case only. The subject declared $\ell'$, the evaluator verified that this case satisfies a
theorem's hypotheses, and the theorem carries $\ell'$ to $\ell$ there. Inferring a declaration
remains forbidden by Definition~\ref{def:declstate}, and Definition~\ref{def:layer} keeps the
assume-guarantee argument apart from a documentary claim. This paper uses two such implications, both
proved: $L_2 \Rightarrow L_1$ on a
code-capacity instance whose weights order candidates by Hamming weight
(Remark~\ref{app:rem:domain}), and a verified-unique optimum forcing the returned edge set
(Corollary~\ref{app:cor:det}).

The layer distinction is what lets one measurement carry two meanings. A library that declares
exactness and returns a heavier correction than another presentation of the same instance has
produced a conformance verdict. A library that declares nothing and does the same has produced an
observation. The arithmetic is identical and the standing is not.

\section{Certificates that need no oracle}
\label{sec:certificates}

Checking a decoder usually means knowing the right answer. Not here. Both checks compare a decoder's
output against arithmetic the caller can already do, and the second compares a decoder against
itself.

\subsection{The answer has to explain the syndrome}
\label{sec:cert:residual}

The weakest thing a caller can ask is that the returned correction accounts for the detectors that
fired. That is a parity computation over the model the caller supplied, needing no guarantee from the
library and no knowledge of the true error.

\begin{definition}[Residual property]
\label{def:pb}
Gate $L_0$, which is mandatory, so the gate opens on every subject. Hypothesis: the subject exposes
$\mathrm{dec\_corr}$. Claim: for the returned index set $\corr$ and the syndrome $\syn$,
$\bd_D \corr = \syn$, with the detectors of each mechanism read in the caller's index space. A violation is a conformance failure against the interface contract of
Section~\ref{sec:model:declare}, not an instance-measured observation. This follows because $L_0$ is
incurred, not declared.
\end{definition}

The phrase about the caller's index space is the whole content of the check. A decoder that solves a
different model perfectly and reports in that model's indices produces exactly this failure, so the
caller's model is the reference and a mismatch is a defect of the interface even when the internal
solve was flawless.

\subsection{Two presentations of one instance}
\label{sec:cert:relabel}

Now the second check. A decoding instance can be handed to a library in many ways that differ only in
bookkeeping. The mechanisms can be listed in any order. The detectors can be numbered in any order.
Neither choice belongs to the physics. Both belong to whoever built the model.

\begin{definition}[Presentation]
\label{def:presentation}
Let $\graph = (D \cup \{b\}, E, w, \lambda)$ be a weighted decoding graph with virtual boundary
$b$, where $\lambda(e) \in \mathbb{F}_2^{k}$ records which of the $k$ declared observables edge $e$
flips. A \emph{relabelling} $\relab$ is a pair of bijections, one on $D \cup \{b\}$ fixing $b$ and
one on $E$, that preserves incidence, preserves the weight of every edge, and preserves its observable
label: $w(\relab e) = w(e)$ and $\lambda(\relab e) = \lambda(e)$. The graph $\relab \graph$
together with the permuted syndrome $\relab \syn$ is a \emph{presentation} of the same instance.
\end{definition}

Hand one instance to a decoder under two presentations and transport the second answer back through
$\relab^{-1}$. Both answers now live in the same index space and can be compared. The comparison that
matters is not whether they agree. It is what it means when they do not.

\begin{lemma}[Weight disagreement certifies non-minimality]
\label{lem:certificate}
Let $S_1 \subseteq E$ and $S_2' \subseteq \relab E$ satisfy $\bd_D S_1 = \syn$ and
$\bd_D S_2' = \relab \syn$, and write $S_2 = \relab^{-1}(S_2')$. If $\wt{S_1} \neq \wt{S_2}$, then the
heavier of the two is not a minimum-weight explanation of $\syn$.
\end{lemma}

\begin{proof}
Since $\relab$ preserves incidence and fixes $b$, it carries the detector boundary along:
$\bd_D(\relab S) = \relab(\bd_D S)$ for every $S \subseteq E$. Applying this to $S_2$ gives
$\bd_D S_2 = \relab^{-1}(\bd_D S_2') = \relab^{-1}(\relab \syn) = \syn$, so $S_2$ explains $\syn$ in
the original presentation. Both $S_1$ and $S_2$ are therefore feasible for $\syn$. Since $\relab$
preserves every edge weight, $\wt{S_2} = \wt{S_2'}$, so the two totals are comparable in one weight
function. Suppose without loss of generality that $\wt{S_1} > \wt{S_2}$. Then $S_2$ is feasible for
$\syn$ and strictly cheaper than $S_1$, so $S_1$ does not attain
$\min\{\wt{S} : \bd_D S = \syn\}$.
\end{proof}

The proof is one paragraph and it uses nothing beyond feasibility. No oracle solves the instance, no
uniqueness is assumed, and no ground truth is consulted. Two decode calls and an addition produce a
statement that a specific returned correction is not minimum-weight, and the lighter answer is the
evidence.

\begin{observation}[When the certificate speaks]
\label{obs:whenfires}
The certificate fires exactly on the pairs where the two presentations return different total weight.
Suboptimality that both presentations share leaves the totals equal and produces no certificate, and
syndromes carrying two minimum-weight solutions of equal weight also produce none. What it produces on every
firing is a witness that a reader can check with the arithmetic of Lemma~\ref{lem:certificate}.
\end{observation}

\subsection{The running example, certified}
\label{sec:cert:example}

Section~\ref{sec:bg:example} left two answers on the table. Both explain the syndrome
$\syn = \{0,1,2,5\}$ of the distance-$3$ unrotated surface code, one with three mechanisms at total
weight $13.785$ and one with five at total weight $22.976$. They came from one decoder, belief
propagation with ordered-statistics post-processing \cite{panteleev2021degenerate,roffe2020decoding},
called twice on the same instance under two orderings of the mechanism list.

The relabelling here permutes the mechanism indices and fixes every detector, so it preserves
incidence and every weight, and Definition~\ref{def:presentation} is satisfied. Transporting the
second answer back through $\relab^{-1}$ puts both in the caller's index space, where the first is
$\{0, 6, 10\}$ and the second is $\{1, 3, 6, 7, 10\}$. Both have detector boundary exactly $\syn$.
The totals differ. By Lemma~\ref{lem:certificate} the five-mechanism answer is not a minimum-weight
explanation of $\syn$, and the three-mechanism answer is the proof of it.

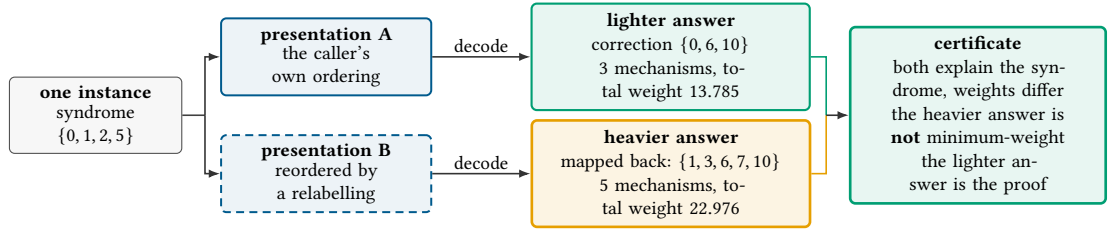
\begin{figure}[t]
\centering
\begin{tikzpicture}[
  font=\fontsize{7.15}{8.6}\selectfont, node distance=3mm,
  every node/.style={inner xsep=1.4mm, inner ysep=1.15mm},
  instance/.style={draw=black!70, fill=black!3, rounded corners=2pt, align=center, text width=20mm},
  presentation/.style={rounded corners=2pt, align=center, line width=.75pt, text width=25mm},
  answer/.style={rounded corners=2pt, align=center, line width=.8pt, text width=34mm},
  certificate/.style={draw=qTeal, fill=qTeal!9, rounded corners=2pt, align=center, line width=1pt,
                      text width=31mm},
  oplabel/.style={midway, above, font=\fontsize{7.15}{8.6}\selectfont, inner sep=.5mm},
  stem/.style={draw=black!75, line width=.6pt},
  flow/.style={-{Latex[length=1.5mm]}, draw=black!80, line width=.6pt}]

\node[presentation, draw=qBlue, fill=qBlue!7] (pa)
  {\textbf{presentation A}\\[-.4ex]the caller's own ordering};
\node[presentation, draw=qBlue, fill=qBlue!7, densely dashed, below=5mm of pa] (pb)
  {\textbf{presentation B}\\[-.4ex]reordered by a relabelling};

\coordinate (pmid) at ($(pa.west)!0.5!(pb.west)$);
\node[instance, left=5mm of pmid] (inst)
  {\textbf{one instance}\\[-.4ex]syndrome $\{0,1,2,5\}$};
\coordinate (ifork) at ($(inst.east)+(3mm,0)$);
\draw[stem] (inst.east) -- (ifork);
\draw[flow] (ifork) |- (pa.west);
\draw[flow] (ifork) |- (pb.west);

\node[answer, draw=qTeal, fill=qTeal!9, right=13mm of pa] (ra)
  {\textbf{lighter answer}\\[.15ex]correction $\{0,6,10\}$\\[.15ex]
   3 mechanisms, total weight $13.785$};
\node[answer, draw=qOrange, fill=qOrange!11, line width=1pt, right=13mm of pb] (rb)
  {\textbf{heavier answer}\\[.15ex]mapped back: $\{1,3,6,7,10\}$\\[.15ex]
   5 mechanisms, total weight $22.976$};
\draw[flow] (pa.east) -- node[oplabel] {decode} (ra.west);
\draw[flow] (pb.east) -- node[oplabel] {decode} (rb.west);

\coordinate (rmid) at ($(ra.east)!0.5!(rb.east)$);
\node[certificate, right=5mm of rmid] (out)
  {\textbf{certificate}\\[.15ex]
   both explain the syndrome, weights differ\\[.15ex]
   the heavier answer is \textbf{not} minimum-weight\\[.15ex]
   the lighter answer is the proof};
\coordinate (rjoin) at ($(out.west)-(3mm,0)$);
\draw[stem, draw=qTeal] (ra.east) -| (rjoin);
\draw[stem, draw=qOrange] (rb.east) -| (rjoin);
\draw[flow] (rjoin) -- (out.west);
\end{tikzpicture}
\caption{The certificate on the running example of Section~\ref{sec:bg:example}: one instance, two
presentations differing only in the order the mechanisms are listed, and two feasible answers of
different total weight. Since $13.785 < 22.976$ and both explain the syndrome, the heavier answer
is not minimum-weight and the lighter one is the proof. No oracle appears anywhere in the diagram.}
\label{fig:certificate}
\end{figure}

Figure~\ref{fig:certificate} draws the whole argument. Nothing in it solved the instance. No exact
matching was run, no ground-truth error was consulted, and the argument would read the same if the
true minimum were two mechanisms or three. What the caller learns is that on one of the two
orderings, an ordering it did not choose deliberately, the library returned an answer $67\%$ heavier
than one the library itself produced. Observation~\ref{obs:whenfires} says where that earns its
place: an exact oracle answers a harder question and costs accordingly, and Section~\ref{sec:cost}
measures how much the certificate's independence from any solver is worth.

\subsection{Reading a firing correctly}
\label{sec:cert:conditions}

A weight disagreement is arithmetic, and turning one into a statement about a library requires the
arithmetic to be about that library. Several ways of getting it wrong have nothing to do with the
decoder. The conditions for reading a firing as a verdict were fixed before the search of
Section~\ref{sec:hunt} ran, and Appendix~\ref{app:conditions} lists all of them. Four shape the
experiment itself.

\emph{The comparison is in exact integers.} Weights enter the libraries as exact even integers drawn
from a broad range, every total is accumulated exactly, and a firing must differ by at least two
units in the solver's own domain. A reader suspecting a floating-point tie-break is asking the right
question, and integer arithmetic answers it in advance.

\emph{The two presentations are verified to be one problem.} Before either decode call the permuted
graph is checked edge by edge against the image of the canonical one, for the same incidence,
boundary status and weight. What is checked is the object handed to the library, after any adaptation
layer, and the inverse map is applied and compared as well.

\emph{The graph is simple.} One edge per endpoint pair, asserted at construction. Parallel edges of
different weight would make a returned endpoint pair ambiguous, and merge policies would then decide
the answer instead of the solver.

\emph{The declaration has to cover the configuration.} A verdict against exactness requires the
subject to declare exactness for the operation actually called. A library that reduces a different
problem to matching and then calls an exact matcher is answering a question of its own, and a
disagreement there is a statement about the reduction.

\subsection{Where the properties sit}
\label{sec:cert:catalogue}

Table~\ref{tab:properties} lists the properties this paper measures with their gates and hypotheses.
Three need no declaration and no oracle. Three need both, and one of those appears in
Section~\ref{sec:accessmap} as the strongest established alternative to the certificate. The
certificate itself has no gate, deliberately. A weight disagreement is arithmetic and is always a
measurement, and it is an adverse verdict only against a subject that declared exactness.

\begin{table}[t]
\caption{Properties, their gates, and what each one needs. \emph{Oracle} is two different things in
this paper and the columns keep them apart: a \emph{known injected ground truth}, which
bounded-distance correctness needs because its claim is about the error that was injected, and an
\emph{optimisation solver}, which exact optimality and equivariance need because their claims are
about a minimum. Only the latter is what Section~\ref{sec:cost} prices, and only the latter is what
\emph{oracle-free} refers to throughout. The gate names the capability a subject must declare before
the property yields an adverse conformance verdict, except for $L_0$, which is mandatory: it is
incurred by exposing the interface and is never declared. A property with no gate still returns a
measurement, and Definition~\ref{def:layer} keeps that measurement apart from a verdict.}
\label{tab:properties}
\small
\setlength{\tabcolsep}{4pt}
\begin{tabular}{L{0.20\linewidth}L{0.11\linewidth}ccL{0.40\linewidth}}
\toprule
& & injected & optimisation & \\
property & gate & ground truth & solver & what it claims \\
\midrule
residual (Def.~\ref{def:pb}) & $L_0$, mandatory & no & no & the correction explains the syndrome \\
purity & $\mathit{det}$ & no & no & the answer is a function of the syndrome \\
relabelling certificate (Lem.~\ref{lem:certificate}) & none & no & no & a heavier answer is not minimal \\
bounded distance & $L_1$ & \textbf{yes} & no & the correction explains $\syn$ \emph{and} the residual is identity \\
exact optimality & $L_2$ & no & \textbf{yes} & the returned weight equals the optimum \\
matching-core equivariance & $L_2$ & no & \textbf{yes} & a unique optimum is returned permuted \\
\bottomrule
\end{tabular}
\end{table}

The purity property is the one line of Table~\ref{tab:properties} we have not yet described. It asks
whether decoding a syndrome after other calls returns what decoding it on a fresh instance returns. A
caller would like to assume it, but wanting a guarantee is not the same as being owed one, so it
yields a conformance verdict only against a declared determinism guarantee, is withheld from a subject
that declares a randomised schedule, and elsewhere is an instance-measured observation.

\section{Relations and the population they run on}
\label{sec:relations}

Section~\ref{sec:certificates} used one relabelling. A decoding instance admits several other changes
of presentation, and they are not all of one kind. R1 to R4 are bookkeeping and leave the model
itself untouched. R5 rewrites the model while preserving the detector and observable distribution it
induces. R6 preserves neither. Scaling every weight rescales the channel probabilities encoded by the model. It preserves the
minimum-weight argmin. The test therefore checks an exactness claim, not a distributional claim. This section names the six, says which subjects each can be asked
of, and fixes the population.

\subsection{Six ways to present one instance}
\label{sec:rel:six}

Each relation transforms a presentation and states what the answer must satisfy afterwards. The
transformations come from the degrees of freedom the interface has, and Table~\ref{tab:relations}
lists them.

\begin{table}[t]
\caption{The six relations. Each row transforms the presentation of one instance and states the
invariant the answer must satisfy. $\relab$ is a relabelling of Definition~\ref{def:presentation},
$\corr$ is the returned correction and $\obs$ the returned observable vector of length $k$.}
\label{tab:relations}
\small
\begin{tabular}{L{0.15\linewidth}L{0.35\linewidth}L{0.41\linewidth}}
\toprule
name & transformation & invariant required of the answer \\
\midrule
R1 mechanism order & permute the listing order of the mechanisms
  & $\relab^{-1}(\corr')$ explains $\syn$ and carries the same weight as $\corr$ \\
R2 detector labels & permute the detector numbering, and $\syn$ with it
  & $\corr' = \corr$ as an edge set, and $\obs' = \obs$ \\
R3 observable width & declare one further observable that no mechanism flips
  & $\obs'$ has length $k{+}1$, agrees with $\obs$ on the first $k$ entries, and is zero on the last \\
R3b observable offset & shift every observable bit up by one, leaving bit $0$ unused
  & $\obs'$ is zero on bit $0$ and agrees with $\obs$ above it \\
R4 call order & decode from a fresh subject, after other calls, and twice in a row
  & the same $(\corr, \obs)$ in every case \\
R5 mechanism split & replace each mechanism by two parallel ones whose probabilities combine to it
  & \emph{where the split leaves the optimum where it was}, the correction still explains $\syn$ and
  lands in the same observable class \\
R6 weight scale & multiply every weight by $\gamma > 0$
  & the chosen correction is unchanged, since scaling cannot move an argmin \\
\bottomrule
\end{tabular}
\end{table}

R1 and R2 are the two index spaces the interface has. R3 and R3b are the width of the declared
observable set. R4 is purity. R5 is compilation provenance, and it matters because Stim emits parallel
mechanisms whenever two circuit-level error sources produce the same detector and observable pattern.
R6 is numeric representation, and it is the relation the $\mathit{int\text{-}weights}$ flag exists to
predict.

R5 carries the hypothesis its row states, and it is easy to assume the relation is unconditional.
Splitting a mechanism of probability $p$ into two parallel mechanisms of probabilities $p_1$ and
$p_2$ requires
$p = p_1 + p_2 - 2 p_1 p_2$, and for $p_1, p_2 \in (0, 1/2)$ this gives
$p - p_1 = p_2(1 - 2p_1) > 0$ and likewise $p > p_2$. Since the weight
$\log\frac{1-p}{p}$ decreases in $p$, \emph{both} replacements are strictly heavier than the mechanism
they replace. A cheapest explanation that ran through the original therefore becomes more expensive,
and on some syndromes the minimum moves to a different correction, possibly in a different logical
class. R5 preserves the detector and observable distribution the model induces, but it does not preserve
the optimisation problem.

So the harness decides faithfulness per cell, not by assumption. A syndrome is excluded, not scored,
when the split moves the optimum. Of $38{,}340$ judged R5 results, $12{,}537$ are
excluded on that ground and counted as exclusions, which is the same discipline
Definition~\ref{app:def:scale} applies to R6 at the ends of its $\gamma$ ladder. An excluded cell is
never a pass.

That discipline has a price, and it separates R5 from every other relation here. Faithfulness for R1
and R2 is decided structurally. A relabelling that preserves every weight preserves the feasible
family and the objective, and checking that costs a pass over the edge table. Faithfulness for R6 is
decided the same way, by inspecting the scaled weights. R5 admits no such test. Whether the split
moved the optimum is a question about the optimum, and the only way we know to answer it is to solve
the split model exactly and check that its solution carries the base observable parity. R5 therefore
costs one exact solve per judged syndrome, and it is the only relation in this paper that does.

We account for this consequence explicitly. R5 belongs to the measurement instrument of
Section~\ref{sec:measurement}, where an exact oracle is deliberate ground truth and its cost is
accounted for. It is \emph{not} part of the deployable contract whose cost Section~\ref{sec:cost}
reports and whose detection Section~\ref{sec:accessmap} measures. That contract is R1, R2, R3, R3b,
R4 and R6, every relation whose precondition a caller can settle from the model in front of it. A
version of R5 without the faithfulness test would be oracle-free and unsound, reporting a legitimate
re-optimisation as a violation, and an unsound detection is not a cheaper one.

Cross-decoder disagreement is not an oracle here and is never used as one. Two decoders that differ
tell a reader which two decoders differ. Every verdict in this paper comes from the caller's own model
by the arithmetic of Section~\ref{sec:certificates}, or from an exact solve of that model.

\subsection{A relation runs where its quantity exists}
\label{sec:rel:gating}

Table~\ref{tab:relations} asks about two different returned objects, and the panel splits on which
one a library produces. All nine subjects expose a correction in the caller's index space, which is
why the relations over corrections are dispatched to all of them. Only two also compute a predicted
observable vector inside the library. For the other seven the caller computes observables from the
correction, so a relation over the observable vector would measure our own arithmetic and report it
as a property of the library.

\begin{definition}[Applicability]
\label{def:applicability}
R1, R2 and R5 are dispatched to every subject, since every subject exposes a correction in the
caller's index space. R3 and R3b are dispatched only to subjects that compute the observable vector
themselves, and are recorded as \textsc{not-applicable} for the rest. R4 yields a conformance verdict only against a
subject that declares $\mathit{det}$, is a legal skip for a subject that declares $\mathit{rand}$,
and for a subject silent on both is dispatched on the instance-measured layer only, where it records
an observation and can never produce an adverse verdict. Reading silence as ``not randomised'' would
be an inference from the absence of a claim, which is what Definition~\ref{def:declstate} forbids. R6 is dispatched to every subject and the
result is reported separately for those that carry $\mathit{int\text{-}weights}$, declared or
measured. The flag gates nothing, so its provenance changes how the split is described and not which
subjects fall on which side.
\end{definition}

A \textsc{not-applicable} entry sits in its own column and is never added to a pass, for the reason
Definition~\ref{def:dispatch} keeps $\unassessable$ apart from a legal skip. The split is itself a
measured property of the panel, and Section~\ref{sec:accessmap} returns to it.

\subsection{The population}
\label{sec:rel:population}

Table~\ref{tab:instances} lists the instances. Every syndrome of every one of them is tested, so
there is no seed, no sample and no sampling error to report, and the scope of a claim is exactly the
population in the table.

\begin{table}[t]
\caption{The instance population. Built at code capacity with $p = 0.01$, so syndrome extraction is
noiseless. Every syndrome of every instance is tested. Bounded-distance correctness has further
hypotheses that these instances do not all meet, and Appendix~\ref{app:population} says which. The
$d$ column is the nominal code distance; for the repetition rows it is the protected-sector distance,
and the full-Pauli distance those instances certify is $1$.}
\label{tab:instances}
\small
\begin{tabular}{llrrrr}
\toprule
id & family & $d$ & detectors & edges & syndromes tested \\
\midrule
I1 & repetition & 3 & 2 & 3 & all 4 \\
I2 & repetition & 5 & 4 & 5 & all 16 \\
I3 & repetition & 7 & 6 & 7 & all 64 \\
I4 & rotated surface & 3 & 4 & 7 & all 16 \\
I5 & unrotated surface & 3 & 6 & 13 & all 64 \\
I6 & rotated surface & 5 & 12 & 21 & all 4096 \\
\midrule
\multicolumn{4}{l}{total} & & 4260 \\
\bottomrule
\end{tabular}
\end{table}

Of those 4260 syndromes, 2867 have a minimum-weight solution that an exact solver proves unique, and
none was left undecided by the solver. Invariants that compare chosen corrections are counted over
those 2867, since two distinct optima of equal weight leave a decoder free to return either.

R3 and R3b run on a second family built for the purpose, in which mechanism $k$ flips detector $k$ and
observable $k$. Every syndrome there has exactly one explanation, so no tie-breaking is involved, and
the declared observable count is the only quantity that varies. We use the counts
$k \in \{1, 2, 8, 30, 31, 32, 33, 62, 63, 64, 65\}$, which bracket the $32$-bit and $64$-bit machine
word boundaries on both sides.

\subsection{The panel}
\label{sec:rel:panel}

Table~\ref{tab:panel} lists the subjects, what each declares, and which return path it uses.
Appendix~\ref{app:audit} carries the sentence behind every declaration and the terms searched behind
every silence.

\begin{table}[t]
\caption{The panel: nine configurations of five public libraries. Declarations are documentary facts
at the pinned version, in the vocabulary of Definition~\ref{def:declstate}. \textsc{pmc} is configured for correlated decoding and is
behaviourally identical to \textsc{pm}, for the reason given in Section~\ref{sec:measurement:corr}.}
\label{tab:panel}
\small
\begin{tabular}{lllll}
\toprule
id & library and configuration & $L_1$ & $L_2$ & returns \\
\midrule
\textsc{pm} & pymatching 2.4.0, sparse blossom & \silent & \yes & observables \\
\textsc{pmc} & pymatching 2.4.0, correlations enabled & \silent & \silent & observables \\
\textsc{fb} & fusion-blossom 0.2.13 & \silent & \yes & correction \\
\textsc{ldpc} & ldpc 2.4.1, BP with OSD-0 & \silent & \silent & correction \\
\textsc{ldpc-r} & ldpc 2.4.1, randomised serial schedule & \silent & \silent & correction \\
\textsc{mwpf-j} & mwpf, serial joint-single-hair & \silent & \silent & correction \\
\textsc{mwpf-s} & mwpf, serial single-hair & \silent & \silent & correction \\
\textsc{mwpf-u} & mwpf, serial union-find & \silent & \silent & correction \\
\textsc{tess} & tesseract-decoder & \silent & \silent & correction \\
\bottomrule
\end{tabular}
\end{table}

Two libraries declare exact minimum-weight optimisation, PyMatching \cite{higgott2022pymatching}
and fusion-blossom \cite{fusionblossom}, and those two declarations are what Section~\ref{sec:hunt}
puts under load.

\section{What the panel declares, and what that costs}
\label{sec:measurement}

This section reports the paper's main finding. It measures the audited panel, not any one library on it. It also says what a capability-gated
checker can reach there.

\subsection{Nobody declares bounded-distance correctness}
\label{sec:measurement:decl}

We read the documentation of every subject in Table~\ref{tab:panel} at its pinned version and
recorded a declaration state for each capability, in the vocabulary of
Definition~\ref{def:declstate}. The result is Table~\ref{tab:declarations}. Appendix~\ref{app:audit}
carries the audit itself, with the sentence behind every $\yes$ and the terms searched behind every
$\silent$.

\begin{table}[t]
\caption{Declaration states across the seven declarations on the panel. A state describes the
documentation at the pinned version and never the implementation. The seven declarations cover six
behaviourally distinct configurations, since \textsc{pmc} computes what \textsc{pm} computes.}
\label{tab:declarations}
\small
\begin{tabular}{lrrr}
\toprule
capability & \yes & \no & \silent \\
\midrule
$L_1$ bounded-distance correctness & \textbf{0/7} & 0/7 & \textbf{7/7} \\
$L_2$ exact minimum-weight optimisation & 2/7 & 0/7 & 5/7 \\
$L_3$ exact maximum-likelihood decoding & 0/7 & 0/7 & 7/7 \\
\bottomrule
\end{tabular}
\end{table}

Two numbers in Table~\ref{tab:declarations} carry the section. $L_1$ is $\silent$ seven times out of seven: bounded-distance
correctness is what a caller most directly depends on, and nobody claims it in writing. And $\no$ is
empty at every capability, so the safe harbour ISO/IEC 9646-7 Appendix~I.4.1.3 grants an explicit NO
applies nowhere here. These are not libraries stating their limits but libraries saying nothing, and
\textsc{silent} is never evidence that a capability is absent or withheld.

A $\no$ requires a sentence declining the capability. A paraphrase does not supply one, and asserting
a NO from an unquoted paraphrase is the same inference as reading a YES from a class name.

\subsection{What silence costs a checker}
\label{sec:measurement:yield}

Definition~\ref{def:dispatch} turns silence into $\unassessable$. Table~\ref{tab:yield} measures what
that does to a capability-gated checker over 9 subjects and 791 cases, so 7119 cells per property. The
cases come from a grid of 29 instance configurations spanning both code families at $d = 3, 5, 7$,
both weight regimes, both noise regimes and thirteen declared-observable strata, at up to 30 cases
each. It is broader and shallower than the exhaustive population of Table~\ref{tab:instances}, and it
is used only here.

\begin{table}[t]
\caption{The three dispatch yields of Definition~\ref{def:yield}, over 9 subjects $\times$ 791 cases
$=$ 7119 cells. \emph{Direct} counts cells where the documentation declares the gated capability, and
any required flag, itself. \emph{Derived} also counts cells where a proved implication carries a
declared capability to the gated one on a case the evaluator has checked. \emph{Applicable} counts
cases where the hypotheses hold, ignoring declarations, and so has 791 as its denominator.
$\ast$ marks a \emph{mandatory} gate that is incurred by exposing the interface, not declared. Its
$100\%$ is therefore a property of the gate, not a count of declarations. Each row carries its own
gate, its own hypotheses and its own denominator, and is read on its own terms.}
\label{tab:yield}
\small
\begin{tabular}{lllrrr}
\toprule
property & gate & needs & direct & derived & applicable \\
\midrule
residual & $L_0$\textsuperscript{$\ast$} & --- & 100.0\% & 100.0\% & 100.0\% \\
exact optimality & $L_2$ & --- & 22.2\% & 22.2\% & 100.0\% \\
matching-core equivariance & $L_2$ & --- & 20.5\% & 20.5\% & 92.2\% \\
strict equivariance & $L_2$ & --- & 1.7\% & 1.7\% & 7.5\% \\
\textbf{bounded-distance correctness} & $L_1$ & --- & \textbf{0.0\%} & 13.8\% & \textbf{62.1\%} \\
\bottomrule
\end{tabular}
\end{table}

\begin{figure}[t]
\centering
\includegraphics[width=0.85\linewidth]{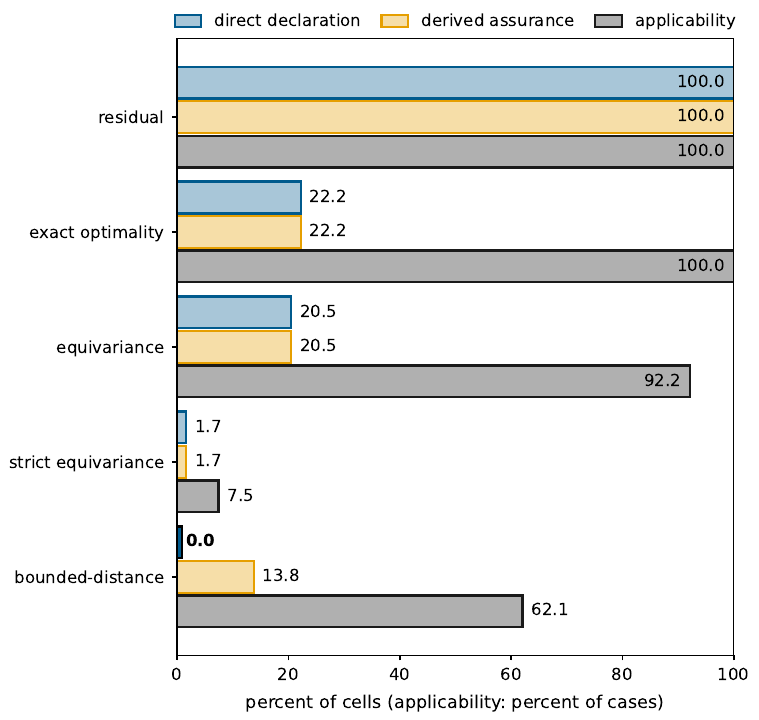}
\caption{The three yields of Table~\ref{tab:yield}, drawn from the artifact that table reads. Direct declaration and derived assurance are percentages of the $7119$ subject-by-case cells.
Applicability is a percentage of the $791$ cases. It depends on the population, not on any
documentation.
The row to read is bounded-distance correctness, at the foot of the chart. It is arithmetically
available on
$62.1\%$ of cases and declared by nobody, and the distance from its middle bar to its lower one is
what the panel's silence costs a caller.}
\label{fig:yield}
\end{figure}

The bounded-distance row is the paper's main finding and the three columns say why, and
Figure~\ref{fig:yield} draws the distance between them. Its hypotheses
hold on $62.1\%$ of cases, so the arithmetic is available on nearly two thirds of the population.
Its direct declaration yield is $0.0\%$ uniformly. Every one of the nine subjects contributes $0$ of
its $491$ applicable cells. A conformance regime reaches exactly as far as a library's willingness to
write its guarantee down, and here that reach is zero.

The obvious objection is that the zero was manufactured by an over-strict dispatch rule, and the
derived column answers it. Switch compositional reasoning on --- let a declared $L_2$ carry to $L_1$
on the code-capacity cases whose weights order candidates by Hamming weight, which
Remark~\ref{app:rem:domain} proves --- and the yield moves from $0.0\%$ to $13.8\%$. It moves that far
and no further, because the entire gain comes from the two subjects that declare $L_2$. The other
seven declare nothing that implies anything, so no theorem can help them. The deficit is therefore a
property of what the panel publishes and not of how strictly we read it.

The two equivariance rows gate on $L_2$ alone, and a determinism flag is not among their
requirements. Both properties assume a \emph{verified-unique} optimum, and
Corollary~\ref{app:cor:det} proves that a unique optimum forces the answer, so determinism follows
from the hypotheses instead of being one. Nobody on the panel declares determinism, so demanding it
would have put both rows at $0.0\%$ for a reason unrelated to what the properties test. The rows read
$20.5\%$ and $1.7\%$, which is $2/9$ of their applicability.

Exact optimality is the control that shows the machinery is not inflating anything. Its hypotheses
hold everywhere, two of nine subjects declare $L_2$, and its direct yield is $22.2\%$, which is
exactly $2/9$. Nothing in the gate, the implication rule or the population moves it off that value.

The residual property is the one that asks for nothing, because $L_0$ is mandatory and cannot be
declined. It dispatches on every cell. The rest of this section can therefore measure all nine subjects at
once. Its $100\%$ reflects the mandatory gate, not a declaration.

Strict equivariance is a negative result about our own contract, not about any library. It
requires an automorphism preserving the complete weighted model including the observable masks. At
code capacity no non-trivial such automorphism exists on either family, at any distance, under either
weight regime, so it is dispatched $0$ times out of $1200$. At circuit level on the rotated surface
code they do exist, and it is dispatched $4320$ times. It never fires. Inference from strict
equivariance is defined on populations that contain a reachable positive control, and this population
contains none.

\subsection{Corrections come back in the caller's index space}
\label{sec:measurement:resid}

We ran every relation of Table~\ref{tab:relations} against every subject its gate opens for, over the
whole population of Table~\ref{tab:instances}, and the two observable relations over the second
family built for them in Section~\ref{sec:rel:population}. That is $717{,}244$ judged results.

\begin{observation}[No malformed or syndrome-infeasible correction was observed]
\label{obs:resid}
Across all $717{,}244$ judged results, every returned correction explained its syndrome in the
caller's index space. No return contained an out-of-range index and none contained a repeated index.
\end{observation}

On this population the interface obligation of Section~\ref{sec:cert:residual} is met by every
library we can reach, in every presentation we hand them. The property evaluates boundary
preservation under the caller's interpretation of the returned indices, and the observation
establishes that every return satisfies it over the complete tested syndrome space.

\subsection{Which presentation you choose can decide whether an approximate decoder is optimal}
\label{sec:measurement:presentation}

The two exactness-declaring subjects returned the unique minimum-weight correction under every
presentation, on all $2867$ syndromes whose optimum is unique. So did \textsc{tess}, which declares
nothing. The approximate solvers behaved differently, and \textsc{mwpf-u} shows the effect most
clearly. On $26$ of those $2867$ syndromes it returned the unique optimum under one detector numbering
and a suboptimal correction under another. Here is one, on the rotated surface code at $d=5$:

\begin{quote}\small
syndrome $011111101000$. Under the model's own detector numbering the answer is edges
$\{3,4,10,12\}$ at weight $17.687434$, which is the unique optimum. Under a relabelling it is
$\{3,5,6,7,12\}$ at weight $22.282554$. Both explain the syndrome.
\end{quote}

A logical-error-rate benchmark runs at one fixed detector numbering, so it never varies the quantity
that moved here. Run under a second numbering it would catch this particular syndrome, whose two
answers sit in different logical classes, and would still miss any change of correction inside one
class. The numbering is a convention of whoever built the model, not a property of the experiment.

\subsection{The effect survives replication over independent relabellings}
\label{sec:measurement:replication}

Exhausting the syndrome space replicates over syndromes. It says nothing about the one arbitrary
choice we made, which is \emph{which} relabelling to apply. So we drew $20$ independent relabellings
per instance per relation and repeated the whole measurement.

\begin{table}[t]
\caption{Changed answers on unique-optimum syndromes, over 20 independent relabellings per relation.
Ranges are over draws. The four subjects at the top produced no change and no certificate in any of
their 40 draws.}
\label{tab:replication}
\small
\begin{tabular}{lrrrr}
\toprule
subject & R1 median (range) & R2 median (range) & R1 certificates & R2 certificates \\
\midrule
\textsc{pm}      & 0 (0--0)   & 0 (0--0)     & 0--0   & 0--0 \\
\textsc{pmc}     & 0 (0--0)   & 0 (0--0)     & 0--0   & 0--0 \\
\textsc{fb}      & 0 (0--0)   & 0 (0--0)     & 0--0   & 0--0 \\
\textsc{tess}    & 0 (0--0)   & 0 (0--0)     & 0--0   & 0--0 \\
\textsc{mwpf-j}  & 2 (1--4)   & 2 (1--5)     & 2--12  & 5--16 \\
\textsc{mwpf-s}  & 3.5 (1--6) & 5 (2--8)     & 5--30  & 13--38 \\
\textsc{mwpf-u}  & 20 (3--43) & 104.5 (69--160) & 14--42 & 55--108 \\
\textsc{ldpc}    & 12 (6--17) & 0 (0--0)     & 160--224 & 0--3 \\
\textsc{ldpc-r}  & 79 (52--135) & 0 (0--0)   & 154--293 & 0--1 \\
\bottomrule
\end{tabular}
\end{table}

The draws are paired. A relabelling is drawn per instance and per relation, not per subject, so every
subject receives the same one. Comparing marginal ranges throws that pairing away and a paired sign
test over the $20$ draws uses it. The test gives a directional tendency among the three mwpf solvers,
reported as post hoc since the paired analysis was added after the data were in hand. Single-hair
exceeds joint-single-hair
on $16$ of $20$ draws under R1 with four ties and no reversals, and on all $20$ under R2. Union-find
exceeds single-hair on $18$ of $20$ under R1 and on all $20$ under R2. Across the six comparisons the
$p$ values run from $1.91 \times 10^{-6}$ to $4.02 \times 10^{-4}$, and all six survive a Bonferroni
correction over the $111$ comparisons the checker performs, whose threshold is
$4.50 \times 10^{-4}$. What this supports is the tendency $\textsc{mwpf-j} < \textsc{mwpf-s} <
\textsc{mwpf-u}$ and not pointwise dominance: single-hair exceeds union-find on $2$ of the $20$ R1
draws.

The four zeros at the top of Table~\ref{tab:replication} need a precise reading. Those subjects
produced no changed answer on a unique-optimum syndrome and no certificate in any draw, but they did
change the correction returned on syndromes with two optima of equal weight, which the problem
permits: $12{,}835$ such changes for \textsc{pm}, the same for \textsc{pmc}, $17{,}017$ for
\textsc{fb} and $20{,}253$ for \textsc{tess} across forty draws each. The invariant is about unique
optima.

\subsection{Whether any of this matters at operating rates}
\label{sec:measurement:rates}

Everything above is measured on an enumerated code-capacity population, where every syndrome counts
once regardless of how likely it is. That establishes the effect exists but not whether a running
experiment would meet it, so we measured that separately. The setting is circuit level on the rotated surface code with the $10^{-3}$ symmetric noise profile.
It uses native physical weights, not the integers of Section~\ref{sec:hunt}. The syndromes are
sampled from the noise model, not enumerated. Mean
syndrome weight is $1.8$ at $d = 5$ and $5.3$ at $d = 7$, the sparse regime a low-noise experiment
produces. Twenty thousand shots per cell, one relabelling, the same R1 used throughout.

\begin{table}[t]
\caption{Presentation dependence on syndromes drawn from the noise model, 20{,}000 shots per cell,
circuit level with native weights. The last column is the only one a logical error rate could see.}
\label{tab:rates}
\small
\begin{tabular}{llrrr}
\toprule
$d$ & subject & correction changed & weight changed & \textbf{logical class changed} \\
\midrule
5 & \textsc{ldpc}  & 47 \; (0.23\%) & 34 \; (0.17\%) & \textbf{1} \; (0.005\%) \\
5 & \textsc{pm}    & 0 & 0 & 0 \\
5 & \textsc{fb}    & 0 & 0 & 0 \\
\midrule
7 & \textsc{ldpc}  & 135 \; (0.68\%) & 89 \; (0.45\%) & \textbf{0} \\
7 & \textsc{pm}    & 0 & 0 & 0 \\
7 & \textsc{fb}    & 3 \; (0.01\%) & \textbf{0} & 0 \\
\bottomrule
\end{tabular}
\end{table}

The effect is real here and it is small. Belief propagation with post-processing returns a different
correction on two to seven shots per thousand, and a correction of different total weight on two to
five per thousand, purely because the mechanisms were listed in a different order. Those are
certificates firing on syndromes an experiment would really see.

It almost never reaches the logical class, at one shot in twenty thousand at $d = 5$ and none at
$d = 7$. Presentation dependence at these parameters is therefore an interface-level effect that reaches the
correction two to seven times per thousand shots and the logical class once in twenty thousand. A
logical error rate reads only the place where it is rarest. It is not
a defect, and this paper does not call it one. \textsc{ldpc} declares no capability that a heavier or
different answer could contradict, so under Definition~\ref{def:layer} what is reported here is an
instance-measured observation about an approximate algorithm. The word would be available only
against a subject that had declared exactness. The
logical error rates themselves are not compared, because at $10^{-3}$ noise these cells produce two or
three logical errors in twenty thousand shots.

\textsc{fb}'s three changed corrections at $d = 7$ show the certificate declining to fire when it
should. All three changed the edge set and none the weight, an exact solver choosing between two
minimum-weight solutions of equal weight, which is a freedom the problem grants it.

Table~\ref{tab:rates} is a measurement of \textsc{ldpc} at these parameters. The mwpf family showed
the largest effect on the enumerated population and is not installed in the environment this
measurement runs in. What the table supports is that statement and one direction. Presentation
dependence reaches the logical class at a rate far below the rate at which it reaches the correction.

\subsection{The certificate against an exact oracle}
\label{sec:measurement:soundness}

Lemma~\ref{lem:certificate} is a proof, and the implementation of a proof is a separate question. Over
$76{,}680$ paired comparisons the relabelling relations produced $639$ certificates. For each one we
asked an exact solver for the minimum weight of that syndrome and checked that the heavier answer
exceeds it.

\begin{observation}[The implementation agrees with the lemma]
\label{obs:soundness}
All $639$ certificates place the heavier correction strictly above the exact minimum weight. No
certificate contradicts the oracle.
\end{observation}

Weights are summed with compensated arithmetic and every return is checked for range and repetition
before it is used, so neither summation order nor a malformed answer can manufacture a certificate.
The certificates also reach where this particular oracle-based count does not. That count compares
the returned correction with a designated optimal edge set, and there is no such set when two optima
have equal weight; comparing the returned \emph{weight} with the optimum value stays well defined
everywhere, and a weight difference between two feasible answers is a proof either way.

\subsection{The flag for numeric representation predicted the wrong subjects}
\label{sec:measurement:intweights}

Relation R6 multiplies every weight by $\gamma$, an operation that cannot move an argmin. The
$\mathit{int\text{-}weights}$ flag exists to predict which subjects are sensitive to it, because a
solver that rounds weights to integers can collapse two distinct weights onto one value.

Over $\gamma \in [10^{-6}, 10^{4}]$ the four subjects that carry $\mathit{int\text{-}weights}$
changed no answer at all. Five carry no such flag, and two of those five, both BP with OSD, changed
$1598$ and $1748$ answers on unique-optimum syndromes and moved optimality $1590$ and $1730$ times.
The other three changed nothing. The flag marked four invariant subjects and missed both sensitive
ones. One of the
four carries the flag as a declaration, with the sentence quoted in Appendix~\ref{app:audit}. The
other three carry it as a measurement, because our own driver integerises their weights and nothing
in their documentation claims integer arithmetic. The grouping is the one the finding is about
either way, and marking its provenance is the point of Definition~\ref{def:declstate}. The reason
is not rounding. Scaling the weights rescales the channel probabilities that belief propagation is
initialised from, and belief propagation is not scale invariant even where the minimum-weight solution
is. A flag that predicts the wrong subjects is a defect of our capability vocabulary, and it is one
this measurement found.

One third of the R6 cells are excluded and counted, at both ends of the $\gamma$ ladder, where the
probability representation clamps and the transformation stops being a scaling.
Appendix~\ref{app:tables} gives the per-instance counts.

\subsection{One configuration on the panel is a duplicate}
\label{sec:measurement:corr}

\textsc{pmc} enables correlated decoding, and it returns what \textsc{pm} returns on every syndrome of
every instance. Correlated decoding consumes correlations carried by the non-graphlike mechanisms of a
detector error model, and the reduction to a matching graph has already merged those away before an
adapter sees the model. The mode is therefore not reachable through the matching-graph interface, and
we report the panel as seven declarations over six behaviourally distinct configurations. The
declaration of \textsc{pmc} still differs from that of \textsc{pm}, which is why it stays in
Table~\ref{tab:declarations}.

\section{Putting the exactness declarations under load}
\label{sec:hunt}

Two libraries on the panel declare exact minimum-weight optimisation. Section~\ref{sec:measurement}
found no certificate against either one on the enumerated population. This section asks the same
question on instances an order of magnitude larger, where enumerating the syndrome space is out of
reach and the check has to work from two decode calls alone.

\subsection{One question, and the conditions for answering it}
\label{sec:hunt:question}

The question is narrow. Does a subject that declares exact minimum-weight optimisation ever return
two feasible corrections of different total weight for two presentations of one instance? A firing
there is a conformance verdict against a published claim. A firing on a subject that declares nothing
is an observation about an approximate algorithm doing approximate work, which
Section~\ref{sec:measurement} has already reported.

The conditions for reading a firing were fixed before the search ran. Appendix~\ref{app:conditions}
lists all twelve, and Section~\ref{sec:cert:conditions} gave the four that shape the experiment. The
search was budgeted at one workstation day and the budget was not renewed.

We also recorded what we expected. From $80$ prior subject-and-draw blocks with no certificate on
either exactness-declaring subject, a Jeffreys posterior gives a rate of $0.0062$ and a
posterior-predictive probability of at least one firing of $10.5\%$ over $20$ further independent
cells and $18.3\%$ over $40$. Relabellings of one graph exercise correlated code paths, so the
committed planning figure was $5$ to $15\%$. A null here is the expected outcome.

\subsection{Where we looked}
\label{sec:hunt:where}

The instances are circuit-derived, which is where the graphs stop being small. A rotated surface code
at $d=13$ over $13$ rounds gives $2184$ detectors and $11{,}398$ edges, and its syndrome space cannot
be enumerated. Exact matching at that size is still polynomial, so the point is not that an oracle is
impossible but that the certificate needs two decode calls and none, which is what lets the search run
at this size inside its budget.

The physical weights are replaced by exact even integers from a broad deterministic spectrum. The
result is a set of circuit-derived matching graphs with artificial weights, not an extracted
circuit-level noise model. The comparison happens in integer arithmetic, where a firing must differ
by at least two units. Syndromes are the detector boundaries of deterministic edge subsets,
stratified over $|\syn| \in \{8, 16, 32, 64\}$ as the protocol fixed in advance, because low detector
counts make exact matching easy and measure little.
Both presentations are built through each library's own graph interface, and the permuted graph is
verified edge by edge against the image of the canonical one before either decode call.

\subsection{The answer}
\label{sec:hunt:answer}

Table~\ref{tab:hunt} reports it, together with the instrument health the answer depends on.

\begin{table}[t]
\caption{The search over the two subjects that declare exact minimum-weight optimisation. Both saw
identical canonical graphs, weights and syndromes. Instrument health was clean throughout: no
isomorphism check failed and no return was infeasible or unmappable.}
\label{tab:hunt}
\small
\begin{tabular}{llrrrrr}
\toprule
subject & instance & detectors & edges & syndromes & relabellings & certificates \\
\midrule
fusion-blossom 0.2.13 & $d{=}11$, circuit level & 1320 & 6718 & 800 & 50 & \textbf{0} \\
fusion-blossom 0.2.13 & $d{=}13$, circuit level & 2184 & 11398 & 800 & 50 & \textbf{0} \\
pymatching 2.4.0 & $d{=}11$, circuit level & 1320 & 6718 & 800 & 50 & \textbf{0} \\
pymatching 2.4.0 & $d{=}13$, circuit level & 2184 & 11398 & 800 & 50 & \textbf{0} \\
\midrule
\multicolumn{5}{l}{total paired comparisons} & & 160{,}000 \\
\bottomrule
\end{tabular}
\end{table}

Over $160{,}000$ paired comparisons neither declaration was contradicted. Both libraries returned
corrections of identical total weight under every relabelling we drew, on graphs an order of magnitude
larger than the enumerated population.

\subsection{The same population against an exact optimum}
\label{sec:hunt:optimum}

Section~\ref{sec:cost:oracle} establishes that the exact oracle is affordable on these graphs, so the
declaration can be checked against the optimum directly. We attempted an exact solve of all $1600$
canonical
syndromes and compared what each subject returned against the optimum, in exact integers. The oracle
settled $1599$ inside its budget. The one it did not is reported as unsettled and never as a pass.

\begin{table}[t]
\caption{Every canonical syndrome of the search population against an exact optimum. A syndrome the
oracle did not settle inside its time limit is reported as unsettled, never as a pass.}
\label{tab:canonopt}
\small
\begin{tabular}{lrrrrr}
\toprule
instance & syndromes & solved & oracle time & pymatching $=$ optimum & fusion-blossom $=$ optimum \\
\midrule
$d{=}11$, circuit level & 800 & 800 & $1062$ s & \textbf{800/800} & \textbf{800/800} \\
$d{=}13$, circuit level & 800 & 799 & $3171$ s & \textbf{799/799} & \textbf{799/799} \\
\bottomrule
\end{tabular}
\end{table}

On every settled syndrome both subjects returned a correction of exactly the optimum weight.
Neither was ever above it, and neither ever returned one the oracle could not match. This is a
stronger statement than the paired null above it and it is still bounded in the same way. It covers
the syndromes we drew, on two instances, and $1$ of the $1600$ is not covered at all, because the
oracle did not settle it within its budget.

\subsection{What the full solve found in our own oracle}
\label{sec:hunt:oraclegap}

Running the direct check is also what exposed a defect in the instrument, and the defect had been
there since the instrument was built. The first version of Table~\ref{tab:canonopt} reported
$799$ of $800$ at $d = 11$ for both subjects, with the remaining syndrome in neither the equal nor the
above column. There is only one way to be in neither, and it is to have returned a feasible correction
\emph{lighter} than the value we called the minimum.

It was not a decode failure and not an unexplained syndrome. Over all $1600$ syndromes neither
subject ever declined to answer, and every returned correction had zero residual. On canonical
syndrome $763$, with sixty flipped detectors, our oracle reported the minimum as $40784$ while
pymatching and fusion-blossom independently returned feasible corrections of weight $40782$. Two
exact matching implementations agreeing with each other against our mixed-integer program is already
the verdict, and the size of the disagreement names the cause. $2/40784$ is $4.9 \times 10^{-5}$,
inside the solver's default \emph{relative} MIP gap of $10^{-4}$, which our wrapper never overrode.
The solver returns its optimal status the moment that gap closes, and the wrapper read that status as
a proved optimum. At a zero gap the same solve returns $40782$ in the same $1.6$ seconds, so the
tolerance had bought nothing.

The error has a direction, and it is the forgiving one. A gap-terminated program returns a feasible
incumbent, so the value it reports is at or above the true optimum and never below it. An exactness
violation is credited when a subject returns \emph{more} than the reported optimum, so it exceeds the
true optimum too, and a certificate asserts that the heavier of two feasible answers exceeds the
reported minimum, which implies it exceeds the true one. Both survive. What does not follow is a
uniqueness verdict, which is a dispatch precondition for both equivariance properties.

So we measured how far it reached. Re-solving the enumerated populations
of Sections~\ref{sec:measurement} and~\ref{sec:accessmap} at both settings, over $1096$ syndromes
spanning both code families, both noise regimes and $d = 3, 5, 7$, gives $0$ minimum weights that
differ and $0$ uniqueness verdicts that differ, with $859$ syndromes unique under both. The gap needs
an objective large enough for $10^{-4}$ relative to span a real weight difference, and it found one
only on a search graph with $6718$ edges. Every number in this paper comes from the corrected
solver, and the uncorrected run is kept in the artifact.

\subsection{What a null of this shape is worth}
\label{sec:hunt:worth}

A null is worth what its conditions make it worth, and the conditions here were designed in rather
than checked afterwards. The comparison ran in exact integers with a two-unit margin, so the result is
a statement about exact arithmetic and not about rounding. The graphs are simple, one edge per
endpoint pair, so no merge policy could decide an answer. The permutation and its inverse were both verified against the canonical graph. A bug in our own
transport would therefore have surfaced as an isomorphism failure, not a finding. Neither presentation went through anything but the
library's own graph interface.

Read together with Section~\ref{sec:measurement}, the two sections give a caller a bounded statement.
Across an enumerated population and a search on much larger graphs, neither published exactness
declaration was contradicted under any presentation we constructed. The paired certificate returned no unequal-weight feasible pair over those presentations, and the
exact-optimum arm establishes minimum-weight agreement on every settled syndrome. That is why
Section~\ref{sec:hunt:optimum} does not stop there and settles the stronger question directly, with an
independent optimum for every syndrome the oracle settles, and on all $1599$ of those both subjects
returned exactly it. The
paired null still earns its place, because it is what a caller can run without an oracle at all, and
what it establishes is that the check which would have caught a presentation-dependent failure was
actually run, at this size, under conditions fixed in advance, and that it is cheap enough for
continuous integration.

\subsection{A third interface class}
\label{sec:hunt:chromobius}

We also planned a probe on a colour-code decoder, since colour codes are a genuinely different family
and the panel is otherwise built from matching. Chromobius 1.1.1 exposes
$\mathrm{predict\_obs\_flips\_from\_dets\_bit\_packed}$ and a weighted variant of it. Both return
predicted observable flips. Neither returns a correction.

The certificate of Lemma~\ref{lem:certificate} compares the weight of a returned correction, so it has
nothing to act on here. This is the same split that Definition~\ref{def:applicability} draws, one step
further along. The panel already divides into libraries that return observables and libraries that
return corrections. Chromobius is a third case, where the published interface exposes neither the
correction nor an index space to transport one into, so a correction-space check cannot be dispatched
to it at all. That is a fact about what the interface publishes, and Section~\ref{sec:accessmap}
measures the same fact from the other direction.

\section{Which part of the answer does your evaluation actually read?}
\label{sec:accessmap}

Definition~\ref{def:subject} gives an answer three possible components: a correction, an observable
vector, and the weight of the correction chosen. This section measures which of the three each
evaluation method inspects. The separations it reports are deterministic, and saying so is part of
reporting them.

\subsection{Thirteen operations on one interface}
\label{sec:access:design}

We take a decoder that Section~\ref{sec:measurement} measured to be exact on every unique-optimum
syndrome and wrap it so that exactly one enumerated interface operation is perturbed. We chose the
operations, six instances and twelve behaviours per cell. What is not chosen is which distinct
behaviours arise inside that design, every parameter being drawn from a hash of the operator, the
instance and an index.

Thirteen operations. Six act on the correction: permute its indices, shift them, drop an edge, add an
edge, duplicate an edge, and sort the edges by weight. Four act on the observable vector: truncate it,
pad it, rotate it, and flip a bit. One acts on the call history and returns the answer from one to
three calls ago. One rounds the model's weights before decoding. The last one perturbs the returned
weight.

Behaviours that coincide are merged before anything is counted, over the whole domain the detectors
exercise: every presentation, the twenty-thousand-call sequence the error-rate arm uses, and the
fresh-decoder lifecycle the purity relation needs. That leaves $387$ distinct behaviours after $48$
cross-operator collisions are removed, of which $6$ reproduce the unmutated decoder exactly and $381$
do not. No behaviour counted as equivalent is detected by anything.

Seven measurement columns are compared, and it is worth being exact about where each comes from.
\emph{Five are established practice}. They are the logical error rate along each of its two return
paths, observable-width validation, exhaustive observable checking, and exhaustive correction
checking against an exact oracle. \emph{The sixth is ours}, the capability-gated contract this paper builds.
\emph{The seventh is proposed here}. It is an assertion a caller can make directly from the returned
object. Leaving it out would have made the comparison incomplete, not fairer. A sixth
established method, the library's own test suite, is recorded \textsc{not-applicable} throughout,
since a wrapper is not the library, and that was ruled before the study ran.

The logical error rate appears twice because the panel splits on which object a library hands back.
Two subjects compute the observable vector inside the library. Seven return a correction and leave
the caller to derive it. A benchmark reads whichever it is given, so a single error-rate arm would
make a statement about one return path and be read as a statement about both. \emph{Error rate,
observable-native} samples $20{,}000$ error patterns per
behaviour, decodes them, and compares the observable vector the subject returns against the truth with
McNemar's exact test on the paired shots at $\alpha = 0.01$. \emph{Error rate, correction-native} runs
the same exact test on the same paired shots, with the observable recomputed from the returned
correction under the caller's own observable masks, which is what a caller of the other seven subjects
actually computes.

\emph{Width} checks that the returned observable vector has the length the
model declares. \emph{Observable check} compares the prediction against the classes an exact
minimum-weight solution can reach, on every syndrome. \emph{Correction check} compares the returned
correction against an exact optimum, on every syndrome and in every presentation. \emph{Contract} is
the residual property of Definition~\ref{def:pb} together with the deployable relations of
Table~\ref{tab:relations} --- R1, R2, R3, R3b, R4 and R6, with R5 excluded because
Section~\ref{sec:relations} shows its precondition cannot be settled without the oracle this arm
exists to avoid --- self-referential throughout, and credited only where a relation fires on the
perturbed decoder and not on the baseline. \emph{Weight check} is the seventh, and it is the assertion a
caller can always make when it holds both the returned correction and its own model: the returned
weight must equal $\sum_{e \in \corr} w_e$ under the caller's weights. It needs no oracle and no
second decoder.

This study is confirmatory. Its protocol was committed alone, before the code that acts on it
existed, with data that postdate that commit.

A pre-registered falsifier fired in it. The correction-check arm first detected $215$ of the $219$
correction behaviours, having been given only the base presentation, and the protocol had committed
in advance that such a result meant the arm was misimplemented and could support no cost claim. The
arm was repaired as the protocol prescribes and the study re-run. Every number below comes from the
run in which that arm receives every presentation.

\subsection{The map}
\label{sec:access:map}

\begin{table}[t]
\caption{What each method detects, over 381 distinct non-equivalent behaviours from 13 operators.
Entries are behaviours detected out of the behaviours in that class. The error rate appears in both
of its return paths: \emph{obs.} reads the observable vector the subject returns, \emph{corr.}
recomputes it from the returned correction. Wilson $95\%$ intervals accompany the pooled scores,
which summarise the operator mix specified above.}
\label{tab:accessmap}
\footnotesize
\setlength{\tabcolsep}{4pt}
\begin{tabular}{lrrrrrrrr}
\toprule
& & \multicolumn{2}{c}{error rate} & & & & & \\
\cmidrule(lr){3-4}
perturbed component & live & obs. & corr. & width & observable & correction & contract & weight \\
\midrule
correction, six operations & 219 & \textbf{0} & 176 & \textbf{0} & \textbf{0} & \textbf{219} & \textbf{215} & 139 \\
observable, truncate and pad & 23 & 2 & 0 & \textbf{23} & \textbf{23} & 0 & \textbf{0} & 0 \\
observable, rotate and flip & 18 & 12 & 0 & 0 & 12 & 0 & 12 & 0 \\
call history & 41 & 41 & 41 & 0 & 39 & 41 & 41 & 0 \\
model weights in & 8 & 0 & 0 & 0 & 0 & 2 & 0 & 0 \\
\textbf{returned weight} & \textbf{72} & \textbf{0} & \textbf{0} & \textbf{0} & \textbf{0} & \textbf{0} & \textbf{0} & \textbf{72} \\
\midrule
pooled & 381 & 14.4\% & 57.0\% & 6.0\% & 19.4\% & 68.8\% & 70.3\% & 55.4\% \\
Wilson 95\% & & [11.3, & [51.9, & [4.1, & [15.8, & [63.9, & [65.6, & [50.4, \\
 & & 18.3] & 61.8] & 8.9] & 23.7] & 73.2] & 74.7] & 60.3] \\
\bottomrule
\end{tabular}
\end{table}

Table~\ref{tab:accessmap} has three rows worth reading closely.

\emph{The correction row, and the two paths through it.} Along the observable-native path the answer
is a flat zero. Error-rate benchmarking and exhaustive observable checking detect none of the $219$
behaviours that corrupt the correction, because both score a decoder through its predicted observable
vector and all six operations leave that vector untouched. The zero is exact and narrow, describing
only the two subjects that return an observable vector.

The correction-native path is the one seven of the nine subjects put a caller on, and it splits the
class in two. Of the $219$ behaviours, $176$ move the logical class at the pre-registered nominal
$\alpha = 0.01$, or $173$ under a Bonferroni correction over the whole $1206$-test family, and $43$
do not. The $43$ permute, shift or duplicate edges in a way that changes the set the caller applies
while leaving its observable parity alone. So such a caller catches the $176$ and is blind to the
$43$, and no observable-native benchmark catches any of the $176$. The overlap between the two
error-rate columns on this class is empty.

The class an error rate cannot see is therefore smaller than a single-path measurement suggests and
is not empty, and which number applies to a reader is decided by which object their library hands
back. The correction check detects all $219$ on either path; the contract detects $215$, a
coincidence of arithmetic with the falsifier count above and a different arm on a different cause.
The four it misses are the ones R5 detected alone, for the reason Section~\ref{sec:relations} gives.

The multi-observable population says what the blind fraction depends on. There the observable-native
path again detects $0$ of $98$, the correction check $98$ of $98$, the contract $95$, and the
correction-native path $95$ of $98$. The invisible class falls from $20\%$ at one logical observable
to $3\%$ above it, because two, eight or thirty-two observables give a corrupted correction far fewer
ways to keep every parity intact. The blindness is worst exactly where a memory experiment sits.

\emph{The observable row.} The contract detects none of the $23$ behaviours that truncate or pad the
observable vector, while the width check and the observable check detect all $23$. A relation compares
a decoder with itself, so a corruption applied consistently on every call preserves every relation,
and the residual property never inspects the observable vector at all. The property that would catch
the class is bounded-distance correctness, whose claim covers the width and content of the returned
observable vector wherever a subject exposes one (Definition~\ref{app:def:bd}), and
Section~\ref{sec:measurement:yield} measured its direct declaration yield at $0.0\%$, because it
gates on a capability no library on the panel declares. The panel's silence is what leaves this class unassessable by a capability-gated
contract. Where the model declares more than one observable the offset relation recovers part of it,
$14$ of $31$ on the second population, so the blindness is complete at one observable and partial
above it.

\emph{The weight row.} None of the five methods compared here detects any of the $72$ behaviours that
perturb the returned weight, on either return path, and none detects any of the $36$ on the second population
either. Nor does our contract. The seventh column detects all $72$ and all $36$, and it is one
summation. Add the caller's own weights over the returned edges and compare. It needs no oracle, no second decoder and no ground truth, and our own
certificate verifier had been performing exactly this computation on every witness bundle all along.

This is a gap in the methods compared here, not a property of the interface. The returned weight is not unobservable but
unchecked, which is a different and more actionable thing. A caller holding the correction and the
model can validate it in one line, and none of the five ways compared here does so.

\subsection{Reading the table correctly}
\label{sec:access:reading}

The separations follow from which field each method reads, and we state that with the numbers rather
than after them. The six correction operations change only the correction. Two arms read only the observable vector.
Those zeros are therefore consequences of the construction, not discoveries. What
the table settles is the question answered wrongly by default. A reader who scores a decoder by its
logical error rate has measured its observable prediction, and the correction the caller applies has
gone unexamined.

The pooled percentages summarise the operator mix, in which six operations act on the correction and
four on the observable vector. The two separations and the weight row hold across any mix, being
statements about which field each method reads. That field-access question is what these behaviours
establish, and it is the estimand throughout \cite{papadakis2019mutation}.

Two arms come close enough to invite the wrong summary. The contract's $70.3\%$ against the
correction check's $68.8\%$ is $1.5$ points on our own operator mix with overlapping Wilson
intervals, and the contract loses $0$ to $23$ on the observable class, so neither dominates. Both
detect every correction behaviour the observable-based methods miss. What separates them is that the
correction check needs an exact oracle and the contract needs none, and one solve per canonical
syndrome suffices on the oracle side, because a weight-preserving relabelling carries the optimum.
Section~\ref{sec:cost} measures both sides and reports a factor of four to six, not an order of
magnitude.

\section{What the check costs, and where that decides anything}
\label{sec:cost}

Section~\ref{sec:accessmap} left one difference between the two arms that detect the correction
class. One needs an exact optimum and the other two decode calls. A cost ratio is only meaningful
with its numerator, its denominator and a statement of what each side excludes, so every figure below
carries all three.

\subsection{The oracle is the whole bill, and uniqueness is most of it}
\label{sec:cost:oracle}

Two oracles appear in this paper. One returns a minimum-weight explanation of a syndrome. The other
also proves it unique, which Section~\ref{sec:measurement} needs before it may compare chosen
corrections. The second is a solve-exclude-solve pair, and the second solve is what grows.

\begin{table}[t]
\caption{Oracle cost on the rotated surface code at code capacity, per syndrome and projected over the
benchmark's syndrome set. That set is a deterministic stride-spread selection from the error patterns
of weight at most four, capped per instance, and it is not the instance's complete syndrome space. The
per-syndrome timings come from a 30-syndrome timing sample repeated three times. Recorded with platform, processor, interpreter
and every library version alongside the numbers.}
\label{tab:oraclecost}
\small
\begin{tabular}{lrrrr}
\toprule
$d$ & syndromes in sweep & min-weight, one solve & with uniqueness, two solves & projected over the sweep \\
\midrule
7  & 3797 & 6.66 ms & 28.44 ms & 25.3 s \; / \; 108.0 s \\
11 & 4604 & 7.22 ms & 66.75 ms & 33.2 s \; / \; 307.3 s \\
15 & 4709 & \textbf{7.54 ms} & \textbf{95.05 ms} & \textbf{35.5 s \; / \; 447.6 s} \\
\bottomrule
\end{tabular}
\end{table}

\begin{figure}[t]
\centering
\includegraphics[width=\linewidth]{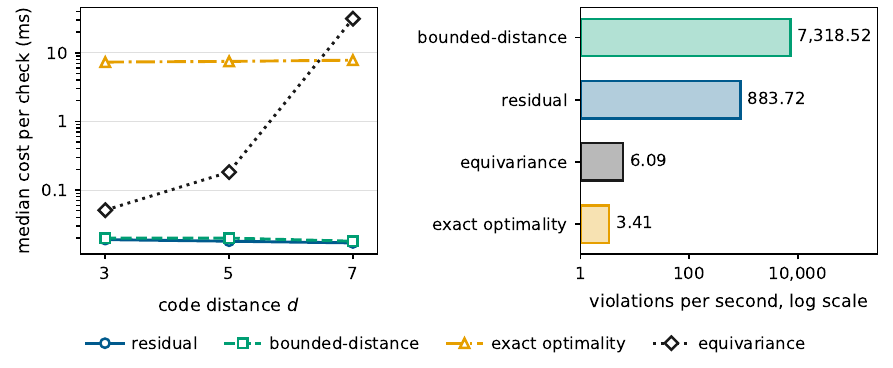}
\caption{Left: median per-check cost against code distance, four properties, logarithmic in cost. The
two solver-free properties are flat at $0.015$ ms, exact optimality is flat at $6.3$ to $6.7$ ms, and
the equivariance property, which needs the uniqueness oracle, rises from $0.041$ ms at $d = 3$ to
$24.80$ ms at $d = 7$, a factor of $605$ between those two cells. Right: violations found per second
per property on the injected panel, on the same colour for each property as the left panel,
logarithmic because the fastest property finds $7319$ per second and the slowest $3.41$, a span of
$2100$.}
\label{fig:cost}
\end{figure}

Table~\ref{tab:oraclecost} says something we did not expect and had predicted the other way. Computing
an exact minimum weight is nearly flat over this range, from $6.66$ to $7.54$ ms. Proving the optimum
unique grows by $3.3\times$ across the same range. Uniqueness verification, not exact optimisation, is the expensive primitive. It bounds the
population of Table~\ref{tab:instances}, not any decoder.

The same split shows up inside the checker. For the two properties that invoke an optimisation
solver, exact optimality and matching-core equivariance, the solver is $97$ to $99\%$ of wall-clock
time and the decode call is under $1.5\%$. The residual property and bounded-distance correctness
invoke none and cost $0.015$ ms flat, unmoved between $d = 3$ and $d = 7$. Bounded-distance correctness relies on a known injected error. It is solver-free, not
ground-truth-free. Table~\ref{tab:properties} separates the two. The cost of this contract is the cost of
being sure, not the cost of decoding, and it is why dispatching an oracle-backed property only where
its gate opens is worth doing: a median exact optimality check costs $6.29$ ms at $d = 3$ against
$0.015$ ms, a factor of $419$ on that instance.

\subsection{Four ratios, each answering a different question}
\label{sec:cost:ratios}

At $d = 15$ the relabelling sweep costs $0.31$ s for pymatching, $0.35$ s for fusion-blossom, and
$0.48$ to $0.54$ s for the two BP-with-OSD configurations, against a single min-weight oracle pass of
$35.5$ s over the same syndromes, which is the comparison Figure~\ref{fig:cost} plots. Four honest
framings follow in Table~\ref{tab:costratios}, two orders of magnitude apart, so naming which one
is in use is not a formality.

\begin{table}[t]
\caption{The oracle-free test against an exact oracle at $d = 15$, over the same syndrome set. The
per-subject figures are for one subject; the panel figures are for eight, which share a single oracle
pass. The uniqueness row is the oracle a study actually pays for when it compares chosen corrections.}
\label{tab:costratios}
\small
\begin{tabular}{llrr}
\toprule
oracle side & test side & oracle cost & speed-up \\
\midrule
min-weight only & one subject & 35.5 s & 26$\times$ to 113$\times$ \\
min-weight only & the eight-subject panel, 5.81 s & 35.5 s & 6.1$\times$ \\
with uniqueness & one subject & 447.6 s & 334$\times$ to 1426$\times$ \\
with uniqueness & the eight-subject panel & 447.6 s & 77$\times$ \\
\bottomrule
\end{tabular}
\end{table}

The panel rows are the conservative ones and belong in any comparison of whole studies, since one
oracle pass serves every subject while each subject pays its own test sweep. The per-subject rows
belong to a caller adding one library to a continuous integration job, who amortises nothing.

\subsection{Where the oracle is the cheaper option}
\label{sec:cost:notthere}

On the population of Section~\ref{sec:accessmap}, the oracle-free contract is the more expensive arm.
The figures below give direct timings for the deployable contract itself. An earlier version
timed a profile that still ran R5 and called the result an upper bound, which answers a different
question. An upper bound cannot say whether the oracle beats the contract that ships, only that it
beats something slower. The study was re-run with R5 out of the relation battery, where
Section~\ref{sec:relations} puts it.

The correction check ran in $15.5$ s and the deployable contract, the residual property with R1, R2,
R3, R3b, R4 and R6, in $69.1$ s, both from that single run. That is a factor of $4.5$ in the
oracle's favour.
R5's share is legible because the correction check barely moves between the two runs, $15.9$ s with
R5 in the battery and $15.5$ s without, so the contract's $81.5$ s to $69.1$ s is R5, about a
seventh. Both runs are from the repaired arm of Section~\ref{sec:access:design}.

The reason the oracle wins here is amortisation. That study built its oracles once per instance and
six per instance for the correction check, one per presentation, which the next subsection shows is
more than it needed. Either way their cost is spread across $381$ behaviours on a $4260$-syndrome
population, which drives the per-behaviour share near zero. The contract pays no oracle at all and
still loses, because it runs the full relation battery with a fresh decoder per syndrome for the
purity relation.

The comparison on the search population turns on how the oracle is amortised. The search of
Section~\ref{sec:hunt} makes $160{,}000$ paired comparisons,
but it does not need $160{,}000$ solves. A relabelling preserves every edge weight, so by
Theorem~\ref{app:thm:equiv} it carries the feasible family and the objective: the minimum weight of a
syndrome is the same in every presentation, and one canonical solve serves them all. It serves every
subject too, since the oracle answers a question about the instance. The honest baseline is one solve
per canonical syndrome, $2 \times 800 = 1600$ of them.

\begin{table}[t]
\caption{The oracle-free check against a correctly amortised exact oracle, on the graphs of
Section~\ref{sec:hunt}. The projected column is the median of ten canonical solves per graph scaled
to that graph's 800 syndromes; the realised column is the wall clock of the complete pass in
Table~\ref{tab:canonopt}. The ratio is computed
against the realised time. Certificate timings are the search's own recorded wall clock for both
subjects.}
\label{tab:amortised}
\small
\begin{tabular}{lrrrrr}
\toprule
instance & detectors & certificate, both subjects & oracle, projected & oracle, realised & ratio \\
\midrule
$d = 11$, circuit-derived & 1320 & 240.7 s & 530 s & \textbf{1062 s} & \textbf{4.4}$\times$ \\
$d = 13$, circuit-derived & 2184 & 533.3 s & 1814 s & \textbf{3171 s} & \textbf{5.9}$\times$ \\
\bottomrule
\end{tabular}
\end{table}

Table~\ref{tab:amortised} settles it. The oracle is affordable on this population. A full canonical
pass is eighteen minutes at $d = 11$ and just under an hour at $d = 13$, and every syndrome but one
was settled inside its budget. On this population the oracle-free check is faster by a factor of four
to six, and the exact oracle follows the search comfortably.

The two oracle columns differ by about a factor of two, and the reason belongs in the table rather
than in a footnote. A median of ten solves is a poor estimator of a total over eight hundred, because
the solve time has a long right tail: the realised pass averages $1.33$ s per syndrome at $d = 11$
against a sampled median of $0.66$ s, and $3.96$ s against $2.27$ s at $d = 13$. Every ratio in the
table is computed from the realised column, and both are shown so that a reader can see which number
each one came from.

The direction of the amortisation runs the same way. The oracle's cost is per syndrome and shared
across every subject and presentation. The certificate's is per subject and grows with the panel, so
on a large enough panel the oracle is the cheaper arrangement. What the panel does not change is that
the certificate needs no solver in the dependency set, asks no uniqueness question, and emits a
witness a reader can check.

\subsection{What this costs a caller who adopts it}
\label{sec:cost:adoption}

The figures above are for a study. A caller integrating the check pays something smaller. Two decode
calls, one relabelling, one parity computation and one weight comparison per syndrome is the whole
per-syndrome obligation, and the relabelling can be drawn once per instance and reused. The benchmark
figure that bounds it is the full sweep at $d = 15$, which runs the base presentation and both
relabelling relations for pymatching in $0.314$ s over $4709$ syndromes, or $0.067$ ms per syndrome
across all three passes, against $7.54$ ms for one exact solve of a single syndrome. Neither the residual property nor the certificate asks the
library to declare anything, which is what makes them the two rows of Table~\ref{tab:properties} that
reach every subject on the panel. Section~\ref{sec:artifact} describes the packaging that turns that
obligation into a dependency and a function call.

\section{Every certificate travels with its own proof}
\label{sec:artifact}

A certificate is worth what a reader can check without our tooling. Each one is therefore exported as
a self-contained bundle, and a second program that shares no code with the checker re-derives the
claim from the bundle alone.

\subsection{What a bundle contains}
\label{sec:art:bundle}

A bundle is one JSON object holding the whole instance and nothing else: the graph as an edge list
with a virtual boundary and a weight per edge, the syndrome, the relabelling as two explicit
permutations, the variant edge list actually handed to the library, both returned corrections in the
index space each was returned in, and the two claimed weights. Provenance travels with it, naming the
study, the instance, the relation, the seed, and the subject's name and pinned version.

Weights are carried as strings, never as JSON numbers. Exact integers are written in decimal, and
IEEE-754 doubles are written in the round-trip hexadecimal form, so a bundle survives serialisation
without losing a bit. Float totals are accumulated with compensated summation, which matters because a
certificate is a statement about two sums being different.

\subsection{The verifier}
\label{sec:art:verifier}

\texttt{deccon-cert} is a standard-library Python program that imports nothing from this project and
nothing from any decoder library. Given a bundle it runs eleven checks in order and stops at the first
failure. Four are structural: the schema is the expected one, the graph is well formed with every
endpoint a detector or the boundary, both permutations are bijections, and the variant edge list is
edge by edge the image of the canonical one under the relabelling. Three concern the inputs and
outputs, in that the syndrome names detectors that exist, both corrections are duplicate-free and in
range, and the second correction transports back through the inverse permutation. The last four are
the claim itself. Both transported corrections have detector boundary exactly the syndrome, the two
claimed weights are the weights the graph gives, the two weights differ, and the side the bundle calls
heavier is the heavier one. Appendix~\ref{app:schema} lists them individually.

The eleventh check is the conclusion of Lemma~\ref{lem:certificate} recomputed from the bundle. A
reader who runs the verifier has not taken our word for the arithmetic, the transport, or the
feasibility of either correction.

\begin{observation}[Independent re-verification]
\label{obs:certkit}
All $639$ certificates were exported and re-verified. $639$ passed and $0$ failed. A positive control
passes all eleven checks, and a negative control built by setting the two weights equal is rejected at
the tenth.
\end{observation}

Observation~\ref{obs:certkit} establishes that each certificate is internally checkable without our
tooling. Section~\ref{sec:relations} specifies the population those certificates come from and
Section~\ref{sec:measurement} reports its results.

\subsection{The rest of the artifact}
\label{sec:art:rest}

The artifact also carries the material needed to reproduce and audit the numbers in this paper. Each
study artifact records the commit that carried the protocol it ran under. The confirmatory protocols
were committed before the code that acts on them. The one exploratory protocol is labelled as such in
Section~\ref{sec:accessmap}, because it was committed together with its data. The evaluation pipeline is frozen as a
manifest of file digests with the interpreter, the platform and every library version, and each freeze
names the freeze it supersedes. The artifact carries the errata file, which records each correction
made during this work and what it changed. The
claim inventory names, for every table in this paper, the artifact it was read from, and a script
re-reads those artifacts and diffs them against the manuscript source.

The check itself installs as a dependency. A caller supplies a detector error model and a decoder, and
receives the residual verdict and any certificate the relabelling produces. Nothing in that path needs
an oracle, a ground-truth error, or a declaration from the library, which is the property that lets it
run in a continuous integration job on every commit.

\subsection{Where the artifact is}
\label{sec:art:where}

The artifact is archived at \url{https://doi.org/10.5281/zenodo.22235839}, which resolves to the
version current at publication. The repository it was built from is at
\url{https://github.com/Mercury0828/deccon-artifact}.

\section{Related work}
\label{sec:related}

This work sits between two literatures that rarely cite each other. One is decoding for quantum error
correction, the other the testing of programs whose correct output nobody can cheaply compute. We take each
in turn and say at the end which parts of our construction are new and which instantiate something
known.

\subsection{Decoders and the libraries this paper measures}
\label{sec:rel:decoders}

Minimum-weight perfect matching decoding for topological codes goes back to Dennis, Kitaev, Landahl
and Preskill \cite{dennis2002topological}, on Kitaev's surface code \cite{kitaev2003faulttolerant}
and in the form most callers use today \cite{fowler2012surface}. The matching problem is Edmonds'
\cite{edmonds1965paths}, and the implementation lineage runs through Blossom V
\cite{kolmogorov2009blossomv} to the two libraries whose exactness declarations
Section~\ref{sec:hunt} puts under load: PyMatching \cite{higgott2022pymatching}, whose version $2$ is
built on sparse blossom \cite{higgott2025sparseblossom}, and fusion-blossom \cite{fusionblossom}. Union-find decoding
\cite{delfosse2021almostlinear} trades exactness for almost-linear time and its weighted
interpretation \cite{wu2022unionfindweighted} is on our panel. Belief propagation with
ordered-statistics post-processing \cite{panteleev2021degenerate,roffe2020decoding} is standard for
quantum LDPC codes \cite{breuckmann2021qldpc,panteleev2022asymptotically,bravyi2024high}, and
search-based decoding \cite{tesseract} is a recent entrant. Colour codes
\cite{bombin2006topological,landahl2011colorcodes} have their own decoder \cite{chromobius}, whose
interface Section~\ref{sec:hunt:chromobius} reports on. What a surface code can and cannot be made to
do without extra resources also bounds what a decoder for it has to handle, and in our prior work
\cite{shen2026magicaxis} we give a conditional no-go for resource-free magic-axis measurement on a
static surface code. Surveys exist
\cite{iolius2024decoding,terhal2015quantum}, and the correlated, circuit-level and single-shot
variants \cite{higgott2023improved,higgott2023singleshot,fuentes2021degeneracy} are what make a
matching graph a lossy view of a detector error model.

Decoding is hard in general \cite{berlekamp1978inherent,iyer2015hardness,hsieh2011nphardness}, which
is why libraries here are approximate by design. Inside the graphlike scope of this paper an exact
solve stays polynomial, and Section~\ref{sec:cost} prices what it costs against the certificate. It is
also why decoders are increasingly evaluated as systems. Recent work covers real-time and
parallel-window decoding
\cite{skoric2023parallel,battistel2023realtime}, hardware decoders \cite{barber2025realtime},
predecoders and hierarchical schemes
\cite{delfosse2020hierarchical,chamberland2023techniques,smith2023local}, and learned decoders
\cite{torlai2017neural,bausch2024alphaqubit}. Every one is a program with an interface, and the
interface is what this paper tests.

\subsection{How a decoder is evaluated today}
\label{sec:rel:evaluation}

The established measure is the logical error rate. Sample errors from a noise model, decode, and count
how often the logical observable is wrong. It is the measure behind threshold estimates
\cite{wang2011surface,stephens2014fault}, behind resource projections
\cite{gidney2021factor,beverland2022assessing}, and behind the
hardware demonstrations that define progress in the field
\cite{google2023suppressing,google2025belowthreshold,bluvstein2024logical}.
Simulation at that scale is what Stim made routine \cite{gidney2021stim}, alongside the
fault-tolerance constructions it is used to study \cite{gidney2021honeycomb,horsman2012surface}.

Benchmarking a quantum computer is itself a studied problem, with volumetric frameworks
\cite{blumekohout2020volumetric}, randomised model circuits \cite{cross2019validating}, randomised
benchmarking \cite{magesan2011scalable}, capability measurement \cite{proctor2022measuring},
application-level suites \cite{lubinski2023application,tomesh2022supermarq} and a certification
literature \cite{eisert2020certification}. What all of these score is an outcome distribution. A
decoder library also returns a correction and a weight, and Section~\ref{sec:accessmap} measures what
scoring only the outcome leaves unexamined.

Measuring a deployed system against what it is documented to do is a stance this group has taken
before in other settings. Prior efforts include measurement-driven control of data-centre thermal
systems under distributional uncertainty \cite{shen2026gridthermal} and a network measurement of how
a policy instrument propagates through coupled electricity markets \cite{shen2026cbam}. Those systems
share no mechanism with a decoder library. What carries over is the order of work, which is to fix
what the system is supposed to guarantee before measuring whether it does.

\subsection{Testing without an oracle}
\label{sec:rel:oracle}

The general problem is old \cite{weyuker1982nontestable} and surveyed \cite{barr2015oracle}.
Metamorphic testing answers it by comparing outputs of related inputs, not any single output
\cite{chen2018metamorphic,segura2016survey,segura2020metamorphic}. It has been applied to compilers
\cite{le2014compiler,donaldson2017automated,chen2020survey} and machine-learning classifiers
\cite{xie2011testing}. It has also been applied to spacecraft data tools
\cite{lindvall2015metamorphic} and security \cite{chen2016metamorphic}. Program comprehension is
another application \cite{zhou2020metamorphic}. Our relations are of
established kinds --- input relabelling, equivalent re-encoding, and invariance of an argmin under
positive scaling. Akgün et al. apply this style to constraint solvers \cite{akgun2018metamorphic},
and Barus et al. a permutation relation to a greedy heuristic for an NP-hard problem
\cite{barus2011testing}.
What our prior-art search locates as this paper's contribution is the QEC-decoder instantiation:
these six presentation dimensions, assembled for that interface.

One piece of prior art is closer to this paper's subject than any of the above, and it is not
metamorphic. Maan and Paler enumerate error and syndrome pairs on small surface codes, run
minimum-weight matching and belief-propagation implementations through a common framework, apply the
returned corrections, and classify logical success or failure against an exhaustive lookup
\cite{maan2023testing}. That is correctness-oracle and differential testing of QEC decoder implementations, and it occupies
the ground a claim to be the first systematic testing of QEC decoders would need. What it is not is
capability-declaration-driven. It has no abstention. It has no gate on what a library declares. It is a study, not a reusable
cross-library contract.

The adjacent tradition is differential and fuzz testing of solvers and compilers
\cite{yang2011finding,brummayer2009fuzzing,winterer2020typeaware,mansur2020detecting},
with delta debugging to make a failure readable \cite{zeller2002simplifying,regehr2012testcase}, the
broader fuzzing literature \cite{godefroid2012sage,bohme2016coverage,manes2021art} and its evaluation
methodology \cite{klees2018evaluating}. Property-based testing supplies the execution model
\cite{claessen2000quickcheck,maciver2019hypothesis}.

\subsection{The closest prior art, and what is different here}
\label{sec:rel:closest}

Paxian and Biere supply the refutation principle behind Lemma~\ref{lem:certificate}, using it for
MaxSAT solvers \cite{paxian2023maxsat}: run a portfolio, record the smallest objective any solver
returned, and treat a claimed optimum above that value as an incorrect optimality result. A competing
feasible solution refutes optimality, with no complete oracle anywhere in the argument.

This paper contributes three distinctions. First, the source of the competing solution differs. Theirs comes from another solver. Ours comes
from the same decoder under a weight- and semantics-preserving relabelling of its input. The test
therefore needs one implementation, not a portfolio. Second, what a firing means differs. For them it is an incorrect optimality result; for us it is a
conformance failure only when the library declared exactness, and otherwise an instance-measured
observation about an approximate algorithm. Third, the two answers live in different index spaces and
must be carried into a common one and shown feasible before their weights may be compared. The proof
itself is one paragraph and we present it as one.

\subsection{Conformance testing and capability declarations}
\label{sec:rel:conformance}

ISO/IEC 9646 supplies the vocabulary of Section~\ref{sec:model}. ISO/IEC 9646
\cite{iso9646part1,iso9646part7} defines conformance testing against a declared capability set: an
implementation states what it supports, and only declared capabilities are tested, with an explicit
NO carrying its own consequences. The formal testing tradition around it is deep
\cite{lee1996principles,tretmans2008model,gaudel1995testing,utting2012taxonomy}.
Outside telecoms the same mechanism is live in cryptographic validation, where a client registers the
algorithms and capabilities it wants validated and the server generates vectors from that
registration \cite{acvp}. This paper contributes a capability vocabulary for decoders, the three
declaration states of Definition~\ref{def:declstate}, and the verdict gating that makes one
measurement a conformance failure for one library and an observation for another.

\subsection{Testing quantum software}
\label{sec:rel:qse}

Quantum software engineering has an emerging testing literature
\cite{miranskyy2019testing,ali2022when}: assertions
\cite{huang2019statistical}, property-based testing in Q\# \cite{honarvar2020property} and in
Qiskit \cite{pontolillo2026qucheck}, coverage criteria \cite{ali2021assessing}, search-based test generation \cite{wang2021generating}, differential
testing across stacks \cite{wang2021qdiff}, an empirical study of platform bugs
\cite{paltenghi2022bugs}, and metamorphic testing of a platform \cite{paltenghi2023morphq} and of
oracle programs \cite{abreu2022metamorphic,nguyen2026metamorphq}. Mutation tooling for quantum
programs exists as well \cite{mendiluze2021muskit,fortunato2022mutation,fortunato2022qmutpy}, as does
verified compilation \cite{hietala2021verified,amy2018towards}. Prior efforts on the systems side of
the same stack include our own work on communication-efficient distributed quantum computing
\cite{zhong2025uniq} and on privacy-preserving distributed quantum computing \cite{zhong2024dpdqc},
where the object of study is likewise a piece of software in the quantum stack and not a quantum
state. This literature targets quantum programs and platforms. A decoder is a classical program in the fault-tolerance stack. Its correctness statements are
combinatorial, not quantum. The literature above does not cover it.

\subsection{Mutation analysis, and what we use it for}
\label{sec:rel:mutation}

Mutation analysis is the standard way to ask what a test suite can see
\cite{demillo1978hints,jia2011analysis,papadakis2019mutation}, and it carries an established
literature on equivalent mutants \cite{kurtz2016analyzing,ammann2014establishing}, on how mutants
relate to real faults \cite{just2014mutants,andrews2005mutation,papadakis2018mutation}, and on what
test-suite metrics predict \cite{inozemtseva2014coverage,zhang2015assertions,li2009experimental}.
Section~\ref{sec:accessmap} uses mutation analysis to establish which field of the returned object
each evaluation method reads. That question has a deterministic answer, and the
mutants only exhibit it.

\subsection{Pre-registration and reproducibility}
\label{sec:rel:prereg}

The design discipline of this paper comes from outside computer science
\cite{nosek2018preregistration,chambers2013registered,munafo2017manifesto}, prompted by the same
concerns about analytic flexibility and selective reporting that motivated it there
\cite{ioannidis2005why,simmons2011false,baker2016reproducibility}. Computer systems research has its
own version of the argument \cite{collberg2016repeatability,krishnamurthi2015real}. Our statistical
apparatus is standard: Wilson intervals \cite{wilson1927probable,brown2001interval}, McNemar's paired
test \cite{mcnemar1947note}, a Jeffreys prior for the rate of a rare event
\cite{jeffreys1946invariant}, and Bonferroni or Benjamini--Hochberg control where a family of
comparisons is reported \cite{dunn1961multiple,benjamini1995controlling}.

\subsection{The position, stated once}
\label{sec:rel:position}

The relation set instantiates established testing techniques. MaxSAT fuzzing supplies the refutation
principle behind the certificate, and conformance testing and cryptographic validation supply the
capability-declaration mechanism. What this paper locates as new is the QEC-decoder instantiation:
these presentations, over these libraries, with a residual computed from the caller's own model, and
with the verdict layer gated on what each library declares. Two of the results are measurements over
a stated population. No returned correction failed to explain its syndrome in the caller's index
space over $717{,}244$ judged results across the panel, and bounded-distance correctness is
undeclared by every library on it.

\section{Conclusion}
\label{sec:conclusion}

A decoder's answer may expose a correction, a predicted observable vector and the weight of the
correction chosen, and none of the three is universal. Logical-error-rate evaluation validates the
induced logical class; it does not validate the edge-level correction beyond that class, and it
cannot validate the weight at all. This paper built a checker for the correction, proved the
certificate it rests on, ran it exhaustively over nine decoder configurations, and reported what came
back.

What came back is mostly a fact about documentation. Of $54$ capability questions a caller might put
to this panel's shipped documentation, four are answered. Bounded-distance correctness is the property
a caller most wants and no library claims it, declines it, or mentions it. A capability-gated
conformance regime therefore reaches $0.0\%$ of the cells where that property could otherwise have
been checked, on a population where its hypotheses hold nearly two-thirds of the time. That number is
a property of what the panel publishes and not of the checker. It measures how far such a regime
reaches into this panel today, and it moves the moment a library writes its guarantee down.

The rest is what a caller can use now. Over $717{,}244$ judged results, no malformed or
syndrome-infeasible correction was observed under the caller's own interpretation of the returned
indices. Neither of the two published exactness declarations was contradicted under any presentation
we constructed, in a pre-registered search of $160{,}000$ paired comparisons on much larger graphs, or
against an exact optimum on the $1599$ canonical syndromes our oracle settled --- and running that
last check is what exposed the oracle's own defect. Approximate solvers are presentation-sensitive: one moved
between the unique optimum and a $26\%$ heavier correction on the same syndrome under a different
detector numbering. Neither of them declares exactness, so that is an instance-measured observation
about an approximate algorithm and not a conformance failure. A fixed-presentation logical-error-rate benchmark cannot reveal that,
and even run under several presentations it stays blind to correction changes that preserve the
logical class. The certificate that finds the weight difference costs two decode calls and an addition, needs no
declaration from the library, and ships as a dependency.

One line of the map is the one we would put in front of a reader building the next decoder, and it
concerns the two components a benchmark does not read. On the correction-native path, which is the path
seven of the nine subjects put a caller on, an error rate detected $176$ of $219$ injected
correction corruptions at the pre-registered $\alpha$, $173$ family-wise, and missed the $43$ that
preserved logical parity; on the observable-native path it detected none of the $219$. The returned
weight was detected by nothing we compared: not by an error rate on either return path, not by
observable-width validation, not by exhaustive observable checking, not by exhaustive correction
checking against an exact oracle, and not by our own contract. That is $0$ of $72$ behaviours on one
population and $0$ of $36$ on the other.

It is not that either component cannot be checked. One summation over the caller's own weights
catches every weight corruption we injected, and this paper's own certificate verifier had been
performing that summation on every witness bundle while the comparison omitted it. A library that
returns a weight is returning a number none of the five methods we compared reads, and a caller can
start reading it today.

\bibliographystyle{ACM-Reference-Format}
\bibliography{refs}

@techreport{iso9646part1,
  author      = {{International Organization for Standardization} and {International Electrotechnical Commission}},
  title       = {Information Technology --- Open Systems Interconnection --- Conformance Testing Methodology and Framework --- Part 1: General Concepts},
  institution = {International Organization for Standardization},
  number      = {ISO/IEC 9646-1:1994},
  year        = {1994},
  url         = {https://www.iso.org/standard/17473.html},
  note        = {Also published as ITU-T Recommendation X.290}
}

@techreport{iso9646part7,
  author      = {{International Organization for Standardization} and {International Electrotechnical Commission}},
  title       = {Information Technology --- Open Systems Interconnection --- Conformance Testing Methodology and Framework --- Part 7: Implementation Conformance Statements},
  institution = {International Organization for Standardization},
  number      = {ISO/IEC 9646-7:1995},
  year        = {1995},
  url         = {https://www.iso.org/standard/3084.html},
  note        = {Also published as ITU-T Recommendation X.296}
}

@misc{acvp,
  author = {{National Institute of Standards and Technology}},
  title  = {Automated Cryptographic Validation Protocol},
  year   = {2026},
  url    = {https://pages.nist.gov/ACVP/},
  note   = {Accessed 2026-08-31}
}

@article{higgott2022pymatching,
  author = {Higgott, Oscar},
  title = {{PyMatching: A Python Package for Decoding Quantum Codes with Minimum-Weight Perfect Matching}},
  journal = {ACM Transactions on Quantum Computing},
  year = {2022},
  volume = {3},
  number = {3},
  pages = {1--16},
  doi = {10.1145/3505637}
}

@article{higgott2025sparseblossom,
  author = {Higgott, Oscar and Gidney, Craig},
  title = {{Sparse Blossom: correcting a million errors per core second with minimum-weight matching}},
  journal = {Quantum},
  year = {2025},
  volume = {9},
  pages = {1600},
  doi = {10.22331/q-2025-01-20-1600}
}

@inproceedings{fusionblossom,
  author    = {Wu, Yue and Zhong, Lin},
  title     = {{Fusion Blossom: Fast {MWPM} Decoders for {QEC}}},
  booktitle = {2023 IEEE International Conference on Quantum Computing and Engineering (QCE)},
  year      = {2023},
  pages     = {928--938},
  publisher = {IEEE},
  doi       = {10.1109/QCE57702.2023.00107}
}

@misc{wu2022unionfindweighted,
  author        = {Wu, Yue and Liyanage, Namitha and Zhong, Lin},
  title         = {An Interpretation of Union-Find Decoder on Weighted Graphs},
  year          = {2022},
  eprint        = {2211.03288},
  archivePrefix = {arXiv},
  primaryClass  = {quant-ph},
  doi           = {10.48550/arXiv.2211.03288}
}

@misc{tesseract,
  author        = {Aghababaie Beni, Laleh and Higgott, Oscar and Shutty, Noah},
  title         = {{Tesseract: A Search-Based Decoder for Quantum Error Correction}},
  year          = {2025},
  eprint        = {2503.10988},
  archivePrefix = {arXiv},
  primaryClass  = {quant-ph},
  doi           = {10.48550/arXiv.2503.10988}
}

@misc{chromobius,
  author        = {Gidney, Craig and Jones, Cody},
  title         = {New Circuits and an Open Source Decoder for the Color Code},
  year          = {2023},
  eprint        = {2312.08813},
  archivePrefix = {arXiv},
  primaryClass  = {quant-ph},
  doi           = {10.48550/arXiv.2312.08813}
}

@article{delfosse2021almostlinear,
  author = {Delfosse, Nicolas and Nickerson, Naomi H.},
  title = {{Almost-linear time decoding algorithm for topological codes}},
  journal = {Quantum},
  year = {2021},
  volume = {5},
  pages = {595},
  doi = {10.22331/q-2021-12-02-595}
}

@article{panteleev2021degenerate,
  author = {Panteleev, Pavel and Kalachev, Gleb},
  title = {{Degenerate Quantum LDPC Codes With Good Finite Length Performance}},
  journal = {Quantum},
  year = {2021},
  volume = {5},
  pages = {585},
  doi = {10.22331/q-2021-11-22-585}
}

@article{roffe2020decoding,
  author  = {Roffe, Joschka and White, David R. and Burton, Simon and Campbell, Earl},
  title   = {Decoding Across the Quantum Low-Density Parity-Check Code Landscape},
  journal = {Physical Review Research},
  year    = {2020},
  volume  = {2},
  number  = {4},
  pages   = {043423},
  doi     = {10.1103/PhysRevResearch.2.043423}
}

@article{higgott2023improved,
  author  = {Higgott, Oscar and Bohdanowicz, Thomas C. and Kubica, Aleksander and Flammia, Steven T. and Campbell, Earl T.},
  title   = {Improved Decoding of Circuit Noise and Fragile Boundaries of Tailored Surface Codes},
  journal = {Physical Review X},
  year    = {2023},
  volume  = {13},
  number  = {3},
  pages   = {031007},
  doi     = {10.1103/PhysRevX.13.031007}
}

@article{edmonds1965paths,
  author  = {Edmonds, Jack},
  title   = {Paths, Trees, and Flowers},
  journal = {Canadian Journal of Mathematics},
  year    = {1965},
  volume  = {17},
  pages   = {449--467},
  doi     = {10.4153/CJM-1965-045-4}
}

@article{kolmogorov2009blossomv,
  author = {Kolmogorov, Vladimir},
  title = {{Blossom V: a new implementation of a minimum cost perfect matching algorithm}},
  journal = {Mathematical Programming Computation},
  year = {2009},
  volume = {1},
  number = {1},
  pages = {43--67},
  doi = {10.1007/s12532-009-0002-8}
}

@article{gidney2021stim,
  author = {Gidney, Craig},
  title = {{Stim: a fast stabilizer circuit simulator}},
  journal = {Quantum},
  year = {2021},
  volume = {5},
  pages = {497},
  doi = {10.22331/q-2021-07-06-497}
}

@article{iolius2024decoding,
  author  = {deMarti iOlius, Antonio and Fuentes, Patricio and Or{\'u}s, Rom{\'a}n and Crespo, Pedro M. and Etxezarreta Martinez, Josu},
  title   = {Decoding Algorithms for Surface Codes},
  journal = {Quantum},
  year    = {2024},
  volume  = {8},
  pages   = {1498},
  doi     = {10.22331/q-2024-10-10-1498}
}

@article{higgott2023singleshot,
  author  = {Higgott, Oscar and Breuckmann, Nikolas P.},
  title   = {Improved Single-Shot Decoding of Higher-Dimensional Hypergraph-Product Codes},
  journal = {PRX Quantum},
  year    = {2023},
  volume  = {4},
  number  = {2},
  pages   = {020332},
  doi     = {10.1103/PRXQuantum.4.020332}
}

@article{skoric2023parallel,
  author  = {Skoric, Luka and Browne, Dan E. and Barnes, Kenton M. and Gillespie, Neil I. and Campbell, Earl T.},
  title   = {Parallel Window Decoding Enables Scalable Fault Tolerant Quantum Computation},
  journal = {Nature Communications},
  year    = {2023},
  volume  = {14},
  pages   = {7040},
  doi     = {10.1038/s41467-023-42482-1}
}

@article{battistel2023realtime,
  author  = {Battistel, Francesco and Chamberland, Christopher and Johar, Kauser and Overwater, Ramon W. J. and Sebastiano, Fabio and Skoric, Luka and Ueno, Yosuke and Usman, Muhammad},
  title   = {Real-Time Decoding for Fault-Tolerant Quantum Computing: Progress, Challenges and Outlook},
  journal = {Nano Futures},
  year    = {2023},
  volume  = {7},
  number  = {3},
  pages   = {032003},
  doi     = {10.1088/2399-1984/aceba6}
}

@article{barber2025realtime,
  author  = {Barber, Ben and Barnes, Kenton M. and Bialas, Tomasz and Bu{\u{g}}dayc{\i}, Okan and Campbell, Earl T. and Gillespie, Neil I. and Johar, Kauser and Rajan, Ram and Richardson, Adam W. and Skoric, Luka and Topal, Canberk and Turner, Mark L. and Ziad, Abbas B.},
  title   = {A Real-Time, Scalable, Fast and Resource-Efficient Decoder for a Quantum Computer},
  journal = {Nature Electronics},
  year    = {2025},
  volume  = {8},
  number  = {1},
  pages   = {84--91},
  doi     = {10.1038/s41928-024-01319-5}
}

@article{bausch2024alphaqubit,
  author  = {Bausch, Johannes and Senior, Andrew W. and Heras, Francisco J. H. and Edlich, Thomas and Davies, Alex and Newman, Michael and Jones, Cody and Satzinger, Kevin and Niu, Murphy Yuezhen and Blackwell, Sam and Holland, George and Kafri, Dvir and Atalaya, Juan and Gidney, Craig and Hassabis, Demis and Boixo, Sergio and Neven, Hartmut and Kohli, Pushmeet},
  title   = {Learning High-Accuracy Error Decoding for Quantum Processors},
  journal = {Nature},
  year    = {2024},
  volume  = {635},
  number  = {8040},
  pages   = {834--840},
  doi     = {10.1038/s41586-024-08148-8}
}

@article{torlai2017neural,
  author  = {Torlai, Giacomo and Melko, Roger G.},
  title   = {Neural Decoder for Topological Codes},
  journal = {Physical Review Letters},
  year    = {2017},
  volume  = {119},
  number  = {3},
  pages   = {030501},
  doi     = {10.1103/PhysRevLett.119.030501}
}

@article{chamberland2023techniques,
  author = {Chamberland, Christopher and Goncalves, Luis and Sivarajah, Prasahnt and Peterson, Eric and Grimberg, Sebastian},
  title = {{Techniques for combining fast local decoders with global decoders under circuit-level noise}},
  journal = {Quantum Science and Technology},
  year = {2023},
  volume = {8},
  number = {4},
  pages = {045011},
  doi = {10.1088/2058-9565/ace64d}
}

@article{smith2023local,
  author  = {Smith, Samuel C. and Brown, Benjamin J. and Bartlett, Stephen D.},
  title   = {Local Predecoder to Reduce the Bandwidth and Latency of Quantum Error Correction},
  journal = {Physical Review Applied},
  year    = {2023},
  volume  = {19},
  number  = {3},
  pages   = {034050},
  doi     = {10.1103/PhysRevApplied.19.034050}
}

@misc{delfosse2020hierarchical,
  author        = {Delfosse, Nicolas},
  title         = {Hierarchical Decoding to Reduce Hardware Requirements for Quantum Computing},
  year          = {2020},
  eprint        = {2001.11427},
  archivePrefix = {arXiv},
  primaryClass  = {quant-ph},
  doi           = {10.48550/arXiv.2001.11427}
}

@article{dennis2002topological,
  author = {Dennis, Eric and Kitaev, Alexei and Landahl, Andrew and Preskill, John},
  title = {{Topological quantum memory}},
  journal = {Journal of Mathematical Physics},
  year = {2002},
  volume = {43},
  number = {9},
  pages = {4452--4505},
  doi = {10.1063/1.1499754}
}

@article{kitaev2003faulttolerant,
  author = {Kitaev, A.Yu.},
  title = {{Fault-tolerant quantum computation by anyons}},
  journal = {Annals of Physics},
  year = {2003},
  volume = {303},
  number = {1},
  pages = {2--30},
  doi = {10.1016/s0003-4916(02)00018-0}
}

@article{fowler2012surface,
  author  = {Fowler, Austin G. and Mariantoni, Matteo and Martinis, John M. and Cleland, Andrew N.},
  title   = {Surface Codes: Towards Practical Large-Scale Quantum Computation},
  journal = {Physical Review A},
  year    = {2012},
  volume  = {86},
  number  = {3},
  pages   = {032324},
  doi     = {10.1103/PhysRevA.86.032324}
}

@article{bombin2006topological,
  author  = {Bombin, H. and Martin-Delgado, M. A.},
  title   = {Topological Quantum Distillation},
  journal = {Physical Review Letters},
  year    = {2006},
  volume  = {97},
  number  = {18},
  pages   = {180501},
  doi     = {10.1103/PhysRevLett.97.180501}
}

@misc{landahl2011colorcodes,
  author        = {Landahl, Andrew J. and Anderson, Jonas T. and Rice, Patrick R.},
  title         = {Fault-Tolerant Quantum Computing with Color Codes},
  year          = {2011},
  eprint        = {1108.5738},
  archivePrefix = {arXiv},
  primaryClass  = {quant-ph},
  doi           = {10.48550/arXiv.1108.5738}
}

@article{wang2011surface,
  author  = {Wang, David S. and Fowler, Austin G. and Hollenberg, Lloyd C. L.},
  title   = {Surface Code Quantum Computing with Error Rates over 1\%},
  journal = {Physical Review A},
  year    = {2011},
  volume  = {83},
  number  = {2},
  pages   = {020302},
  doi     = {10.1103/PhysRevA.83.020302}
}

@article{stephens2014fault,
  author  = {Stephens, Ashley M.},
  title   = {Fault-Tolerant Thresholds for Quantum Error Correction with the Surface Code},
  journal = {Physical Review A},
  year    = {2014},
  volume  = {89},
  number  = {2},
  pages   = {022321},
  doi     = {10.1103/PhysRevA.89.022321}
}

@article{terhal2015quantum,
  author = {Terhal, Barbara M.},
  title = {{Quantum error correction for quantum memories}},
  journal = {Reviews of Modern Physics},
  year = {2015},
  volume = {87},
  number = {2},
  pages = {307--346},
  doi = {10.1103/revmodphys.87.307}
}

@article{breuckmann2021qldpc,
  author  = {Breuckmann, Nikolas P. and Eberhardt, Jens Niklas},
  title   = {Quantum Low-Density Parity-Check Codes},
  journal = {PRX Quantum},
  year    = {2021},
  volume  = {2},
  number  = {4},
  pages   = {040101},
  doi     = {10.1103/PRXQuantum.2.040101}
}

@inproceedings{panteleev2022asymptotically,
  author = {Panteleev, Pavel and Kalachev, Gleb},
  title = {{Asymptotically good Quantum and locally testable classical LDPC codes}},
  booktitle = {Proceedings of the 54th Annual ACM SIGACT Symposium on Theory of Computing},
  year = {2022},
  pages = {375--388},
  publisher = {ACM},
  doi = {10.1145/3519935.3520017}
}

@article{bravyi2024high,
  author = {Bravyi, Sergey and Cross, Andrew W. and Gambetta, Jay M. and Maslov, Dmitri and Rall, Patrick and Yoder, Theodore J.},
  title = {{High-threshold and low-overhead fault-tolerant quantum memory}},
  journal = {Nature},
  year = {2024},
  volume = {627},
  number = {8005},
  pages = {778--782},
  doi = {10.1038/s41586-024-07107-7}
}

@article{horsman2012surface,
  author = {Horsman, Dominic and Fowler, Austin G and Devitt, Simon and Meter, Rodney Van},
  title = {{Surface code quantum computing by lattice surgery}},
  journal = {New Journal of Physics},
  year = {2012},
  volume = {14},
  number = {12},
  pages = {123011},
  doi = {10.1088/1367-2630/14/12/123011}
}

@misc{beverland2022assessing,
  author        = {Beverland, Michael E. and Murali, Prakash and Troyer, Matthias and Svore, Krysta M. and Hoefler, Torsten and Kliuchnikov, Vadym and Low, Guang Hao and Soeken, Mathias and Sundaram, Aarthi and Vaschillo, Alexander},
  title         = {Assessing Requirements to Scale to Practical Quantum Advantage},
  year          = {2022},
  eprint        = {2211.07629},
  archivePrefix = {arXiv},
  primaryClass  = {quant-ph},
  doi           = {10.48550/arXiv.2211.07629}
}

@article{gidney2021factor,
  author  = {Gidney, Craig and Eker{\aa}, Martin},
  title   = {How to Factor 2048 Bit {RSA} Integers in 8 Hours Using 20 Million Noisy Qubits},
  journal = {Quantum},
  year    = {2021},
  volume  = {5},
  pages   = {433},
  doi     = {10.22331/q-2021-04-15-433}
}

@article{gidney2021honeycomb,
  author = {Gidney, Craig and Newman, Michael and Fowler, Austin and Broughton, Michael},
  title = {{A Fault-Tolerant Honeycomb Memory}},
  journal = {Quantum},
  year = {2021},
  volume = {5},
  pages = {605},
  doi = {10.22331/q-2021-12-20-605}
}

@article{google2023suppressing,
  author  = {{Google Quantum AI}},
  title   = {Suppressing Quantum Errors by Scaling a Surface Code Logical Qubit},
  journal = {Nature},
  year    = {2023},
  volume  = {614},
  number  = {7949},
  pages   = {676--681},
  doi     = {10.1038/s41586-022-05434-1}
}

@article{google2025belowthreshold,
  author  = {{Google Quantum AI and Collaborators}},
  title   = {Quantum Error Correction Below the Surface Code Threshold},
  journal = {Nature},
  year    = {2025},
  volume  = {638},
  number  = {8052},
  pages   = {920--926},
  doi     = {10.1038/s41586-024-08449-y}
}

@article{bluvstein2024logical,
  author  = {Bluvstein, Dolev and Evered, Simon J. and Geim, Alexandra A. and Li, Sophie H. and Zhou, Hengyun and Manovitz, Tom and Ebadi, Sepehr and Cain, Madelyn and Kalinowski, Marcin and Hangleiter, Dominik and Bonilla Ataides, J. Pablo and Maskara, Nishad and others},
  title   = {Logical Quantum Processor Based on Reconfigurable Atom Arrays},
  journal = {Nature},
  year    = {2024},
  volume  = {626},
  number  = {7997},
  pages   = {58--65},
  doi     = {10.1038/s41586-023-06927-3}
}

@article{berlekamp1978inherent,
  author = {Berlekamp, E. and McEliece, R. and van Tilborg, H.},
  title = {{On the inherent intractability of certain coding problems (Corresp.)}},
  journal = {IEEE Transactions on Information Theory},
  year = {1978},
  volume = {24},
  number = {3},
  pages = {384--386},
  doi = {10.1109/tit.1978.1055873}
}

@article{iyer2015hardness,
  author = {Iyer, Pavithran and Poulin, David},
  title = {{Hardness of Decoding Quantum Stabilizer Codes}},
  journal = {IEEE Transactions on Information Theory},
  year = {2015},
  volume = {61},
  number = {9},
  pages = {5209--5223},
  doi = {10.1109/tit.2015.2422294}
}

@article{hsieh2011nphardness,
  author  = {Hsieh, Min-Hsiu and Le Gall, Fran{\c{c}}ois},
  title   = {{NP}-Hardness of Decoding Quantum Error-Correction Codes},
  journal = {Physical Review A},
  year    = {2011},
  volume  = {83},
  number  = {5},
  pages   = {052331},
  doi     = {10.1103/PhysRevA.83.052331}
}

@article{fuentes2021degeneracy,
  author = {Fuentes, Patricio and Etxezarreta Martinez, Josu and Crespo, Pedro M. and Garcia-Frias, Javier},
  title = {{Degeneracy and Its Impact on the Decoding of Sparse Quantum Codes}},
  journal = {IEEE Access},
  year = {2021},
  volume = {9},
  pages = {89093--89119},
  doi = {10.1109/access.2021.3089829}
}

@article{barr2015oracle,
  author = {Barr, Earl T. and Harman, Mark and McMinn, Phil and Shahbaz, Muzammil and Yoo, Shin},
  title = {{The Oracle Problem in Software Testing: A Survey}},
  journal = {IEEE Transactions on Software Engineering},
  year = {2015},
  volume = {41},
  number = {5},
  pages = {507--525},
  doi = {10.1109/tse.2014.2372785}
}

@article{weyuker1982nontestable,
  author = {Weyuker, E. J.},
  title = {{On Testing Non-Testable Programs}},
  journal = {The Computer Journal},
  year = {1982},
  volume = {25},
  number = {4},
  pages = {465--470},
  doi = {10.1093/comjnl/25.4.465}
}

@article{chen2018metamorphic,
  author = {Chen, Tsong Yueh and Kuo, Fei-Ching and Liu, Huai and Poon, Pak-Lok and Towey, Dave and Tse, T. H. and Zhou, Zhi Quan},
  title = {{Metamorphic Testing}},
  journal = {ACM Computing Surveys},
  year = {2018},
  volume = {51},
  number = {1},
  pages = {1--27},
  doi = {10.1145/3143561}
}

@article{segura2016survey,
  author = {Segura, Sergio and Fraser, Gordon and Sanchez, Ana B. and Ruiz-Cortes, Antonio},
  title = {{A Survey on Metamorphic Testing}},
  journal = {IEEE Transactions on Software Engineering},
  year = {2016},
  volume = {42},
  number = {9},
  pages = {805--824},
  doi = {10.1109/tse.2016.2532875}
}

@article{segura2020metamorphic,
  author = {Segura, Sergio and Towey, Dave and Zhou, Zhi Quan and Chen, Tsong Yueh},
  title = {{Metamorphic Testing: Testing the Untestable}},
  journal = {IEEE Software},
  year = {2020},
  volume = {37},
  number = {3},
  pages = {46--53},
  doi = {10.1109/ms.2018.2875968}
}

@article{zhou2020metamorphic,
  author = {Zhou, Zhi Quan and Sun, Liqun and Chen, Tsong Yueh and Towey, Dave},
  title = {{Metamorphic Relations for Enhancing System Understanding and Use}},
  journal = {IEEE Transactions on Software Engineering},
  year = {2020},
  volume = {46},
  number = {10},
  pages = {1120--1154},
  doi = {10.1109/tse.2018.2876433}
}

@article{chen2016metamorphic,
  author = {Chen, Tsong Yueh and Kuo, Fei-Ching and Ma, Wenjuan and Susilo, Willy and Towey, Dave and Voas, Jeffrey and Zhou, Zhi Quan},
  title = {{Metamorphic Testing for Cybersecurity}},
  journal = {Computer},
  year = {2016},
  volume = {49},
  number = {6},
  pages = {48--55},
  doi = {10.1109/mc.2016.176}
}

@article{xie2011testing,
  author = {Xie, Xiaoyuan and Ho, Joshua W.K. and Murphy, Christian and Kaiser, Gail and Xu, Baowen and Chen, Tsong Yueh},
  title = {{Testing and validating machine learning classifiers by metamorphic testing}},
  journal = {Journal of Systems and Software},
  year = {2011},
  volume = {84},
  number = {4},
  pages = {544--558},
  doi = {10.1016/j.jss.2010.11.920}
}

@inproceedings{lindvall2015metamorphic,
  author = {Lindvall, Mikael and Ganesan, Dharmalingam and Ardal, Ragnar and Wiegand, Robert E.},
  title = {{Metamorphic Model-Based Testing Applied on NASA DAT -- An Experience Report}},
  booktitle = {2015 IEEE/ACM 37th IEEE International Conference on Software Engineering},
  year = {2015},
  pages = {129--138},
  publisher = {IEEE},
  doi = {10.1109/icse.2015.348}
}

@article{donaldson2017automated,
  author = {Donaldson, Alastair F. and Evrard, Hugues and Lascu, Andrei and Thomson, Paul},
  title = {{Automated testing of graphics shader compilers}},
  journal = {Proceedings of the ACM on Programming Languages},
  year = {2017},
  volume = {1},
  number = {OOPSLA},
  pages = {1--29},
  doi = {10.1145/3133917}
}

@article{le2014compiler,
  author = {Le, Vu and Afshari, Mehrdad and Su, Zhendong},
  title = {{Compiler validation via equivalence modulo inputs}},
  journal = {ACM SIGPLAN Notices},
  year = {2014},
  volume = {49},
  number = {6},
  pages = {216--226},
  doi = {10.1145/2666356.2594334}
}

@inproceedings{yang2011finding,
  author    = {Yang, Xuejun and Chen, Yang and Eide, Eric and Regehr, John},
  title     = {Finding and Understanding Bugs in {C} Compilers},
  booktitle = {Proceedings of the 32nd ACM SIGPLAN Conference on Programming Language Design and Implementation},
  year      = {2011},
  pages     = {283--294},
  publisher = {ACM},
  doi       = {10.1145/1993498.1993532}
}

@article{chen2020survey,
  author = {Chen, Junjie and Patra, Jibesh and Pradel, Michael and Xiong, Yingfei and Zhang, Hongyu and Hao, Dan and Zhang, Lu},
  title = {{A Survey of Compiler Testing}},
  journal = {ACM Computing Surveys},
  year = {2020},
  volume = {53},
  number = {1},
  pages = {1--36},
  doi = {10.1145/3363562}
}

@inproceedings{brummayer2009fuzzing,
  author = {Brummayer, Robert and Biere, Armin},
  title = {{Fuzzing and delta-debugging SMT solvers}},
  booktitle = {Proceedings of the 7th International Workshop on Satisfiability Modulo Theories},
  year = {2009},
  pages = {1--5},
  publisher = {ACM},
  doi = {10.1145/1670412.1670413}
}

@article{paxian2023maxsat,
  author  = {Paxian, Tobias and Biere, Armin},
  title   = {{MaxSAT} Fuzzing and Delta Debugging},
  journal = {Journal of Artificial Intelligence Research},
  year    = {2026},
  volume  = {85},
  doi     = {10.1613/jair.1.19716}
}

@article{winterer2020typeaware,
  author = {Winterer, Dominik and Zhang, Chengyu and Su, Zhendong},
  title = {{On the unusual effectiveness of type-aware operator mutations for testing SMT solvers}},
  journal = {Proceedings of the ACM on Programming Languages},
  year = {2020},
  volume = {4},
  number = {OOPSLA},
  pages = {1--25},
  doi = {10.1145/3428261}
}

@inproceedings{mansur2020detecting,
  author = {Mansur, Muhammad Numair and Christakis, Maria and W{\"u}stholz, Valentin and Zhang, Fuyuan},
  title = {{Detecting critical bugs in SMT solvers using blackbox mutational fuzzing}},
  booktitle = {Proceedings of the 28th ACM Joint Meeting on European Software Engineering Conference and Symposium on the Foundations of Software Engineering},
  year = {2020},
  pages = {701--712},
  publisher = {ACM},
  doi = {10.1145/3368089.3409763}
}

@article{zeller2002simplifying,
  author = {Zeller, A. and Hildebrandt, R.},
  title = {{Simplifying and isolating failure-inducing input}},
  journal = {IEEE Transactions on Software Engineering},
  year = {2002},
  volume = {28},
  number = {2},
  pages = {183--200},
  doi = {10.1109/32.988498}
}

@inproceedings{regehr2012testcase,
  author    = {Regehr, John and Chen, Yang and Cuoq, Pascal and Eide, Eric and Ellison, Chucky and Yang, Xuejun},
  title     = {Test-Case Reduction for {C} Compiler Bugs},
  booktitle = {Proceedings of the 33rd ACM SIGPLAN Conference on Programming Language Design and Implementation},
  year      = {2012},
  pages     = {335--346},
  publisher = {ACM},
  doi       = {10.1145/2254064.2254104}
}

@inproceedings{claessen2000quickcheck,
  author = {Claessen, Koen and Hughes, John},
  title = {{QuickCheck}},
  booktitle = {Proceedings of the fifth ACM SIGPLAN international conference on Functional programming},
  year = {2000},
  pages = {268--279},
  publisher = {ACM},
  doi = {10.1145/351240.351266}
}

@article{maciver2019hypothesis,
  author = {MacIver, David and Hatfield-Dodds, Zac and Contributors, Many},
  title = {{Hypothesis: A new approach to property-based testing}},
  journal = {Journal of Open Source Software},
  year = {2019},
  volume = {4},
  number = {43},
  pages = {1891},
  doi = {10.21105/joss.01891}
}

@article{godefroid2012sage,
  author = {Godefroid, Patrice and Levin, Michael Y. and Molnar, David},
  title = {{SAGE: Whitebox Fuzzing for Security Testing}},
  journal = {Queue},
  year = {2012},
  volume = {10},
  number = {1},
  pages = {20--27},
  doi = {10.1145/2090147.2094081}
}

@inproceedings{klees2018evaluating,
  author = {Klees, George and Ruef, Andrew and Cooper, Benji and Wei, Shiyi and Hicks, Michael},
  title = {{Evaluating Fuzz Testing}},
  booktitle = {Proceedings of the 2018 ACM SIGSAC Conference on Computer and Communications Security},
  year = {2018},
  pages = {2123--2138},
  publisher = {ACM},
  doi = {10.1145/3243734.3243804}
}

@inproceedings{bohme2016coverage,
  author = {B{\"o}hme, Marcel and Pham, Van-Thuan and Roychoudhury, Abhik},
  title = {{Coverage-based Greybox Fuzzing as Markov Chain}},
  booktitle = {Proceedings of the 2016 ACM SIGSAC Conference on Computer and Communications Security},
  year = {2016},
  pages = {1032--1043},
  publisher = {ACM},
  doi = {10.1145/2976749.2978428}
}

@article{manes2021art,
  author = {Manes, Valentin J.M. and Han, HyungSeok and Han, Choongwoo and Cha, Sang Kil and Egele, Manuel and Schwartz, Edward J. and Woo, Maverick},
  title = {{The Art, Science, and Engineering of Fuzzing: A Survey}},
  journal = {IEEE Transactions on Software Engineering},
  year = {2021},
  volume = {47},
  number = {11},
  pages = {2312--2331},
  doi = {10.1109/tse.2019.2946563}
}

@article{demillo1978hints,
  author = {DeMillo, R.A. and Lipton, R.J. and Sayward, F.G.},
  title = {{Hints on Test Data Selection: Help for the Practicing Programmer}},
  journal = {Computer},
  year = {1978},
  volume = {11},
  number = {4},
  pages = {34--41},
  doi = {10.1109/c-m.1978.218136}
}

@article{jia2011analysis,
  author = {Jia, Yue and Harman, Mark},
  title = {{An Analysis and Survey of the Development of Mutation Testing}},
  journal = {IEEE Transactions on Software Engineering},
  year = {2011},
  volume = {37},
  number = {5},
  pages = {649--678},
  doi = {10.1109/tse.2010.62}
}

@incollection{papadakis2019mutation,
  author    = {Papadakis, Mike and Kintis, Marinos and Zhang, Jie M. and Jia, Yue and Le Traon, Yves and Harman, Mark},
  title     = {Mutation Testing Advances: An Analysis and Survey},
  booktitle = {Advances in Computers},
  editor    = {Memon, Atif M.},
  volume    = {112},
  year      = {2019},
  pages     = {275--378},
  publisher = {Elsevier},
  doi       = {10.1016/bs.adcom.2018.03.015}
}

@inproceedings{just2014mutants,
  author = {Just, Ren{\'e} and Jalali, Darioush and Inozemtseva, Laura and Ernst, Michael D. and Holmes, Reid and Fraser, Gordon},
  title = {{Are mutants a valid substitute for real faults in software testing?}},
  booktitle = {Proceedings of the 22nd ACM SIGSOFT International Symposium on Foundations of Software Engineering},
  year = {2014},
  pages = {654--665},
  publisher = {ACM},
  doi = {10.1145/2635868.2635929}
}

@inproceedings{andrews2005mutation,
  author    = {Andrews, James H. and Briand, Lionel C. and Labiche, Yvan},
  title     = {Is Mutation an Appropriate Tool for Testing Experiments?},
  booktitle = {Proceedings of the 27th International Conference on Software Engineering (ICSE 2005)},
  year      = {2005},
  pages     = {402--411},
  publisher = {IEEE},
  doi       = {10.1109/ICSE.2005.1553583}
}

@inproceedings{papadakis2018mutation,
  author = {Papadakis, Mike and Shin, Donghwan and Yoo, Shin and Bae, Doo-Hwan},
  title = {{Are mutation scores correlated with real fault detection?}},
  booktitle = {Proceedings of the 40th International Conference on Software Engineering},
  year = {2018},
  pages = {537--548},
  publisher = {ACM},
  doi = {10.1145/3180155.3180183}
}

@inproceedings{kurtz2016analyzing,
  author = {Kurtz, Bob and Ammann, Paul and Offutt, Jeff and Delamaro, M{\'a}rcio E. and Kurtz, Mariet and G{\"o}k{\c{c}}e, Nida},
  title = {{Analyzing the validity of selective mutation with dominator mutants}},
  booktitle = {Proceedings of the 2016 24th ACM SIGSOFT International Symposium on Foundations of Software Engineering},
  year = {2016},
  pages = {571--582},
  publisher = {ACM},
  doi = {10.1145/2950290.2950322}
}

@inproceedings{ammann2014establishing,
  author = {Ammann, Paul and Delamaro, Marcio Eduardo and Offutt, Jeff},
  title = {{Establishing Theoretical Minimal Sets of Mutants}},
  booktitle = {2014 IEEE Seventh International Conference on Software Testing, Verification and Validation},
  year = {2014},
  pages = {21--30},
  publisher = {IEEE},
  doi = {10.1109/icst.2014.13}
}

@inproceedings{li2009experimental,
  author = {Li, Nan and Praphamontripong, Upsorn and Offutt, Jeff},
  title = {{An Experimental Comparison of Four Unit Test Criteria: Mutation, Edge-Pair, All-Uses and Prime Path Coverage}},
  booktitle = {2009 International Conference on Software Testing, Verification, and Validation Workshops},
  year = {2009},
  pages = {220--229},
  publisher = {IEEE},
  doi = {10.1109/icstw.2009.30}
}

@inproceedings{inozemtseva2014coverage,
  author = {Inozemtseva, Laura and Holmes, Reid},
  title = {{Coverage is not strongly correlated with test suite effectiveness}},
  booktitle = {Proceedings of the 36th International Conference on Software Engineering},
  year = {2014},
  pages = {435--445},
  publisher = {ACM},
  doi = {10.1145/2568225.2568271}
}

@inproceedings{zhang2015assertions,
  author = {Zhang, Yucheng and Mesbah, Ali},
  title = {{Assertions are strongly correlated with test suite effectiveness}},
  booktitle = {Proceedings of the 2015 10th Joint Meeting on Foundations of Software Engineering},
  year = {2015},
  pages = {214--224},
  publisher = {ACM},
  doi = {10.1145/2786805.2786858}
}

@article{lee1996principles,
  author = {Lee, D. and Yannakakis, M.},
  title = {{Principles and methods of testing finite state machines-a survey}},
  journal = {Proceedings of the IEEE},
  year = {1996},
  volume = {84},
  number = {8},
  pages = {1090--1123},
  doi = {10.1109/5.533956}
}

@incollection{tretmans2008model,
  author    = {Tretmans, Jan},
  title     = {Model Based Testing with Labelled Transition Systems},
  booktitle = {Formal Methods and Testing},
  editor    = {Hierons, Robert M. and Bowen, Jonathan P. and Harman, Mark},
  series    = {Lecture Notes in Computer Science},
  volume    = {4949},
  year      = {2008},
  pages     = {1--38},
  publisher = {Springer},
  doi       = {10.1007/978-3-540-78917-8_1}
}

@inproceedings{gaudel1995testing,
  author    = {Gaudel, Marie-Claude},
  title     = {Testing Can Be Formal, Too},
  booktitle = {TAPSOFT '95: Theory and Practice of Software Development},
  editor    = {Mosses, Peter D. and Nielsen, Mogens and Schwartzbach, Michael I.},
  series    = {Lecture Notes in Computer Science},
  volume    = {915},
  year      = {1995},
  pages     = {82--96},
  publisher = {Springer},
  doi       = {10.1007/3-540-59293-8_188}
}

@article{utting2012taxonomy,
  author  = {Utting, Mark and Pretschner, Alexander and Legeard, Bruno},
  title   = {A Taxonomy of Model-Based Testing Approaches},
  journal = {Software Testing, Verification and Reliability},
  year    = {2012},
  volume  = {22},
  number  = {5},
  pages   = {297--312},
  doi     = {10.1002/stvr.456}
}

@inproceedings{miranskyy2019testing,
  author = {Miranskyy, Andriy and Zhang, Lei},
  title = {{On Testing Quantum Programs}},
  booktitle = {2019 IEEE/ACM 41st International Conference on Software Engineering: New Ideas and Emerging Results (ICSE-NIER)},
  year = {2019},
  pages = {57--60},
  publisher = {IEEE},
  doi = {10.1109/icse-nier.2019.00023}
}

@article{ali2022when,
  author = {Ali, Shaukat and Yue, Tao and Abreu, Rui},
  title = {{When software engineering meets quantum computing}},
  journal = {Communications of the ACM},
  year = {2022},
  volume = {65},
  number = {4},
  pages = {84--88},
  doi = {10.1145/3512340}
}

@inproceedings{huang2019statistical,
  author = {Huang, Yipeng and Martonosi, Margaret},
  title = {{Statistical assertions for validating patterns and finding bugs in quantum programs}},
  booktitle = {Proceedings of the 46th International Symposium on Computer Architecture},
  year = {2019},
  pages = {541--553},
  publisher = {ACM},
  doi = {10.1145/3307650.3322213}
}

@inproceedings{honarvar2020property,
  author = {Honarvar, Shahin and Mousavi, Mohammad Reza and Nagarajan, Rajagopal},
  title = {{Property-based Testing of Quantum Programs in Q\#}},
  booktitle = {Proceedings of the IEEE/ACM 42nd International Conference on Software Engineering Workshops},
  year = {2020},
  pages = {430--435},
  publisher = {ACM},
  doi = {10.1145/3387940.3391459}
}

@inproceedings{ali2021assessing,
  author = {Ali, Shaukat and Arcaini, Paolo and Wang, Xinyi and Yue, Tao},
  title = {{Assessing the Effectiveness of Input and Output Coverage Criteria for Testing Quantum Programs}},
  booktitle = {2021 14th IEEE Conference on Software Testing, Verification and Validation (ICST)},
  year = {2021},
  pages = {13--23},
  publisher = {IEEE},
  doi = {10.1109/icst49551.2021.00014}
}

@inproceedings{wang2021generating,
  author    = {Wang, Xinyi and Arcaini, Paolo and Yue, Tao and Ali, Shaukat},
  title     = {Generating Failing Test Suites for Quantum Programs With Search},
  booktitle = {Search-Based Software Engineering},
  editor    = {O'Reilly, Una-May and Devroey, Xavier},
  series    = {Lecture Notes in Computer Science},
  volume    = {12914},
  year      = {2021},
  pages     = {9--25},
  publisher = {Springer},
  doi       = {10.1007/978-3-030-88106-1_2}
}

@article{paltenghi2022bugs,
  author = {Paltenghi, Matteo and Pradel, Michael},
  title = {{Bugs in Quantum computing platforms: an empirical study}},
  journal = {Proceedings of the ACM on Programming Languages},
  year = {2022},
  volume = {6},
  number = {OOPSLA1},
  pages = {1--27},
  doi = {10.1145/3527330}
}

@inproceedings{paltenghi2023morphq,
  author = {Paltenghi, Matteo and Pradel, Michael},
  title = {{MorphQ: Metamorphic Testing of the Qiskit Quantum Computing Platform}},
  booktitle = {2023 IEEE/ACM 45th International Conference on Software Engineering (ICSE)},
  year = {2023},
  pages = {2413--2424},
  publisher = {IEEE},
  doi = {10.1109/icse48619.2023.00202}
}

@inproceedings{wang2021qdiff,
  author = {Wang, Jiyuan and Zhang, Qian and Xu, Guoqing Harry and Kim, Miryung},
  title = {{QDiff: Differential Testing of Quantum Software Stacks}},
  booktitle = {2021 36th IEEE/ACM International Conference on Automated Software Engineering (ASE)},
  year = {2021},
  pages = {692--704},
  publisher = {IEEE},
  doi = {10.1109/ase51524.2021.9678792}
}

@inproceedings{mendiluze2021muskit,
  author = {Mendiluze, Enaut and Ali, Shaukat and Arcaini, Paolo and Yue, Tao},
  title = {{Muskit: A Mutation Analysis Tool for Quantum Software Testing}},
  booktitle = {2021 36th IEEE/ACM International Conference on Automated Software Engineering (ASE)},
  year = {2021},
  pages = {1266--1270},
  publisher = {IEEE},
  doi = {10.1109/ase51524.2021.9678563}
}

@article{fortunato2022mutation,
  author  = {Fortunato, Daniel and Campos, Jos{\'e} and Abreu, Rui},
  title   = {Mutation Testing of Quantum Programs: A Case Study With {Qiskit}},
  journal = {IEEE Transactions on Quantum Engineering},
  year    = {2022},
  volume  = {3},
  pages   = {1--17},
  doi     = {10.1109/TQE.2022.3195061}
}

@inproceedings{fortunato2022qmutpy,
  author = {Fortunato, Daniel and Campos, Jos{\'e} and Abreu, Rui},
  title = {{QMutPy: a mutation testing tool for Quantum algorithms and applications in Qiskit}},
  booktitle = {Proceedings of the 31st ACM SIGSOFT International Symposium on Software Testing and Analysis},
  year = {2022},
  pages = {797--800},
  publisher = {ACM},
  doi = {10.1145/3533767.3543296}
}

@article{hietala2021verified,
  author = {Hietala, Kesha and Rand, Robert and Hung, Shih-Han and Wu, Xiaodi and Hicks, Michael},
  title = {{A verified optimizer for Quantum circuits}},
  journal = {Proceedings of the ACM on Programming Languages},
  year = {2021},
  volume = {5},
  number = {POPL},
  pages = {1--29},
  doi = {10.1145/3434318}
}

@article{amy2018towards,
  author = {Amy, Matthew},
  title = {{Towards Large-scale Functional Verification of Universal Quantum Circuits}},
  journal = {Electronic Proceedings in Theoretical Computer Science},
  year = {2019},
  volume = {287},
  pages = {1--21},
  doi = {10.4204/eptcs.287.1}
}

@article{cross2019validating,
  author  = {Cross, Andrew W. and Bishop, Lev S. and Sheldon, Sarah and Nation, Paul D. and Gambetta, Jay M.},
  title   = {Validating Quantum Computers Using Randomized Model Circuits},
  journal = {Physical Review A},
  year    = {2019},
  volume  = {100},
  number  = {3},
  pages   = {032328},
  doi     = {10.1103/PhysRevA.100.032328}
}

@article{blumekohout2020volumetric,
  author = {Blume-Kohout, Robin and Young, Kevin C.},
  title = {{A volumetric framework for quantum computer benchmarks}},
  journal = {Quantum},
  year = {2020},
  volume = {4},
  pages = {362},
  doi = {10.22331/q-2020-11-15-362}
}

@article{lubinski2023application,
  author = {Lubinski, Thomas and Johri, Sonika and Varosy, Paul and Coleman, Jeremiah and Zhao, Luning and Necaise, Jason and Baldwin, Charles H. and Mayer, Karl and Proctor, Timothy},
  title = {{Application-Oriented Performance Benchmarks for Quantum Computing}},
  journal = {IEEE Transactions on Quantum Engineering},
  year = {2023},
  volume = {4},
  pages = {1--32},
  doi = {10.1109/tqe.2023.3253761}
}

@inproceedings{tomesh2022supermarq,
  author = {Tomesh, Teague and Gokhale, Pranav and Omole, Victory and Ravi, Gokul Subramanian and Smith, Kaitlin N. and Viszlai, Joshua and Wu, Xin-Chuan and Hardavellas, Nikos and Martonosi, Margaret R. and Chong, Frederic T.},
  title = {{SupermarQ: A Scalable Quantum Benchmark Suite}},
  booktitle = {2022 IEEE International Symposium on High-Performance Computer Architecture (HPCA)},
  year = {2022},
  pages = {587--603},
  publisher = {IEEE},
  doi = {10.1109/hpca53966.2022.00050}
}

@article{proctor2022measuring,
  author  = {Proctor, Timothy and Rudinger, Kenneth and Young, Kevin and Nielsen, Erik and Blume-Kohout, Robin},
  title   = {Measuring the Capabilities of Quantum Computers},
  journal = {Nature Physics},
  year    = {2022},
  volume  = {18},
  number  = {1},
  pages   = {75--79},
  doi     = {10.1038/s41567-021-01409-7}
}

@article{eisert2020certification,
  author = {Eisert, Jens and Hangleiter, Dominik and Walk, Nathan and Roth, Ingo and Markham, Damian and Parekh, Rhea and Chabaud, Ulysse and Kashefi, Elham},
  title = {{Quantum certification and benchmarking}},
  journal = {Nature Reviews Physics},
  year = {2020},
  volume = {2},
  number = {7},
  pages = {382--390},
  doi = {10.1038/s42254-020-0186-4}
}

@article{magesan2011scalable,
  author  = {Magesan, Easwar and Gambetta, J. M. and Emerson, Joseph},
  title   = {Scalable and Robust Randomized Benchmarking of Quantum Processes},
  journal = {Physical Review Letters},
  year    = {2011},
  volume  = {106},
  number  = {18},
  pages   = {180504},
  doi     = {10.1103/PhysRevLett.106.180504}
}

@article{collberg2016repeatability,
  author = {Collberg, Christian and Proebsting, Todd A.},
  title = {{Repeatability in computer systems research}},
  journal = {Communications of the ACM},
  year = {2016},
  volume = {59},
  number = {3},
  pages = {62--69},
  doi = {10.1145/2812803}
}

@article{krishnamurthi2015real,
  author = {Krishnamurthi, Shriram and Vitek, Jan},
  title = {{The real software crisis}},
  journal = {Communications of the ACM},
  year = {2015},
  volume = {58},
  number = {3},
  pages = {34--36},
  doi = {10.1145/2658987}
}

@article{baker2016reproducibility,
  author = {Baker, Monya},
  title = {{1,500 scientists lift the lid on reproducibility}},
  journal = {Nature},
  year = {2016},
  volume = {533},
  number = {7604},
  pages = {452--454},
  doi = {10.1038/533452a}
}

@article{ioannidis2005why,
  author = {Ioannidis, John P. A.},
  title = {{Why Most Published Research Findings Are False}},
  journal = {PLoS Medicine},
  year = {2005},
  volume = {2},
  number = {8},
  pages = {e124},
  doi = {10.1371/journal.pmed.0020124}
}

@article{nosek2018preregistration,
  author  = {Nosek, Brian A. and Ebersole, Charles R. and DeHaven, Alexander C. and Mellor, David T.},
  title   = {The Preregistration Revolution},
  journal = {Proceedings of the National Academy of Sciences},
  year    = {2018},
  volume  = {115},
  number  = {11},
  pages   = {2600--2606},
  doi     = {10.1073/pnas.1708274114}
}

@article{chambers2013registered,
  author = {Chambers, Christopher D.},
  title = {{Registered Reports: A new publishing initiative at Cortex}},
  journal = {Cortex},
  year = {2013},
  volume = {49},
  number = {3},
  pages = {609--610},
  doi = {10.1016/j.cortex.2012.12.016}
}

@article{munafo2017manifesto,
  author  = {Munaf{\`o}, Marcus R. and Nosek, Brian A. and Bishop, Dorothy V. M. and Button, Katherine S. and Chambers, Christopher D. and Percie du Sert, Nathalie and Simonsohn, Uri and Wagenmakers, Eric-Jan and Ware, Jennifer J. and Ioannidis, John P. A.},
  title   = {A Manifesto for Reproducible Science},
  journal = {Nature Human Behaviour},
  year    = {2017},
  volume  = {1},
  number  = {1},
  pages   = {0021},
  doi     = {10.1038/s41562-016-0021}
}

@article{simmons2011false,
  author  = {Simmons, Joseph P. and Nelson, Leif D. and Simonsohn, Uri},
  title   = {False-Positive Psychology: Undisclosed Flexibility in Data Collection and Analysis Allows Presenting Anything as Significant},
  journal = {Psychological Science},
  year    = {2011},
  volume  = {22},
  number  = {11},
  pages   = {1359--1366},
  doi     = {10.1177/0956797611417632}
}

@article{wilson1927probable,
  author = {Wilson, Edwin B.},
  title = {{Probable Inference, the Law of Succession, and Statistical Inference}},
  journal = {Journal of the American Statistical Association},
  year = {1927},
  volume = {22},
  number = {158},
  pages = {209--212},
  doi = {10.1080/01621459.1927.10502953}
}

@article{brown2001interval,
  author  = {Brown, Lawrence D. and Cai, T. Tony and DasGupta, Anirban},
  title   = {Interval Estimation for a Binomial Proportion},
  journal = {Statistical Science},
  year    = {2001},
  volume  = {16},
  number  = {2},
  pages   = {101--133},
  doi     = {10.1214/ss/1009213286}
}

@article{mcnemar1947note,
  author = {McNemar, Quinn},
  title = {{Note on the Sampling Error of the Difference Between Correlated Proportions or Percentages}},
  journal = {Psychometrika},
  year = {1947},
  volume = {12},
  number = {2},
  pages = {153--157},
  doi = {10.1007/bf02295996}
}

@article{jeffreys1946invariant,
  author = {Jeffreys, Harold},
  title = {{An invariant form for the prior probability in estimation problems}},
  journal = {Proceedings of the Royal Society of London. Series A. Mathematical and Physical Sciences},
  year = {1946},
  volume = {186},
  number = {1007},
  pages = {453--461},
  doi = {10.1098/rspa.1946.0056}
}

@article{dunn1961multiple,
  author = {Dunn, Olive Jean},
  title = {{Multiple Comparisons among Means}},
  journal = {Journal of the American Statistical Association},
  year = {1961},
  volume = {56},
  number = {293},
  pages = {52--64},
  doi = {10.1080/01621459.1961.10482090}
}

@article{benjamini1995controlling,
  author = {Benjamini, Yoav and Hochberg, Yosef},
  title = {{Controlling the False Discovery Rate: A Practical and Powerful Approach to Multiple Testing}},
  journal = {Journal of the Royal Statistical Society Series B: Statistical Methodology},
  year = {1995},
  volume = {57},
  number = {1},
  pages = {289--300},
  doi = {10.1111/j.2517-6161.1995.tb02031.x}
}

@inproceedings{akgun2018metamorphic,
  author    = {Akg{\"u}n, {\"O}zg{\"u}r and Gent, Ian P. and Jefferson, Christopher and Miguel, Ian and Nightingale, Peter},
  title     = {Metamorphic Testing of Constraint Solvers},
  booktitle = {Principles and Practice of Constraint Programming -- CP 2018},
  editor    = {Hooker, John},
  series    = {Lecture Notes in Computer Science},
  volume    = {11008},
  year      = {2018},
  pages     = {727--736},
  publisher = {Springer},
  doi       = {10.1007/978-3-319-98334-9_46}
}

@inproceedings{abreu2022metamorphic,
  author = {Abreu, Rui and Fernandes, Jo{\~a}o Paulo and Llana, Luis and Tavares, Guilherme},
  title = {{Metamorphic testing of oracle quantum programs}},
  booktitle = {Proceedings of the 3rd International Workshop on Quantum Software Engineering},
  year = {2022},
  pages = {16--23},
  publisher = {ACM},
  doi = {10.1145/3528230.3529189}
}

@misc{nguyen2026metamorphq,
  author        = {Nguyen, Ngoc Nhi and Le, John and Vu, Thai T. and Nguyen, Thi Thuy Nga and Shen, Jun},
  title         = {{MetaMorphQ: Physics-Based Metamorphic Testing of Variational Quantum Circuits}},
  year          = {2026},
  eprint        = {2606.28742},
  archivePrefix = {arXiv},
  primaryClass  = {quant-ph},
  doi           = {10.48550/arXiv.2606.28742}
}

@inproceedings{barus2011testing,
  author    = {Barus, A. C. and Chen, T. Y. and Grant, D. and Kuo, F.-C. and Lau, M. F.},
  title     = {Testing of Heuristic Methods: A Case Study of Greedy Algorithm},
  booktitle = {Software Engineering Techniques},
  editor    = {Huzar, Zbigniew and Koci, Radek and Meyer, Bertrand and Walter, Bartosz and Zendulka, Jaroslav},
  series    = {Lecture Notes in Computer Science},
  volume    = {4980},
  year      = {2011},
  pages     = {246--260},
  publisher = {Springer},
  doi       = {10.1007/978-3-642-22386-0_19}
}

@inproceedings{maan2023testing,
  author = {Maan, Arshpreet Singh and Paler, Alexandru},
  title = {{Testing the Accuracy of Surface Code Decoders}},
  booktitle = {2023 IEEE International Conference on Rebooting Computing (ICRC)},
  year = {2023},
  pages = {1--5},
  publisher = {IEEE},
  doi = {10.1109/icrc60800.2023.10386986}
}

@misc{shen2026magicaxis,
  author        = {Jiachen Shen and Hui Zhong},
  title         = {{A conditional no-go for resource-free magic-axis measurement on a static surface code}},
  year          = {2026},
  eprint        = {2607.16968},
  archiveprefix = {arXiv},
  verified      = {arXiv abstract page resolved 2026-09-01; title and both authors match the record}
}

@misc{zhong2025uniq,
  author        = {Hui Zhong and Jiachen Shen and Lei Fan and Xinyue Zhang and Hao Wang and Miao Pan and Zhu Han},
  title         = {{UNIQ: Communication-Efficient Distributed Quantum Computing via Unified Nonlinear Integer Programming}},
  year          = {2025},
  eprint        = {2512.00401},
  archiveprefix = {arXiv},
  verified      = {arXiv abstract page resolved 2026-09-01; title and seven-author list match the record}
}

@misc{zhong2024dpdqc,
  author        = {Hui Zhong and Keyi Ju and Jiachen Shen and Xinyue Zhang and Xiaoqi Qin and Tomoaki Ohtsuki and Miao Pan and Zhu Han},
  title         = {{Differential Privacy Preserving Distributed Quantum Computing}},
  year          = {2024},
  eprint        = {2412.12387},
  archiveprefix = {arXiv},
  verified      = {arXiv abstract page resolved 2026-09-01; title and eight-author list match the record}
}

@misc{shen2026gridthermal,
  author        = {Jiachen Shen and Jian Shi and Yijie Yang and Chenye Wu and Dan Wang and Ju Bin Song and Zhu Han},
  title         = {{Grid-Interactive Thermal Management of AI Data Centers via Contextual Distributionally Robust Optimization}},
  year          = {2026},
  eprint        = {2607.00099},
  archiveprefix = {arXiv},
  verified      = {arXiv abstract page resolved 2026-09-01; title and seven-author list match the record}
}

@misc{shen2026cbam,
  author        = {Jiachen Shen and Jian Shi and Dan Wang and Han Zhu},
  title         = {{Will the Carbon Border Adjustment Mechanism Impact European Electricity Prices? A GNN-Based Network Analysis}},
  year          = {2026},
  eprint        = {2605.03304},
  archiveprefix = {arXiv},
  verified      = {arXiv abstract page resolved 2026-09-01; title and four-author list match the record}
}

@article{pontolillo2026qucheck,
  author  = {Pontolillo, Gabriel and Mousavi, Mohammad Reza and Grzesiuk, Marek},
  title   = {{QuCheck: A Property-based Testing Framework for Quantum Programs in Qiskit}},
  journal = {ACM Transactions on Quantum Computing},
  year    = {2026},
  volume  = {7},
  number  = {4},
  articleno = {24},
  numpages = {32},
  publisher = {ACM},
  doi     = {10.1145/3815169},
  verified = {Crossref resolved 2026-09-01; title, three-author list, journal, volume and issue match the record}
}

\appendix
\section{Definitions, collected}
\label{app:definitions}

The main text introduces each definition where it is used. This appendix collects them in one place,
in the form the implementation actually enforces, so that a reader checking a claim against the code
has one reference to work from.

\subsection{The object under test}

\begin{definition}[Detector error model]
\label{app:def:dem}
A \emph{detector error model} $M$ over a detector set $D$ and an observable set $L$ is a finite
sequence of mechanisms $M_0, \dots, M_{m-1}$. Each $M_i$ carries a probability $p_i \in (0, 1/2)$, a
set $\mathrm{det}(M_i) \subseteq D$, and a set $\mathrm{obs}(M_i) \subseteq L$. The \emph{index
space} of $M$ is $\{0, \dots, m-1\}$ in the order the caller supplied, and every returned index in
this paper is in that space.
\end{definition}

\begin{definition}[Weighted decoding graph]
\label{app:def:graph}
$M$ is \emph{graphlike} when $1 \le |\mathrm{det}(M_i)| \le 2$ for every $i$; a mechanism flipping
no detector is outside the fragment of Definition~\ref{def:scope} and is removed before an adapter
sees the model. Its decoding graph is $\graph = (D \cup \{b\}, E, w, \lambda)$, where $b$ is a
virtual boundary vertex, mechanism $i$ contributes an edge on $\mathrm{det}(M_i)$ if that set has two
elements and an edge from its single detector to $b$ otherwise, $w_i = \log\!\big((1 - p_i)/p_i\big)
> 0$ is its log-likelihood-ratio weight, and $\lambda(e_i) \in \mathbb{F}_2^{|L|}$ is the indicator
of $\mathrm{obs}(M_i)$. Edges carry identities, so parallel mechanisms on one endpoint pair remain
distinct. For $S \subseteq E$, $\bd_D S$ is the set of detectors incident to an odd number of edges of
$S$, and $\wt{S} = \sum_{e \in S} w_e$. A correction $\corr$ \emph{explains} $\syn$ when
$\bd_D \corr = \syn$.
\end{definition}

Definition~\ref{def:subject} of the main text fixes what a subject exposes: $\mathrm{dec\_corr}$
returning a set of mechanism indices, $\mathrm{dec\_obs}$ returning a predicted observable vector, and
optionally the total weight of the chosen correction. Definition~\ref{def:scope} fixes that the whole
framework is defined over $M$ and reads no topology, no gate set and no device description.

\subsection{Capabilities and declarations}

\begin{definition}[Capabilities and flags]
\label{app:def:caps}
\textbf{$L_0$, mandatory correction-interface conformance.} A subject exposing $\mathrm{dec\_corr}$
must return a duplicate-free list $\corr$ of in-range mechanism indices satisfying
$\bd_D \corr = \syn$ in the caller's index space. $L_0$ is incurred by exposing
$\mathrm{dec\_corr}$; it is neither declared nor declinable, and it asserts nothing about optimality,
bounded-distance correctness, determinism or maximum-likelihood decoding. Both a malformed return and
a well-formed return that does not explain the syndrome are $L_0$ failures. Every subject exposing
$\mathrm{dec\_corr}$ incurs it; a subject exposing only $\mathrm{dec\_obs}$ has no correction for it
to constrain, which is the case Section~\ref{sec:hunt:chromobius} reports on.

\textbf{The optional capabilities.} $L_1$: bounded-distance correctness, that every error of weight
at most the correction radius is mapped into the identity logical class, in the ordered sense of
Definition~\ref{app:def:bd}. $L_2$: exact minimum-weight optimisation. $L_3$: exact
maximum-likelihood decoding. The flags are $\mathit{det}$, a determinism guarantee; $\mathit{rand}$, a
declared randomised schedule; and $\mathit{int\text{-}weights}$, integer internal weights. All three
capabilities and all three flags carry the declaration states of Definition~\ref{def:declstate}.

Of the optional capabilities, only $L_1$ and $L_2$ are related by a proved statement, and
Theorem~\ref{thm:incomparable} shows they do not form a chain over the general weight domain;
Remark~\ref{app:rem:domain} gives the restricted domain on which one direction holds. No implication
involving $L_3$ is claimed or used.
\end{definition}

Definition~\ref{def:declstate} fixes the three declaration states $\yes$, $\no$ and $\silent$ as facts
about documentation at a pinned version, never inferred from a name, from behaviour, or from the
absence of another claim. Definition~\ref{def:dispatch} fixes what a gate does with each: a verdict
on $\yes$, a legal skip on $\no$, and $\unassessable$ on $\silent$. Definition~\ref{def:layer}
separates a conformance verdict from an instance-measured observation, and
Definition~\ref{def:yield} defines the dispatch yields over cells. The verdict layer covers both a
declared optional capability and a mandatory one, which is why it is not called claim-based.

\begin{remark}[A yield carries its syndrome distribution]
\label{app:rem:yield}
A yield is a function of the population and of the distribution of syndromes over it. Sampling
weight-one physical errors gives verified-unique optima on essentially every instance; sampling
syndromes uniformly over the graph gives unique optima on $100\%$, $66.8\%$ and $30.5\%$ of instances
at $d = 3, 5, 7$ respectively. The yields of Table~\ref{tab:yield} are over the evaluation grid
described there --- 29 instance configurations with up to 30 cases each, 791 cases in all, crossed
with the nine subjects to give 7119 cells --- and not over the exhaustive population of
Table~\ref{tab:instances}, which is where every relation count in this paper is taken instead.

The direct and derived columns share that denominator of 7119, since a gate is a fact about a
subject. The applicability column does not: hypotheses hold or fail on a case regardless of who is
being asked, so its denominator is 791. The three are reported in three columns for that reason.
\end{remark}

\subsection{Presentations and relations}

Definition~\ref{def:presentation} fixes a relabelling as a pair of bijections, one on
$D \cup \{b\}$ fixing $b$ and one on edges, preserving incidence, every edge weight and every
observable label; a presentation is the image of the instance under one. Definition~\ref{def:applicability} fixes the relation dispatched to each subject. It records a
relation over a quantity the subject does not return as \textsc{not-applicable}, not as a pass.

\begin{definition}[Faithful split]
\label{app:def:split}
A realised mechanism split is \emph{faithful} on a syndrome when the minimum-weight solution of the
split model is the image of a minimum-weight solution of the original. It need not be: for
$p = p_1 + p_2 - 2p_1p_2$ with $p_1, p_2 \in (0,1/2)$ both replacements are strictly heavier than the
mechanism they replace, so the argmin can move. Faithfulness is decided per cell, and a cell on which
the split is not faithful is excluded and counted, never scored.
\end{definition}

\begin{definition}[Faithful scaling]
\label{app:def:scale}
A realised weight scaling by $\gamma > 0$ is \emph{faithful} when the mapping from the caller's
weights to the weights the subject receives is exactly $w \mapsto \gamma w$ on every edge. The
probability representation clamps at both ends of the $\gamma$ ladder, and a cell in which the
realised transformation is not a scaling is excluded and counted, never scored.
\end{definition}

\subsection{Properties}

\begin{definition}[Residual, restated]
\label{app:def:resid}
Gate $L_0$. Hypothesis: the subject exposes $\mathrm{dec\_corr}$. Claim: $\bd_D \corr = \syn$ with the
detectors of each mechanism read in the caller's index space. A returned index outside range, or a
repeated index, is a malformed return and is rejected before any weight is computed.
\end{definition}

\begin{definition}[Bounded-distance correctness, restated]
\label{app:def:bd}
Gate $L_1$. Hypotheses: a code-capacity instance with perfect syndrome extraction; weights that order
candidates by Hamming weight; and an injected error of weight at most $t$, with $t$ derived from the
certified distance of the concrete instance. Claim, as one conjunction in this order: the returned
correction explains the syndrome, and the residual $\corr + e$ lies in the identity logical class.
Where the subject also exposes $\mathrm{dec\_obs}$, the claim extends to what that operation returns:
the observable vector must have length $|L(M)|$ and must agree with the logical class the returned
correction induces. The width half is what Definition~\ref{app:def:cert} reports as
$\textsc{a2-shape}$, and it is included here because a subject that returns a correct correction
alongside a truncated observable vector has failed the caller in exactly the way $L_1$ exists to
forbid.
\end{definition}

\begin{definition}[Exact optimality, restated]
\label{app:def:opt}
Gate $L_2$. Hypotheses: the subject exposes $\mathrm{dec\_corr}$, and an independent exact optimum is
computable for the instance. Claim: the returned solution's weight, recomputed under the caller's
weights, equals that optimum.
\end{definition}

\begin{definition}[Matching-core equivariance, restated]
\label{app:def:equiv}
Gate $L_2$. Hypotheses: a permutation $\relab$ of the weighted graph fixing $b$, \emph{exactly}
preserving the solver-visible weight of every mapped edge and intertwining incidence; a
verified-unique minimum-weight solution on the true graph; and a $\mathrm{dec\_corr}$ interface.
Claim: $\mathrm{Dec}(\relab \syn) = \relab(\mathrm{Dec}(\syn))$ as edge sets. Weight preservation must be exact, not approximate. An approximately weight-preserving map does not
carry an argmin without a bound relating the perturbation to the optimality gap. The permutations in
this paper carry each weight to itself. The hypothesis is therefore met identically. The
implementation also asserts agreement to a relative tolerance of $10^{-12}$. It refuses the cell
otherwise. This guards against a construction defect. It does not weaken the hypothesis.
\end{definition}

\begin{definition}[Purity, restated]
\label{app:def:purity}
Gate $\mathit{det}$ for a conformance verdict; a declared $\mathit{rand}$ is a legal skip; silence on
both dispatches the property on the instance-measured layer alone. Hypothesis: a warm-up sequence is
supplied. Claim: decoding $\syn$ on an instance that has already decoded other syndromes returns the
same correction, observable vector and weight as decoding $\syn$ on a fresh instance.
\end{definition}

\begin{definition}[Failure certificate]
\label{app:def:cert}
Failure certificates are residual-first. A returned correction that leaves a residual yields
$\textsc{a2-residual}$ carrying the support of $\bd_D \corr \oplus \syn$. A returned observable vector
of length other than $|L(M)|$ yields $\textsc{a2-shape}$. Otherwise a disagreement in the observable
class yields $\textsc{a2-logical}$ carrying the support of the difference. A reduction of an instance
preserves a failure exactly when it preserves this certificate.
\end{definition}

The ordering in Definition~\ref{app:def:cert} is forced, not stylistic.
Appendix~\ref{app:proofs} gives the reason. A correction that does not explain the syndrome has no
defined logical class. Reporting a logical-class mismatch there names a quantity that does not exist.

\section{Proofs}
\label{app:proofs}

\subsection{The two capabilities are incomparable}

\begin{theorem}[Theorem~\ref{thm:incomparable}, restated]
\label{app:thm:incomparable}
$L_2 \not\Rightarrow L_1$ and $L_1 \not\Rightarrow L_2$, over legal log-likelihood-ratio weights.
\end{theorem}

\begin{proof}[$L_2 \not\Rightarrow L_1$]
Let $\graph$ be the decoding graph of a code of certified distance $d \ge 3$, so the correction
radius is $t = \lfloor (d-1)/2 \rfloor \ge 1$, and let $C = e_0, e_1, \dots, e_{d-1}$ be a simple
cycle of length $d$ whose edge set is a minimum-support logical operator: flipping every edge of $C$
leaves every detector satisfied and flips a declared observable, and no shorter edge set does. Assign $w(e_0) = d$, assign $w(e_j) = 1$ for $1 \le j \le d-1$, and assign every edge
outside $C$ a weight strictly greater than $d - 1$. All weights are positive and each arises from a
probability $p = 1/(1 + e^{w})$ in $(0, 1/2)$, so the assignment is legal.

Inject the single error $e_0$. The syndrome is $\syn = \bd_D\{e_0\}$, the two endpoints of $e_0$. Two
feasible explanations are $S_1 = \{e_0\}$ at weight $d$ and $S_2 = C \setminus \{e_0\}$ at weight
$d - 1$; $S_2$ is feasible because $\bd_D C = \emptyset$ implies
$\bd_D(C \setminus \{e_0\}) = \bd_D\{e_0\} = \syn$. Any explanation using an edge outside $C$ costs
more than $d-1$ on that edge alone. Hence $S_2$ is the unique minimum-weight explanation, and an $L_2$
subject returns it. The residual is $S_2 + \{e_0\} = C$, which is a logical operator, so the returned
correction fails bounded-distance correctness on an error of weight one. Since $L_1$ requires
correction of every error of weight at most $t \ge 1$, the subject satisfies $L_2$ and violates $L_1$.

The construction concerns the two objectives, not any implementation. Minimum
weight orders candidates by $\sum_e \log\frac{1-p_e}{p_e}$ and bounded-distance correctness orders
them by Hamming weight, and those orders disagree as soon as the mechanism probabilities are
non-uniform enough for one unlikely mechanism to outweigh several likely ones.
\end{proof}

\begin{proof}[$L_1 \not\Rightarrow L_2$]
Let $\mathrm{Dec}$ be a decoder that satisfies $L_1$ on a code with a nonempty stabiliser group and
that returns $\emptyset$ on the empty syndrome, which any minimum-weight decoder does. Modify it to
return $\mathrm{Dec}(\syn) \oplus s$ for a fixed nonidentity stabiliser $s$ with
$\bd_D s = \emptyset$ and $\wt{s} > 0$. Feasibility is preserved, since
$\bd_D(\corr \oplus s) = \bd_D \corr$. The logical class of the residual is preserved, since $s$ is a
stabiliser, so bounded-distance correctness still holds and the modified decoder satisfies $L_1$. On
the empty syndrome the unique minimum-weight explanation is $\emptyset$, of weight $0$, while the
modified decoder returns $\emptyset \oplus s = s$, of weight $\wt{s} > 0$. So it violates $L_2$.
\end{proof}

Both directions together give what Definition~\ref{app:def:caps} states: $L_1$ and $L_2$ do not form
a chain under the proposed implication relation. The dispatch rule of Definition~\ref{def:dispatch}
follows, in that a gate on $\ell$ is opened by a declaration of $\ell$ alone.

\begin{remark}[The domain matters, and the restricted domain is the interesting one]
\label{app:rem:domain}
The construction for $L_2 \not\Rightarrow L_1$ uses weights that do not order candidate corrections by
physical Hamming weight. Restrict the domain to a code-capacity instance whose weights do order them
that way, which is the hypothesis Definition~\ref{app:def:bd} imposes anyway, and the two objectives
agree on the candidates that matter: an exact minimum-weight decoder then corrects every error of
weight at most $t$, so $L_2 \Rightarrow L_1$ there. The capabilities are therefore incomparable over
the general weight domain and ordered over that restricted one. The implication never converts into a
declaration --- inferring one declaration from another is the error Definition~\ref{def:declstate}
exists to forbid --- but it does bind a subject that declares $L_2$ to $L_1$ on any individual case
the evaluator has verified to satisfy its hypotheses. That is what the derived assurance yield of
Definition~\ref{def:yield} counts, and it is why the yield table reports it beside the direct one.
\end{remark}

\subsection{The certificate}

\begin{lemma}[Lemma~\ref{lem:certificate}, restated with the transport made explicit]
\label{app:lem:certificate}
Let $\graph = (D \cup \{b\}, E, w, \lambda)$ and let $\relab$ be a relabelling in the sense of
Definition~\ref{def:presentation}: a bijection $\relab_D : D \cup \{b\} \to D' \cup \{b'\}$ with
$\relab_D(b) = b'$, and a bijection $\relab_E : E \to \relab E$ whose endpoint map is $\relab_D$,
with $w(\relab_E e) = w(e)$ and $\lambda(\relab_E e) = \lambda(e)$ for every $e$. Edges carry identities. The map $\relab_E$ acts on edges, not endpoint pairs. Parallel edges stay
distinguishable. Let $S_1 \subseteq E$ satisfy $\bd_D S_1 = \syn$ and let
$S_2' \subseteq \relab E$ satisfy $\bd_D S_2' = \relab_D \syn$. Write $S_2 = \relab_E^{-1}(S_2')$. If
$\wt{S_1} \ne \wt{S_2}$ then the heavier of the two does not attain
$\min\{\wt S : \bd_D S = \syn\}$.
\end{lemma}

\begin{proof}
Incidence is preserved edge by edge, so for any $S \subseteq E$ a detector $v$ is incident to an odd
number of edges of $S$ if and only if $\relab_D v$ is incident to an odd number of edges of
$\relab_E S$; hence $\bd_D(\relab_E S) = \relab_D(\bd_D S)$. Applying this to $S_2$ and using
$\relab_E S_2 = S_2'$ gives
$\relab_D(\bd_D S_2) = \bd_D S_2' = \relab_D \syn$, and $\relab_D$ is injective, so
$\bd_D S_2 = \syn$. Thus $S_1$ and $S_2$ are both feasible for $\syn$ in the canonical presentation.
Since $w(\relab_E e) = w(e)$ for every edge, $\wt{S_2} = \wt{S_2'}$, so the two totals are comparable
under one weight function. Assume without loss of generality $\wt{S_1} > \wt{S_2}$. Then $S_2$ is
feasible for $\syn$ and strictly cheaper, so $\wt{S_1} > \min\{\wt S : \bd_D S = \syn\}$.
\end{proof}

\begin{remark}[What the lemma does not need]
\label{app:rem:notneeded}
The proof uses feasibility and weight preservation and nothing else. It does not solve the instance,
does not assume the minimum is unique, does not consult the true error, and does not compare two
implementations. It is a witness-producing sufficient condition. It fires exactly on the feasible presentation pairs whose
total weights differ, which is why Section~\ref{sec:measurement} reports the certificate count as a
count of witnesses.
\end{remark}

\subsection{Equivariance forces the answer}

\begin{theorem}[Soundness of matching-core equivariance]
\label{app:thm:equiv}
Let $\graph = (D \cup \{b\}, E, w)$ and let $\relab$ satisfy the hypotheses of
Definition~\ref{app:def:equiv}. If $\syn$ has a unique minimum-weight explanation $E^*$, then
$\relab \syn$ has the unique minimum-weight explanation $\relab E^*$, and any subject that returns a
minimiser returns exactly that set.
\end{theorem}

\begin{proof}
As in Lemma~\ref{app:lem:certificate}, $\relab$ induces a bijection on edge sets with
$\bd_D(\relab S) = \relab(\bd_D S)$, so it maps the feasible family $F(\syn) = \{S : \bd_D S = \syn\}$
bijectively onto $F(\relab \syn)$. It preserves $w$ edge by edge, so it preserves $\wt{\cdot}$ on
every set. A weight-preserving bijection between feasible families carries minimisers to minimisers
and preserves the cardinality of the argmin. With $\operatorname{argmin} F(\syn) = \{E^*\}$ we get
$\operatorname{argmin} F(\relab \syn) = \{\relab E^*\}$, and a subject returning a minimiser has only
that element available.
\end{proof}

\begin{corollary}
\label{app:cor:det}
Under the hypotheses of Theorem~\ref{app:thm:equiv} the returned set is forced, so no determinism
declaration is needed for the instance pair. Observable masks play no part in the optimisation, so the
property is a statement about the matching core and is evaluated on $\mathrm{dec\_corr}$.
\end{corollary}

\subsection{Scaling cannot move the answer, and why one subject family moves anyway}

\begin{proposition}[Positive scaling preserves the argmin]
\label{app:prop:scale}
Let $\gamma > 0$ and let $w' = \gamma w$. Then for every syndrome $\syn$,
$\operatorname{argmin}_{S \in F(\syn)} \sum_{e \in S} w'_e
 = \operatorname{argmin}_{S \in F(\syn)} \sum_{e \in S} w_e$, and the two argmins have the same
cardinality.
\end{proposition}

\begin{proof}
$\sum_{e \in S} w'_e = \gamma \sum_{e \in S} w_e$ for every $S$, and $\gamma > 0$, so the objective is
a strictly increasing function of the original objective and the order of feasible sets is unchanged.
\end{proof}

Proposition~\ref{app:prop:scale} is what relation R6 tests, and it is why a subject that changes its
answer under scaling has changed it for a reason outside the optimisation problem.
Section~\ref{sec:measurement:intweights} reports that only the two belief-propagation configurations
changed their answers, and not the subjects that declare integer internal weights. Their
initialisation depends on the channel probabilities the weights encode, so a scaling can move them
even where it provably cannot move the minimum-weight solution.

\subsection{The residual precedes the class}

\begin{theorem}[Residual precedes class]
\label{app:thm:residualfirst}
For a stabiliser code with stabiliser group $\mathcal{S}$ and parity-check map $\mathrm{syn}(\cdot)$,
the residual $r = e + \corr$ lies in the normaliser $N(\mathcal{S})$ if and only if
$\mathrm{syn}(r) = 0$. Consequently, when $\bd_D \corr \ne \syn$ the logical class
$[r] \in N(\mathcal{S})/\mathcal{S}$ is undefined.
\end{theorem}

\begin{proof}
An operator commutes with every stabiliser generator exactly when its syndrome vanishes, which is the
definition of membership in $N(\mathcal{S})$ for a Pauli operator. If $\bd_D \corr \ne \syn$, then
$\mathrm{syn}(r) = \mathrm{syn}(e) + \mathrm{syn}(\corr) \ne 0$, so $r \notin N(\mathcal{S})$ and the
quotient $N(\mathcal{S})/\mathcal{S}$ does not contain it.
\end{proof}

Because of Theorem~\ref{app:thm:residualfirst}, Definition~\ref{app:def:cert} evaluates the residual
condition first. It reports $\textsc{a2-residual}$, not a logical-class mismatch. A checker
that reports a class disagreement when the residual is nonzero is reporting a quantity that does not
exist.

\section{The conditions a firing must satisfy}
\label{app:conditions}

These twelve conditions were fixed before the search of Section~\ref{sec:hunt} was run, in a
pre-registration committed on its own, before the code that acts on it existed. They are reproduced
here in the form they were committed in, because a condition chosen after seeing a result is not a
condition. A weight disagreement may be reported as a conformance violation of a declared exactness
capability only if every one of them holds.

\begin{enumerate}
\item \textbf{The declaration covers this exact API path.} Disqualify if exactness is claimed only for
another solver mode; if correlated, approximate, streaming or truncated modes are enabled; if the
graph or weight domain is outside documented support; or if the claim belongs to a matching core while
the disagreement arises in a separate approximate front end. This bites hardest for a colour-code
decoder that approximates its problem as matching and then calls an exact matcher: a disagreement in
the outer pipeline is not a violation of the matcher's exactness claim.

\item \textbf{Both presentations are the same solver-visible problem.} Verify
$\graph' = \relab(\graph)$ --- vertices, boundary status, incidences, edge identities, final weights
and enabled options --- after any adaptation layer, on the object actually handed to the native
library. Disqualify if relabelling changes simplification, parallel-edge merging, edge deletion,
boundary representation, quantisation, clipping, integer overflow or shortest-path preprocessing.

\item \textbf{No ambiguous parallel edges.} The search uses a simple graph, one edge per endpoint
pair, asserted at construction. Disqualify any firing whose returned representation gives only
endpoint pairs while two parallel edges of different weight or identity exist.

\item \textbf{Edge identities transported exactly.} A recorded bijection $\relab_E : E \to E'$.
Disqualify the case if fault identifiers are shared. Disqualify the case if identifiers are generated
from post-relabelling insertion order without a recorded inverse. Disqualify the case if the adapter
maps a selected path to a parity projection, not the full chosen edge set.

\item \textbf{Both returned solutions are feasible in the canonical graph.} If one is not, the finding
is an interface violation of the residual property and is reported as that, not as this.

\item \textbf{Weights computed independently.} Never from the decoder's reported weight field.
$\wt{S} = \sum_{e \in S} w_e$ from the frozen canonical edge table. Section~\ref{sec:accessmap}
measures that no method in this study validates that field, our own contract among them.

\item \textbf{Exact integer arithmetic for the primary result.} This is the protection against the
objection that a firing is a floating-point tie-break. Weights are non-negative integers with
$\sum_E w_e < 2^{52}$ and all sums accumulated as exact integers; for the solver that takes native
even integers, all reachable sums are below a pre-committed $2^{60}$ and a firing must differ by at
least two native units. A disagreement visible only in floating-point weights is reported separately
as numerical presentation sensitivity and may not be the headline.

\item \textbf{No adapter-dependent graph construction.} Both presentations go through the library's
native graph interface from the same canonical graph. Never one from a detector error model and one
from an edge list, and never let insertion order select a different merge policy. Disqualify if the
firing disappears when reproduced by a standalone script that imports none of our code.

\item \textbf{Fresh decoder objects and fresh processes.} For each paired comparison: a clean process,
build for presentation A, decode once, terminate; then independently for B; then repeat in reverse
order. A result that depends on prior calls is a purity finding, not this one.

\item \textbf{Randomness pinned or excluded.} A reported firing must reproduce under a pinned
execution.

\item \textbf{Survives the inverse transformation.} Re-run the returned canonical set forward and
require $\relab(S_2) = S_2'$ exactly. This catches the case where the apparent decoder defect is an
inverse-permutation defect in the harness.

\item \textbf{A standalone bundle reproduces it}: canonical graph, boundary, exact integer weights,
canonical syndrome, both relabelling maps, both raw returned edge sets, independently computed
residuals and weights, library version and configuration, and a replay command. If it cannot be
reproduced from that bundle without importing our code, there is no adverse verdict.
\end{enumerate}

\begin{table}[t]
\caption{The classification fixed in advance for every outcome a weight disagreement can have. The
first row is the only one that is a conformance verdict.}
\label{tab:classification}
\small
\begin{tabular}{ll}
\toprule
condition & classification \\
\midrule
all twelve pass & conformance violation of the declared capability \\
the graph differs after adaptation or preprocessing & adapter or construction mismatch \\
objective equality holds under native integer values & no violation; weight-computation error \\
only a small floating-point disagreement & numerical presentation sensitivity \\
one output does not explain the syndrome & residual property violation \\
disagreement only after prior calls & purity violation \\
the declaration does not cover the configuration & out of scope \\
only an approximate outer mapping differs & approximate front-end behaviour \\
\bottomrule
\end{tabular}
\end{table}

Table~\ref{tab:classification} is what makes the null of Section~\ref{sec:hunt} readable. A search
whose positive outcome has eight possible names, only one of which is a verdict, cannot be steered
towards a verdict after the fact. Two of the eight names were used during the study: the colour-code
probe closed as out of scope under condition~1, and the instrument-health counters that conditions~2,
5 and~11 require were recorded as clean throughout the four cells reported in
Table~\ref{tab:hunt}.

\section{Per-subject and per-instance results}
\label{app:tables}

Section~\ref{sec:measurement} reports the panel's results by effect. This appendix reports them by
subject and by instance, so that a reader can locate any number in the main text in the cell it came
from.

\subsection{The panel, relation by relation}

\begin{table}[t]
\caption{Changed answers and moved optimality, per subject and per relation, over the exhaustive
population of Table~\ref{tab:instances}. ``chg'' counts syndromes whose returned correction changed
under the relation; ``flip'' counts the subset of those on which the optimality of the answer moved,
and is defined only on the 2867 syndromes with a proved-unique optimum. ``base subopt'' counts
suboptimal answers in the unmodified presentation. RESID counts corrections that failed to explain
their syndrome, and OPT is the conformance layer, dispatched only where $L_2$ is
declared. ``---'' in the declared column means the subject declares nothing optional; $L_0$ is
mandatory and so is never listed as a declaration.}
\label{tab:panelfull}
\small
\begin{tabular}{llrrrrrrr}
\toprule
subject & declared & R1 chg/flip & R2 chg/flip & R5 chg & R6 chg/flip & base subopt & RESID & OPT \\
\midrule
\textsc{pm}     & $L_2$ & 0 / 0 & 0 / 0 & 0 & 0 / 0 & 0 & \textbf{0} & \textbf{0} \\
\textsc{pmc}    & --- & 0 / 0 & 0 / 0 & 0 & 0 / 0 & 0 & \textbf{0} & \textbf{0} \\
\textsc{fb}     & $L_2$ & 0 / 0 & 0 / 0 & 0 & 0 / 0 & 0 & \textbf{0} & \textbf{0} \\
\textsc{tess}   & --- & 0 / 0 & 0 / 0 & 0 & 0 / 0 & 0 & \textbf{0} & --- \\
\textsc{mwpf-j} & --- & 3 / 3 & 3 / 3 & 0 & 0 / 0 & 2 & \textbf{0} & --- \\
\textsc{mwpf-s} & --- & 6 / 6 & 6 / 5 & 1 & 0 / 0 & 3 & \textbf{0} & --- \\
\textsc{mwpf-u} & --- & 43 / 18 & 91 / 26 & 248 & 0 / 0 & 458 & \textbf{0} & --- \\
\textsc{ldpc}   & --- & 10 / 10 & 0 / 0 & 19 & \textbf{1598 / 1590} & 18 & \textbf{0} & --- \\
\textsc{ldpc-r} & $\mathit{rand}$ & 84 / 84 & 0 / 0 & 0 & \textbf{1748 / 1730} & 47 & \textbf{0} & --- \\
\bottomrule
\end{tabular}
\end{table}

Three columns of Table~\ref{tab:panelfull} carry the section's claims. RESID is zero in every cell,
which over all relations and instances is the $717{,}244$ judged results of
Observation~\ref{obs:resid}. OPT is zero wherever it is dispatched. This is a sanity check that the instrument agrees with the
two exactness declarations, not a discovery. An exact optimiser returning the optimum is what the
declaration says. The informative zero in that column's
neighbourhood is \textsc{tess}, which behaves exactly on every syndrome while declaring nothing, and
is the reason an instance-measured result may never be promoted into a capability claim.

The R6 column is the one that overturned a prediction. The four subjects carrying
$\mathit{int\text{-}weights}$ --- one by declaration, three by measurement --- changed nothing under
weight scaling; of the five carrying no such flag, two changed $1598$ and $1748$ answers and three
changed nothing. Proposition~\ref{app:prop:scale} says the argmin cannot move, so
the movement is not in the optimisation: rescaling the weights rescales the channel probabilities that
belief propagation is initialised from.

\subsection{Where the certificate fires, against what the oracle can see}

\begin{table}[t]
\caption{Certificate recall, per subject and per relation, on the syndromes where the comparison is
defined. The denominator is the suboptimal variant answers an exact oracle finds on unique-optimum
syndromes; the numerator is how many of those the weight comparison also witnesses. Both relations are
reported, because either alone would be a selected result.}
\label{tab:recall}
\small
\begin{tabular}{lrrr}
\toprule
subject & R1 & R2 & pooled \\
\midrule
\textsc{mwpf-j} & 1/1  & 1/1  & \textbf{2/2} \\
\textsc{mwpf-s} & 4/5  & 3/4  & \textbf{7/9} \\
\textsc{ldpc-r} & 48/59 & 0/47 & \textbf{48/106} \\
\textsc{ldpc}   & 6/20 & 0/18 & \textbf{6/38} \\
\textsc{mwpf-u} & 2/444 & 24/476 & \textbf{26/920} \\
\bottomrule
\end{tabular}
\end{table}

Pooled recall in Table~\ref{tab:recall} ranges from $3\%$ to $100\%$ by subject, which is the
quantitative form of Observation~\ref{obs:whenfires}: the certificate is a witness-producing lower
bound and not a detector. Suboptimality that both presentations share leaves the totals equal and
produces nothing. In the other direction the certificate reaches where this particular oracle-based
count cannot, since that count is defined against a designated optimal edge set and there is no such
set when two optima have equal weight, while a weight difference between two feasible answers is a
proof either way.

\subsection{Beyond the enumerable population}

\begin{table}[t]
\caption{Certificates on the rotated surface code at code capacity, R1 / R2, over a deterministic
stride-spread selection from the error patterns of weight at most four, capped per instance. That
selection is the benchmark's population and is not the instance's complete syndrome space, which
cannot be enumerated at these distances. No randomness and no seed. No infeasible pair was encountered anywhere, and both
corrections are checked for feasibility before a certificate is issued.}
\label{tab:scale}
\small
\begin{tabular}{lrrrrrrrrrr}
\toprule
$d$ & edges & syndromes & \textsc{pm} & \textsc{fb} & \textsc{tess} & \textsc{mwpf-j} & \textsc{mwpf-s} & \textsc{mwpf-u} & \textsc{ldpc} & \textsc{ldpc-r} \\
\midrule
7  & 43  & 3797 & 0/0 & 0/0 & 0/0 & 0/0 & 5/3 & 0/2 & 80/5 & 101/3 \\
9  & 73  & 4527 & 0/0 & 0/0 & 0/0 & 0/0 & 0/0 & 0/2 & 27/8 & 42/5 \\
11 & 111 & 4604 & 0/0 & 0/0 & 0/0 & 0/0 & 0/0 & 0/1 & 10/0 & 29/1 \\
13 & 157 & 4651 & 0/0 & 0/0 & 0/0 & 0/0 & 0/0 & 0/0 & 4/2 & 17/1 \\
15 & 211 & 4709 & 0/0 & 0/0 & 0/0 & 0/0 & 0/0 & 0/0 & 2/0 & 12/0 \\
\bottomrule
\end{tabular}
\end{table}

Table~\ref{tab:scale} is the same check run where the exhaustive population stops. \textsc{pmc} is
omitted because it computes what \textsc{pm} computes. The two exactness-declaring subjects produce no
certificate at any distance, which is the code-capacity companion to the circuit-level search of
Section~\ref{sec:hunt}. The belief-propagation configurations produce certificates at every distance,
and the count falls with $d$ because the error patterns are of bounded weight while the graph grows.

\subsection{The observable-width hazard}

\begin{table}[t]
\caption{Relations R3 and R3b on the second population, over the declared-observable strata
$k \in \{1, 2, 8, 30, 31, 32, 33, 62, 63, 64, 65\}$. The seven correction-returning subjects do not
compute an observable vector, so both relations are \textsc{not-applicable} to them and are never
totalled with a pass.}
\label{tab:width}
\small
\begin{tabular}{lrr}
\toprule
subject & R3 width violations & R3b offset violations \\
\midrule
\textsc{pm}  & \textbf{391 / 391} & 0 / 391 \\
\textsc{pmc} & \textbf{391 / 391} & 0 / 391 \\
the seven correction-returning subjects & \textsc{n/a} $\times 11$ each & \textsc{n/a} $\times 11$ each \\
\bottomrule
\end{tabular}
\end{table}

Table~\ref{tab:width} is the R3 and R3b result. A model declaring $k+1$ observables produced a width-$k$
vector in every case: the library sizes its
output from the fault identifiers it was given, and our adapter does not widen it. This was
pre-registered as a caller-side interface hazard, reported with its documented mitigation --- calling
the library's own routine to raise the fault-identifier count, verified to change the returned width
from $k$ to $k+1$ --- and it is not a library defect and not a detection. R3b is the complementary
case: when the top observable is carried by a mechanism, the width is right. The hazard is precisely
that a declared but unexercised observable disappears.

\subsection{Transformation faithfulness}

R5 and R6 rebuild the model, not relabel it. The transformation can therefore move the optimum. A
cell on which it does is excluded and counted, not scored. R5 was judged over $38{,}340$
results with $12{,}537$ excluded on that ground, by Definition~\ref{app:def:split}. R6 was judged over
$460{,}080$ results with $153{,}360$ excluded, which is exactly one third: four of the twelve $\gamma$
rungs, at both ends of the ladder, where the probability representation clamps and the transformation
stops being a scaling. R6's exclusions are decided analytically, not by a solver for each scaled graph. A realised scaling
is a positive scaling of the objective. Proposition~\ref{app:prop:scale} therefore settles it. This
made the full population affordable.

\subsection{Purity}

Relation R4 was judged over $102{,}240$ results across three protocols, namely a fresh adapter per
call, reverse order, and immediate repeat, on every instance. All eight dispatched subjects produced
zero differences in every protocol, including zero first-versus-second differences on the
immediate-repeat protocol. No subject on the panel declares $\mathit{det}$, so none of those results
is a conformance verdict; they are instance-measured observations, which is the layer this relation
reaches on a panel that documents no determinism guarantee. \textsc{ldpc-r}, which declares a
randomised schedule, is a legal skip, and was measured anyway as a control over $12{,}780$ results,
where it also showed zero differences. The control supports one claim. This subject's declaration is conservative on these instances. The
control does not show that the gate is load-bearing. This is an observation about the declaration,
not a finding against the library.

\section{Replication over independent relabellings}
\label{app:replication}

Exhausting a syndrome space replicates over syndromes. The remaining free choice the experiment
makes is \emph{which} relabelling to apply. This appendix reports the replication that varies it, with twenty independent relabellings drawn per instance and per relation
and the whole measurement repeated on each.

\subsection{The design}

A draw fixes one relabelling per instance and per relation, and every subject on the panel receives
the same draw. The draws are therefore paired across subjects, which is what licenses the paired tests
below and forbids comparing two subjects by whether their marginal ranges overlap. Seeds are derived
by hashing the study seed together with the instance identifier and the relation name, so a draw is
reproducible and is not a function of process state. The derivation is stated here so that a reader
can regenerate any draw.

\subsection{Per-subject results over the twenty draws}

\begin{table}[t]
\caption{Changed answers on unique-optimum syndromes and certificates produced, over 20 independent
relabelling draws per relation. ``draw 0'' is the draw reported in the main text, included so that a
reader can see it is not an outlier. The four subjects at the top produced no change and no
certificate in any of their 40 draws.}
\label{tab:replicationfull}
\small
\begin{tabular}{llrrrrrrrr}
\toprule
& & \multicolumn{4}{c}{changed on unique optima} & \multicolumn{4}{c}{certificates} \\
\cmidrule(lr){3-6}\cmidrule(lr){7-10}
subject & rel. & draw 0 & min & median & max & draw 0 & min & median & max \\
\midrule
\textsc{pm}      & R1 & 0 & 0 & 0 & 0 & 0 & 0 & 0 & 0 \\
\textsc{pm}      & R2 & 0 & 0 & 0 & 0 & 0 & 0 & 0 & 0 \\
\textsc{pmc}     & R1 & 0 & 0 & 0 & 0 & 0 & 0 & 0 & 0 \\
\textsc{pmc}     & R2 & 0 & 0 & 0 & 0 & 0 & 0 & 0 & 0 \\
\textsc{fb}      & R1 & 0 & 0 & 0 & 0 & 0 & 0 & 0 & 0 \\
\textsc{fb}      & R2 & 0 & 0 & 0 & 0 & 0 & 0 & 0 & 0 \\
\textsc{tess}    & R1 & 0 & 0 & 0 & 0 & 0 & 0 & 0 & 0 \\
\textsc{tess}    & R2 & 0 & 0 & 0 & 0 & 0 & 0 & 0 & 0 \\
\midrule
\textsc{mwpf-j}  & R1 & 4 & 1 & 2.0 & 4 & 11 & 2 & 10.0 & 12 \\
\textsc{mwpf-j}  & R2 & 2 & 1 & 2.0 & 5 & 5 & 5 & 10.0 & 16 \\
\textsc{mwpf-s}  & R1 & 5 & 1 & 3.5 & 6 & 20 & 5 & 19.5 & 30 \\
\textsc{mwpf-s}  & R2 & 5 & 2 & 5.0 & 8 & 25 & 13 & 24.5 & 38 \\
\textsc{mwpf-u}  & R1 & 20 & 3 & 20.0 & 43 & 28 & 14 & 29.0 & 42 \\
\textsc{mwpf-u}  & R2 & 96 & 69 & 104.5 & 160 & 67 & 55 & 72.0 & 108 \\
\midrule
\textsc{ldpc}    & R1 & 7 & 6 & 12.0 & 17 & 175 & 160 & 180.0 & 224 \\
\textsc{ldpc}    & R2 & 0 & 0 & 0.0 & 0 & 3 & 0 & 2.0 & 3 \\
\textsc{ldpc-r}  & R1 & 109 & 52 & 79.0 & 135 & 283 & 154 & 217.5 & 293 \\
\textsc{ldpc-r}  & R2 & 0 & 0 & 0.0 & 0 & 1 & 0 & 1.0 & 1 \\
\bottomrule
\end{tabular}
\end{table}

Two features of Table~\ref{tab:replicationfull} are worth naming. The two belief-propagation
configurations produce many certificates under R1 and almost none under R2, while the union-find
solver produces them under both: the relations ask different questions, and a subject can be sensitive
to mechanism ordering without being sensitive to detector numbering. And the certificate count is not
a monotone function of the changed-answer count, because a change between two optima of equal weight
produces no certificate.

\subsection{The paired tests}

\begin{table}[t]
\caption{Paired sign tests over the 20 draws, on the count of changed answers at unique optima. Each
row compares two subjects that saw the same relabelling on every draw. $p$ is two-sided exact.}
\label{tab:signtests}
\small
\begin{tabular}{llrrrr}
\toprule
comparison & rel. & first greater & ties & second greater & $p$ \\
\midrule
\textsc{mwpf-j} vs \textsc{mwpf-s} & R1 & 0 & 4 & 16 & $3.05 \times 10^{-5}$ \\
\textsc{mwpf-j} vs \textsc{mwpf-s} & R2 & 0 & 0 & 20 & $1.91 \times 10^{-6}$ \\
\textsc{mwpf-j} vs \textsc{mwpf-u} & R1 & 0 & 1 & 19 & $3.81 \times 10^{-6}$ \\
\textsc{mwpf-j} vs \textsc{mwpf-u} & R2 & 0 & 0 & 20 & $1.91 \times 10^{-6}$ \\
\textsc{mwpf-s} vs \textsc{mwpf-u} & R1 & 2 & 0 & 18 & $4.02 \times 10^{-4}$ \\
\textsc{mwpf-s} vs \textsc{mwpf-u} & R2 & 0 & 0 & 20 & $1.91 \times 10^{-6}$ \\
\bottomrule
\end{tabular}
\end{table}

The six tests of Table~\ref{tab:signtests} order the three configurations of one library:
joint-single-hair changes fewest answers, single-hair more, union-find most, under both relations. The
checker performs $111$ paired comparisons in total, so a Bonferroni threshold at $\alpha = 0.05$ is
$4.50 \times 10^{-4}$, and all six survive it against a largest observed value of
$4.02 \times 10^{-4}$. We report the correction, not the six tests alone. The six were selected after the comparisons were
computed. A threshold that accounts for the whole family is the only honest way to present a selected
subset.

\subsection{What the four zeros mean}

The subjects at the top of Table~\ref{tab:replicationfull} produced no changed answer on a
unique-optimum syndrome and no certificate in any draw. They did change the correction they returned
on syndromes with two optima of equal weight, which is a freedom the problem grants. Across their
forty draws each there were $12{,}835$ such changes for \textsc{pm}, the same for \textsc{pmc},
$17{,}017$ for \textsc{fb} and $20{,}253$ for \textsc{tess}. The invariant that holds is about unique
optima, and those counts come from syndromes where it does not apply. A decoder that returns the
unique optimum whenever one exists is presentation-invariant on exactly those syndromes, by
definition, and free elsewhere.

That one-line argument also explains the study's central failed prediction. The pre-registration
predicted that the first two relations would hold for all subjects. It held for the four that returned
the unique optimum on every such syndrome, for the reason just given, and failed for the five that did
not. The contrapositive of that reason is what the certificate of Lemma~\ref{lem:certificate}
exploits.

\section{The full detection matrices}
\label{app:detection}

Section~\ref{sec:accessmap} groups the thirteen operators by the component they perturb. This appendix
gives the matrices operator by operator, for both populations, so that the grouping can be checked and
so that a reader designing a different operator mix can see what each one contributed.

\subsection{Primary population}

\begin{table}[t]
\caption{Detection per operator on the primary population. 387 distinct behaviours after
deduplication across operators within each instance, with 48 cross-operator collisions removed and
counted; 6 equivalent and 381 live. Equivalence is judged over the full domain the detectors exercise,
so no behaviour is marked equivalent and then detected by anything.}
\label{tab:detectionprimary}
\footnotesize
\setlength{\tabcolsep}{4pt}
\begin{tabular}{llrrrrrrrr}
\toprule
& & & \multicolumn{2}{c}{error rate} & & & & & \\
\cmidrule(lr){4-5}
operator & perturbs & live & obs. & corr. & width & obs.\ chk & corr.\ chk & contract & weight \\
\midrule
O1 permute indices  & correction & 72 & 0 & 66 & 0 & 0 & \textbf{72} & \textbf{68} & 24 \\
O2 shift indices    & correction & 50 & 0 & 50 & 0 & 0 & \textbf{50} & \textbf{50} & 18 \\
O3 drop an edge     & correction & 19 & 0 & 19 & 0 & 0 & \textbf{19} & \textbf{19} & \textbf{19} \\
O4 add an edge      & correction & 59 & 0 & 22 & 0 & 0 & \textbf{59} & \textbf{59} & \textbf{59} \\
O5 duplicate an edge& correction & 19 & 0 & 19 & 0 & 0 & \textbf{19} & \textbf{19} & \textbf{19} \\
O6 sort by weight   & correction & 0  & \multicolumn{7}{l}{all six behaviours equivalent} \\
\midrule
O7 truncate obs.    & observables & 6  & 2 & 0 & \textbf{6} & \textbf{6} & 0 & 0 & 0 \\
O8 pad obs.         & observables & 17 & 0 & 0 & \textbf{17} & \textbf{17} & 0 & 0 & 0 \\
O9 rotate obs.      & observables & 6  & 0 & 0 & 0 & 0 & 0 & \textbf{6} & 0 \\
O10 flip an obs. bit& observables & 12 & \textbf{12} & 0 & 0 & \textbf{12} & 0 & 6 & 0 \\
\midrule
O11 stale answer    & call history & 41 & \textbf{41} & \textbf{41} & 0 & 39 & \textbf{41} & \textbf{41} & 0 \\
O12 round weights in& model in     & 8  & 0 & 0 & 0 & 0 & 2 & 0 & 0 \\
\midrule
\textbf{O13 perturb the returned weight} & \textbf{weight} & \textbf{72} & \textbf{0} & \textbf{0} & \textbf{0} & \textbf{0} & \textbf{0} & \textbf{0} & \textbf{72} \\
\midrule
pooled & & 381 & 55 & 217 & 23 & 74 & 262 & 268 & 211 \\
\% & & & 14.4 & 57.0 & 6.0 & 19.4 & 68.8 & 70.3 & 55.4 \\
Wilson 95\% & & & [11.3, & [51.9, & [4.1, & [15.8, & [63.9, & [65.6, & [50.4, \\
 & & & 18.3] & 61.8] & 8.9] & 23.7] & 73.2] & 74.7] & 60.3] \\
\bottomrule
\end{tabular}
\end{table}

Three rows in Table~\ref{tab:detectionprimary} are worth reading beyond the grouping in the main text.
O6, which sorts the edges of the returned correction by weight, produced six behaviours. All six
were equivalent. Sorting a set that the caller reads as a set changes nothing observable. Each
behaviour is therefore reported as equivalent, not as undetected. O10 flips an observable bit. It separates the two observable-based arms from the width check. A
flipped bit has the wrong value, not the wrong length. O12, which rounds the model weights before decoding, is detected twice by
the correction check and never by the contract, because rounding is applied consistently and a
relation compares a decoder with itself.

The correction-native error-rate column varies across the correction operators. The pooled figure
hides this variation. The variation is the mechanism, not noise. O2, O3 and O5 move the logical
class on every behaviour they produce: shifting every index, dropping an edge and duplicating an edge
each change the parity the caller computes. O1 moves it on $66$ of $72$, since a permutation of the
returned indices can happen to land on a set of the same parity. O4 moves it on only $22$ of $59$,
because an added edge that touches no logical operator leaves the class alone while still handing the
caller a correction it should not apply. A benchmark on the correction-native path therefore sees most of this class. It misses the quietest
part. The $43$ misses are concentrated in the operator that adds. They are not concentrated in the
operator that removes.

The weight column repays reading beyond its own row. It fires on $139$ of the $219$ correction
behaviours as well, because a correction whose edges have been permuted, dropped, added or duplicated
no longer carries the weight the decoder computed for the answer it actually chose: the returned
object has become internally inconsistent, and the check notices. It fires on none of the observable
or call-history classes, which is right, since those leave the correction and its weight agreeing with
each other.

The library's own test suite is recorded as \textsc{not-applicable} throughout. These behaviours live
in a wrapper, not in the library. That decision was made before the study ran, not after the results
were seen. It is never counted as a miss.

The weight arm's tolerance is a measured quantity. An
unmutated pymatching returns a weight agreeing with the caller's own summation only to about
$2.2 \times 10^{-8}$ relative on the rotated surface instances, so at a $10^{-9}$ tolerance the check
fired on the \emph{baseline} for two of six instances, and the differential rule then suppressed the
arm across every mutant of those instances, reporting $48$ of $72$ instead of $72$. The tolerance is
$10^{-6}$, three orders above that measured noise and five below the smallest perturbation O13
applies, which multiplies the weight by a factor drawn from $[1.5, 3.0]$.

\subsection{Secondary population, never pooled with the primary}

The second population declares more than one logical observable, at strata $k \in \{2, 8, 32\}$. It
has its own denominator and is never pooled with the primary; the analysis script refuses to pool
them. It contributes $225$ distinct behaviours, $50$ cross-operator collisions removed, $3$ equivalent
and $222$ live, again with no behaviour marked equivalent and then detected.
Table~\ref{tab:detectionsecondary} pools it.

\begin{table}[t]
\caption{Pooled detection on the secondary population, with its own denominator.}
\label{tab:detectionsecondary}
\small
\begin{tabular}{lrr}
\toprule
method & score & Wilson 95\% \\
\midrule
error rate, observable-native & 57/222 = 25.7\% & [20.4, 31.8] \\
error rate, correction-native & 107/222 = 48.2\% & [41.7, 54.7] \\
width              & 31/222 = 14.0\% & [10.0, 19.1] \\
observable check   & 84/222 = 37.8\% & [31.7, 44.4] \\
correction check   & 106/222 = 47.7\% & [41.3, 54.3] \\
contract           & 146/222 = 65.8\% & [59.3, 71.7] \\
weight check       & 82/222 = 36.9\% & [30.9, 43.5] \\
\bottomrule
\end{tabular}
\end{table}

The two separations reappear on this population in the form the main text describes. On the
correction class the two observable-based arms detect $0$ of $98$, the correction check detects $98$
of $98$ and the contract $95$; the three it misses are again the ones R5 detected alone. The
correction-native error rate detects $95$ of $98$ here against $176$ of $219$ on the primary
population, and the reason for the difference is the number of observables: with two, eight or
thirty-two logical operators in the model, a perturbation of the correction has more parities to
disturb and far fewer ways to leave the class alone. The blind fraction of the correction class falls
from $20\%$ at one observable to $3\%$ above it. On the observable class the contract recovers part of
what it misses at one observable, $14$ of $31$, through the offset relation. And the weight class is
again detected by none of the established arms on either return path, $0$ of $36$, and by all $36$ of
the weight check.

\subsection{Statistical treatment}

The error rate runs one exact McNemar test per behaviour per return path on paired shots, so $1206$
tests over the $603$ behaviours of the two populations, at a nominal $\alpha = 0.01$. Bonferroni over
that whole family gives a threshold of $0.01/1206 = 8.29 \times 10^{-6}$.

Every one of the $112$ observable-native detections survives it, with a largest credited $p$ of
$1.86 \times 10^{-9}$. Of the $324$ correction-native detections, $321$ survive and \emph{three do
not}: two behaviours on instance I2 at $p = 1.95 \times 10^{-3}$ and one on I5 at
$1.50 \times 10^{-3}$, all three from O4, the operator that adds an edge. They are credited in the
tables at the nominal $\alpha$, and dropping them would move the correction-native pooled score from
$57.0\%$ to $56.2\%$ on the primary population and the correction class from $176$ to $173$ of $219$.
That the exceptions are all O4 is the same fact the operator discussion above reports from the other
side: an added edge that misses the logical operators produces the weakest signal this arm can carry,
and at nine discordant shots out of twenty thousand it is the one place where the paired test is
running near its resolution.

All interval estimates use Wilson intervals, not normal approximations. This choice matters at the
small cell counts in Table~\ref{tab:detectionprimary}.

\subsection{The falsifier that fired}

Section~\ref{sec:access:design} states what fired and why. The numbers are these. The first run of
the correction check gave $215$ of $219$ against a pre-registered prediction of all of O1--O5; the
arm had been given only the base presentation, while the property it implements is defined over every
presentation. After the repair the protocol prescribes, the re-run gave $219$ of $219$. Every number
in this appendix and in Section~\ref{sec:accessmap} comes from the repaired arm. Of the five
pre-registered falsifiers, this one fired and the other four did not.

\section{The witness bundle and its verifier}
\label{app:schema}

\subsection{Schema}

A bundle is one JSON object and Table~\ref{tab:schema} is the whole of it. Everything the verifier
needs is present, and nothing it needs is fetched.

\begin{table}[t]
\caption{The bundle schema. Every field is required except \texttt{variant\_edges\_given}, whose
absence causes the corresponding check to be skipped and reported as skipped.}
\label{tab:schema}
\small
\begin{tabular}{lp{0.62\linewidth}}
\toprule
field & content \\
\midrule
\texttt{schema} & the schema identifier the verifier expects \\
\texttt{kind} & the claim the bundle makes; here, a suboptimality witness \\
\texttt{origin} & study, instance identifier, code family, distance, relation, seed \\
\texttt{subject} & library name, pinned version, and whether it declares exactness \\
\texttt{graph} & \texttt{num\_detectors}; \texttt{edges} as endpoint pairs, with $-1$ for the virtual
boundary; \texttt{weights} as a representation tag and a list of strings \\
\texttt{variant\_edges\_given} & the edge list actually handed to the library for the second call \\
\texttt{relabelling} & \texttt{vertex\_permutation} and \texttt{edge\_permutation}, both explicit \\
\texttt{syndrome} & the detectors that fired, as indices into the canonical graph \\
\texttt{output\_A\_edge\_ids} & the first returned correction, in canonical indices \\
\texttt{output\_B\_edge\_ids\_variant\_space} & the second returned correction, in the indices it was
returned in \\
\texttt{expected} & the two claimed weights as strings, and which side is claimed heavier \\
\bottomrule
\end{tabular}
\end{table}

Weights are carried as strings and never as JSON numbers. The representation tag is either
\texttt{int}, in which case each value is an exact decimal integer, or \texttt{hex}, in which case
each value is an IEEE-754 double in round-trip hexadecimal form. Both survive serialisation exactly.
Totals over hexadecimal weights are accumulated with compensated summation, which matters because the
claim is that two totals differ.

\subsection{The eleven checks}

The verifier is a standard-library program that imports nothing from this project and nothing from any
decoder library. It runs these checks in order and stops at the first failure.

\begin{enumerate}
\item \textbf{schema} --- the identifier is the expected one.
\item \textbf{graph well formed} --- the detector count is a positive integer, the edge list is
non-empty, every edge is a pair of integers, the first endpoint is a detector in range, the second is
a detector in range or the boundary, and there is one weight per edge.
\item \textbf{relabelling is a bijection} --- both permutations are permutations, of the detector set
and of the edge set respectively.
\item \textbf{variant graph is the image} --- for every edge, the variant edge at the permuted position
has the incidence the relabelling predicts, with boundary status handled explicitly on both sides.
\item \textbf{syndrome in range} --- every detector named exists in the graph.
\item \textbf{outputs well formed} --- both returned corrections contain in-range edge identifiers and
no duplicates.
\item \textbf{output B transported} --- the second correction is carried back through the edge
permutation into canonical identifiers.
\item \textbf{both outputs explain the syndrome} --- the detector boundary of each transported
correction equals the syndrome exactly.
\item \textbf{weights match the claim} --- each claimed total is recomputed from the graph and
compared; exactly for integers, and within $10^{-9}$ absolute for doubles.
\item \textbf{weights differ} --- the two recomputed totals are different. A bundle whose totals agree
is not a witness and is rejected here.
\item \textbf{heavier side agrees} --- the side the bundle names as heavier is the heavier one.
\end{enumerate}

Check~10 is where a negative control is caught, and check~11 is the conclusion of
Lemma~\ref{lem:certificate} recomputed from the bundle. A reader who runs the verifier has taken our
word for nothing: not the arithmetic, not the transport, and not the feasibility of either correction.

\subsection{Result and controls}

All $639$ certificates produced by the study were exported and re-verified: $639$ passed and $0$
failed. A positive control passes all eleven checks. A negative control, built by taking a real
witness and setting the two weights equal, is rejected at check~10 with the reason that the two
corrections have the same weight and the bundle is therefore not a suboptimality witness.

What this establishes is bounded, and the bound is worth stating. It says each certificate is
internally checkable without our tooling. It says nothing about whether the population of certificates
was well chosen, which is a question about Section~\ref{sec:relations} and not about the verifier.

\section{Reproduction}
\label{app:reproduction}

\subsection{Environment}

Every measurement in this paper was made on one workstation, with no quantum hardware and no
simulator of a quantum device beyond stabiliser simulation. The pinned environment is recorded inside each artifact. It is not confined to the prose. The
workstation runs CPython $3.11.4$ on Windows with Stim $1.16.0$. It includes PyMatching $2.4.0$ and
fusion-blossom $0.2.13$. The \texttt{ldpc} package is version $2.4.1$, and NumPy is version $2.2.6$.
SciPy is version $1.17.1$, and NetworkX is version $3.6.1$. Two subjects run in a separate environment for packaging reasons. That
environment lacks the oracle's numerical dependency, so the tesseract subject appears in the
correctness tables and not in the cost tables.

\subsection{Freezing, and what a freeze is for}

The evaluation pipeline is frozen in a manifest. For every file that participates in a verdict, the
manifest records its path, its role, its size and its SHA-256. It also records the interpreter, the
platform, every library version, the git commit, whether the working tree was dirty at the time, and
a combined digest over the whole set. Each freeze names the freeze it supersedes, and the manifest carries the admission rule that
governs which defects may enter a prospective denominator, so that the rule cannot be restated after
a result is seen.

A freeze is not a reproducibility convenience. It makes ``this analysis was fixed before that data existed'' checkable by someone else. The
discipline works only if the freeze is verified, not asserted. The combined digest validates each freeze against its current manifest and its recorded supersession
chain.

\subsection{What is in the artifact}

\begin{itemize}
\item The harness --- code construction, the matching-graph reduction, the relation implementations,
the exact oracle, and the adapters for every subject on the panel.
\item The drivers, one script per study, each writing a JSON artifact with its own provenance header
--- the protocol commit, the implementation commit, whether the tree was dirty at the start, and the
head commit at the moment of writing.
\item The commit identifier of the protocol each study ran under, recorded in the artifact that
study produced. The confirmatory protocols were committed alone and before the code that acts on
them existed. The single exploratory protocol is marked exploratory because it was not.
\item The frozen pipeline manifests, in supersession order.
\item The errata file, recording each correction made to an instrument during this work and the
result it affects.
\item The $639$ witness bundles and the standalone verifier of Appendix~\ref{app:schema}, which is
also published as an installable package.
\item The claim inventory, which records for each claim the permitted wording, the forbidden
formulation, and the artifact it is evidenced by. Every number in this paper was read from an artifact
named there at the moment of writing.
\end{itemize}

\subsection{Running the studies}

Each study is one command and writes one artifact. The cross-library measurement collects raw returns
first and judges them in a second process that imports no decoder library, so the oracle half of the
study can be run in an environment where no subject is installed --- which is also how we know the
judging code cannot consult a subject. The mutation study writes its own lock file and refuses to
start if another run holds it, which serialises runs that would otherwise race for one output file.

The independent verifier runs on a single bundle or on a directory of them, exits zero only if every
bundle passes every check, and prints a per-check reason otherwise. It is the one component a reader
can run without installing any decoder library at all.

\subsection{Numbers, and where each one came from}

The claim inventory names, for every table in this paper, the artifact it was read from. A script re-reads those artifacts and diffs the values against the manuscript source. It catches any
number that drifts from its artifact mechanically, not through proofreading. The bibliography follows the same discipline. Each entry is resolved against a registry by digital
object identifier or preprint identifier. Its title is compared against the resolved record. An entry
that cannot be resolved is dropped, not guessed at. Two entries in this paper are standards documents whose registries refuse programmatic
requests; both are marked as manually verified, with the clauses quoted by an independent reader
recorded alongside them.

\section{The instance population in full}
\label{app:population}

\subsection{How the instances are built}

Every instance in Table~\ref{tab:instances} is a stabiliser memory experiment compiled to a detector
error model, at code capacity with $p = 0.01$ and perfect syndrome extraction. For the particular
bounded-distance property evaluated here, code capacity is not a convenience but a requirement, since
that property is stated in terms of physical error weight and needs syndrome extraction to be
noiseless. Other bounded-distance guarantees exist against circuit faults and are stated over
spacetime fault distance; they are outside the property this paper checks.

Two further hypotheses are not met everywhere in the table, and both matter. Bounded-distance
correctness is evaluated against an injected error of weight at most $t$, so it speaks only about the
cells that carry such an error, not about every enumerated syndrome. And the repetition stabiliser
code has full-Pauli distance $1$, which under the frozen precondition gives $t = 0$, so no repetition
instance is dispatched to it at all. The familiar distance-$d$ guarantee for repetition codes is
sector-restricted --- it holds inside the protected Pauli sector --- and the frozen precondition is
stated over full-Pauli distance and does not restrict to that sector. That is a finding about our precondition, not about repetition decoding. Repetition decoding does
perform bounded-distance correction in its sector.

The ground-truth check matrix is built by single-qubit injection, not read off the detector error
model. The reason is a real hazard, not a stylistic preference. A model generator merges
error mechanisms that share a detector-and-observable signature, and on a rotated surface code two
distinct single-qubit errors do share one whenever they differ by a weight-two boundary stabiliser.
That merging is correct as a model and it destroys the correspondence between a column and a physical
qubit, which is the correspondence a bounded-distance hypothesis is stated over. So the ground-truth
model keeps one column per physical qubit, while the decoder consumes the merged model, which is what
a decoder sees in production.

\subsection{Sizes, syndromes and unique optima}

Table~\ref{tab:populationfull} gives every instance and its syndrome count.

\begin{table}[t]
\caption{The primary population in full. Every syndrome of every instance is tested, so there is no
sample and no sampling error. ``unique'' counts syndromes whose minimum-weight solution an exact
solver proves unique; no syndrome was left undecided.}
\label{tab:populationfull}
\small
\begin{tabular}{llrrrrr}
\toprule
id & family & $d$ & detectors & edges & syndromes & unique optima \\
\midrule
I1 & repetition       & 3 & 2  & 3  & 4    & 4 \\
I2 & repetition       & 5 & 4  & 5  & 16   & 16 \\
I3 & repetition       & 7 & 6  & 7  & 64   & 64 \\
I4 & rotated surface  & 3 & 4  & 7  & 16   & 16 \\
I5 & unrotated surface& 3 & 6  & 13 & 64   & 31 \\
I6 & rotated surface  & 5 & 12 & 21 & 4096 & 2736 \\
\midrule
total & & & & & \textbf{4260} & \textbf{2867} \\
\bottomrule
\end{tabular}
\end{table}

Every invariant that compares a chosen correction against a designated edge set is counted over the
$2867$ syndromes with a unique optimum, because two distinct optima of equal weight leave a decoder
free to return either. Every invariant that only requires feasibility --- the residual property, and
the certificate of Lemma~\ref{lem:certificate} --- is counted over all $4260$. The distinction is
about which comparison is well posed, and it is worth stating precisely, because optimality itself
remains well posed everywhere: the returned weight can always be compared with the optimum value. What
uniqueness buys is the comparison with a particular \emph{edge set}, which is what the equivariance
invariant needs. Keeping the two denominators apart is also what lets the certificate reach syndromes
the edge-set comparison cannot, since two feasible answers of different weight are a proof whether or
not the optimum is unique.

Every relation runs over the whole population of the table above, with no size cap anywhere in the
harness, which is what makes the counts in Section~\ref{sec:measurement} panel-wide.

\subsection{The second population}

Relations R3 and R3b need a model whose declared observable count is the only quantity that varies, so
they run on a family built for the purpose: mechanism $k$ flips detector $k$ and observable $k$. Every
syndrome there has exactly one explanation, so no tie-breaking is involved and no oracle is needed.
The declared-observable strata are
$k \in \{1, 2, 8, 30, 31, 32, 33, 62, 63, 64, 65\}$, chosen to bracket the $32$-bit and $64$-bit
machine word boundaries on both sides, which is where a width defect would be expected to live if it
were a packing defect. The mutation study of Section~\ref{sec:accessmap} uses the multi-observable
strata $k \in \{2, 8, 32\}$ from the same family as its secondary population, with its own denominator
that is never pooled with the primary one.

\subsection{Circuit level}

The search of Section~\ref{sec:hunt} runs on circuit-derived graphs, where the graphs stop being small
and an exhaustive population is out of reach. Two instances are used, the rotated surface code at $d = 11$
over $11$ rounds, giving $1320$ detectors and $6718$ edges, and at $d = 13$ over $13$ rounds, giving
$2184$ detectors and $11{,}398$ edges. Weights there are replaced by exact even integers drawn from a
broad deterministic spectrum in $[2, 4096]$, with every reachable sum well inside the exactly
representable range, so the comparison happens in integer arithmetic and a firing must differ by at
least two units in the solver's own domain.

The physical weights are replaced. The resulting circuit-derived matching graphs carry artificial
integer weights, not an extracted circuit-level noise model. The search is a solver stress test on
realistic topology, not a noise-model experiment. Syndromes are the
detector boundaries of deterministic edge subsets, stratified over
$|\syn| \in \{8, 16, 32, 64\}$ as the protocol fixed in advance. The lower bound is deliberate: syndromes with very few fired detectors make
exact matching easy and measure almost nothing. Eight hundred syndromes and fifty relabellings per
instance give $40{,}000$ paired comparisons per cell, and four cells give the $160{,}000$ of
Table~\ref{tab:hunt}. Both presentations are built through each library's own graph interface from the
same canonical graph, the graph is asserted simple at construction, and the permuted graph is verified
edge by edge against the image of the canonical one --- and the inverse map applied and compared ---
before either decode call.

\section{The declaration audit}
\label{app:audit}

Table~\ref{tab:declarations} is the paper's main finding, and a finding about documentation has to
ship the documentation. This appendix gives the audit behind it: what was read, at which version,
what it said in its own words, and what was searched for and not found.

\subsection{Method}

Declarations were extracted on 2026-08-28 from the \emph{installed} packages, not from project
websites. A website is not versioned against the artefact a caller depends on. For each
package the source is its \texttt{dist-info/METADATA}, which carries that project's README verbatim,
together with the docstrings of the classes and methods a caller actually calls. Every quotation below
is from those files at the pinned version in Table~\ref{tab:panel}.

A capability was recorded $\yes$ only where a sentence claims it. When no sentence claims a capability, the audit records $\silent$. It also records the search terms.
This shows the reader that the absence was established by a search, not an assumption. A $\no$ requires a sentence declining the capability, and none was
found anywhere on the panel.

\subsection{The complete matrix}
\label{app:audit:matrix}

Table~\ref{tab:auditmatrix} contains every cell. It covers nine subject configurations by six
capability dimensions, totalling $54$ codings. The codings are emitted from the adapters by
\texttt{scripts/r12\_audit\_matrix.py}, not retyped. The published matrix and the yield of
Table~\ref{tab:yield} therefore read the same declarations. Every cell is published, because a list
of only the cells where a sentence was found leaves a reader unable to separate a capability that was
searched for and not found from one that was never searched for.

\begin{table}[t]
\caption{The declaration audit, complete: nine subject configurations by six capability
dimensions. \yes{} means a sentence in the shipped documentation claims the capability;
\silent{} means the terms listed below were searched and no sentence was found; \no{} means a
sentence declines it, and no cell on this panel is \no{}. A dagger marks a flag our own
driver established, not by a sentence. The flag is kept typographically distinct for the
reason Definition~\ref{def:declstate} gives. Four of the $54$ cells carry a declaration.}
\label{tab:auditmatrix}
\small
\setlength{\tabcolsep}{4pt}
\begin{tabular}{L{0.30\linewidth}cccccc}
\toprule
configuration, pinned version & $L_1$ & $L_2$ & $L_3$ & $\mathit{det}$ &
$\mathit{rand}$ & $\mathit{int\text{-}w.}$ \\
\midrule
pymatching 2.4.0 & \silent & \yes & \silent & \silent & \silent & \silent \\
pymatching 2.4.0, correlations on & \silent & \silent & \silent & \silent & \silent & \silent \\
fusion-blossom 0.2.13 & \silent & \yes & \silent & \silent & \silent & \yes \\
ldpc 2.4.1 & \silent & \silent & \silent & \silent & \silent & \silent \\
ldpc 2.4.1, random serial schedule & \silent & \silent & \silent & \silent & \yes & \silent \\
mwpf, joint-single-hair & \silent & \silent & \silent & \silent & \silent & \silent\textsuperscript{$\dagger$} \\
mwpf, single-hair & \silent & \silent & \silent & \silent & \silent & \silent\textsuperscript{$\dagger$} \\
mwpf, union-find & \silent & \silent & \silent & \silent & \silent & \silent\textsuperscript{$\dagger$} \\
tesseract-decoder & \silent & \silent & \silent & \silent & \silent & \silent \\
\bottomrule
\end{tabular}
\end{table}

Every $\silent$ cell was searched. The terms are fixed per dimension, not chosen per library. No cell
is silent because the coder stopped looking early. For $L_1$: \texttt{distance},
\texttt{bounded}, \texttt{correct}, \texttt{guarantee}, \texttt{up to}. For $L_2$: \texttt{exact},
\texttt{optimal}, \texttt{minimum-weight}, \texttt{MWPM}, \texttt{approximation}. For $L_3$:
\texttt{maximum likelihood}, \texttt{degenerate}, \texttt{ML decod}. For $\mathit{det}$:
\texttt{determinis}, \texttt{reproducib}, \texttt{same result}, \texttt{seed}. For $\mathit{rand}$:
\texttt{random}, \texttt{stochastic}, \texttt{schedule}, \texttt{shuffle}. For
$\mathit{int\text{-}weights}$: \texttt{integer}, \texttt{int weight}, \texttt{discret},
\texttt{round}, \texttt{truncat}. The files searched are the ones named in Section~\ref{app:audit},
per subject, and the per-cell record of source file, state, provenance, quotation and search terms is
in \texttt{artifacts/r12\_audit\_matrix.json}.

The matrix carries the paper's finding in its plainest form. Of $54$ capability questions a caller
might put to this panel's documentation, four are answered. Two of those four are the same claim of
exactness, from the two libraries that make it; the other two are flags. The column that matters most
to a caller, $L_1$, is empty across all nine rows, and so are $L_3$ and $\mathit{det}$.

\subsection{What each library says}

Table~\ref{tab:audit} gives the evidence behind the matrix's non-empty cells and near misses. The same
artifact emits the table. It is not transcribed, so the search terms shown against each cell are the
per-dimension terms Section~\ref{app:audit:matrix} fixes and cannot drift from them.

\begin{table}[t]
\caption{The evidence behind every cell the matrix does not leave empty, and behind the
near misses. A quotation is verbatim from the pinned package's own shipped documentation,
extracted 2026-08-28, cut at the clause that carries the claim. A \silent{} row appears here
only where the audit found something short of a claim and the coder recorded why it is not
one; every other \silent{} cell rests on the per-dimension search terms above. Emitted from
\texttt{artifacts/r12\_audit\_matrix.json}.}
\label{tab:audit}
\small
\begin{tabular}{L{0.16\linewidth}L{0.09\linewidth}L{0.66\linewidth}}
\toprule
configuration & state & the sentence, or what was found instead \\
\midrule
pymatching 2.4.0 & $L_2$ \yes & ``The new version is also exact --- unlike previous versions of PyMatching, no approximation is made.'' \\
\midrule
pymatching 2.4.0, correlations on & $L_2$ \silent & the mode is documented as running sparse blossom twice, reweighting after the first pass; the exactness sentence is not repeated for it \\
\midrule
fusion-blossom 0.2.13 & $L_2$ \yes & ``Correctness: This is an exact MWPM solver, verified against the Blossom V library with millions of randomized test cases.'' \\
 & $\mathit{int\text{-}weights}$ \yes & ``we can safely use integers to save weights by e.g. multiplying the weights by 1e6 and truncate to nearest integer.'' \\
\midrule
ldpc 2.4.1 & $L_2$ \silent & OSD is documented as a fallback used when BP does not converge; nothing claims exactness, optimality or bounded-distance correctness \\
\midrule
ldpc 2.4.1, random serial schedule & $L_2$ \silent & as LDPC; the randomised schedule option changes nothing about the claim \\
 & $\mathit{rand}$ \yes & ``random\_serial\_schedule ... Whether to use a random serial schedule order.'' \\
\midrule
mwpf, joint-single-hair & $L_2$ \silent & the README states hypergraph parity-factor decoding is NP-hard and declines the exactness conjecture in general, naming no solver \\
 & $\mathit{int\text{-}weights}$ \silent & no sentence claims integer internal weights; the driver integerises before the call \\
\midrule
mwpf, single-hair & $L_2$ \silent & as MWPF-J; the README never mentions this solver class \\
 & $\mathit{int\text{-}weights}$ \silent & as MWPF-J \\
\midrule
mwpf, union-find & $L_2$ \silent & as MWPF-J; the README's general statement about the exactness conjecture names no solver \\
 & $\mathit{int\text{-}weights}$ \silent & as MWPF-J \\
\midrule
tesseract-decoder & $L_2$ \silent & no documented exactness claim was located \\
\bottomrule
\end{tabular}
\end{table}

\subsection{Why a name and a paraphrase are not declarations}

Two codings on this panel turn on the distinction Definition~\ref{def:declstate} draws. A solver class
named for exactness is a name, and the audit reads names as evidence of nothing, so the mwpf family
is $\silent$ at $L_2$. A README sentence that declines an exactness \emph{conjecture} in general,
naming no solver, is not a sentence declining the capability of any particular solver, so
Table~\ref{tab:declarations} reports $\no$ as empty and not as $1/7$. Both are the general rule
applied to concrete text, and both are the reason the rule is stated in the strong form it has.

\subsection{Determinism has no documentary source anywhere on the panel}

No library on the panel documents a determinism guarantee. The audit searched for a determinism guarantee and recorded its absence for every subject. That
absence exposed the two defects reported in Section~\ref{sec:measurement:yield}. A mechanical probe,
not a declaration, populated one flag. Corollary~\ref{app:cor:det} shows that one requirement was
redundant. The probe
is gone and so is the requirement. $\mathit{det}$ is now demanded only by purity. There, the guarantee is the property under test, not a
precondition. The panel's silence means that the property runs on its instance-measured layer alone.

\subsection{What this audit is not}

It is one coder reading shipped documentation once. There is no second independent coder, no blinded
adjudication of disagreements and therefore no inter-rater statistic, and no attempt to sample a
registry of decoder libraries in a way that would support a statement about the field. Those are the
things a claim about the \emph{ecosystem} would require, and their absence is exactly why every claim
in this paper is scoped to the audited panel. What the audit supports is a statement about nine
configurations of five libraries at pinned versions, which is what Table~\ref{tab:declarations} says
and no more.

\end{document}